\pdfoutput=1
\documentclass[12pt,leqno]{amsart}
\usepackage{fullpage}
\usepackage{mathrsfs}
\usepackage{color}
\usepackage{amsthm}
\usepackage{amssymb}
\usepackage{amsmath}
\usepackage{amsxtra}     
\usepackage{graphicx}
\usepackage{verbatim}
\def\Xint#1{\mathchoice
   {\XXint\displaystyle\textstyle{#1}}%
   {\XXint\textstyle\scriptstyle{#1}}%
   {\XXint\scriptstyle\scriptscriptstyle{#1}}%
   {\XXint\scriptscriptstyle\scriptscriptstyle{#1}}%
   \!\int}
\def\XXint#1#2#3{{\setbox0=\hbox{$#1{#2#3}{\int}$}
     \vcenter{\hbox{$#2#3$}}\kern-.5\wd0}}

\def\dashint{\Xint-}

\newcommand{\mat}[1]{\ensuremath{\mathbf{#1}}}

\theoremstyle{plain}
\newtheorem{theorem}{Theorem}
\newtheorem{prop}{Proposition}
\newtheorem{corollary}{Corollary}
\newtheorem{lemma}{Lemma}

\theoremstyle{definition}
\newtheorem{definition}{Definition}

\theoremstyle{remark}
\newtheorem{remark}{Remark}
\numberwithin{theorem}{section}        
\numberwithin{lemma}{section}
\numberwithin{definition}{section}
\numberwithin{remark}{section}
\numberwithin{prop}{section}
\numberwithin{corollary}{section}

\newcommand{\ueps}{u_\epsilon}

\newcommand{\ut}{\tilde{u}_\epsilon}
\newcommand{\Ut}{\tilde{U}_\epsilon}
\newcommand{\lamt}{\tilde{\lambda}}
\newcommand{\gamt}{\tilde{\gamma}}

\title{On the Zero-Dispersion Limit of the Benjamin-Ono Cauchy Problem for Positive Initial Data}
\author{Peter D. Miller \and Zhengjie Xu}
\address{Department of Mathematics\\University of Michigan\\East Hall\\530 Church St.\\ Ann Arbor, MI 48109}
\date{\today}
\begin{document}

\maketitle
\begin{abstract}
  We study the Cauchy initial-value problem for the Benjamin-Ono
  equation in the zero-disperion limit, and we establish the existence
  of this limit in a certain weak sense by developing an appropriate
  analogue of the method invented by Lax and Levermore to analyze the
  corresponding limit for the Korteweg-de Vries equation.
\end{abstract}

\section{Introduction}
\label{sec:intro}
The Benjamin-Ono (BO) equation 
\begin{equation}
\label{BOE} \frac{\partial \ueps}{\partial t}+2\ueps\frac{\partial\ueps}{\partial x}+
\epsilon \mathcal{H}\left[\frac{\partial^2\ueps}{\partial x^2}\right]=0,
\quad x\in\mathbb{R},\quad t>0,
\end{equation}
where $\epsilon>0$ is a constant and $\mathcal{H}$ is the Hilbert
transform operator defined by the Cauchy principal value integral
\begin{equation}
\label{BOE operator H}
\mathcal{H}[f](x):=\frac{1}{\pi}\dashint_\mathbb{R}\frac{f(y)}{y-x}\,dy
\end{equation}
is a model for weakly nonlinear dispersive waves on the interface
between two ideal immiscible fluids, one of which may be considered to
be infinitely deep.  Applications include the modeling of internal
waves in deep water \cite{Benjamin 1967, Davis 1967, Ono 1975, Choi
  1996}, and also the modeling of atmospheric waves like the dramatic
``morning glory'' phenomenon of northeastern Australia
\cite{Porter 2002}.
The relevant Cauchy problem 
is to determine the solution $\ueps(x,t)$ of \eqref{BOE} subject to a suitable 
initial condition $\ueps(x,0)=u_0(x)$ given for all $x\in\mathbb{R}$.

The parameter $\epsilon>0$ is a measure of the relative strength of
the dispersive and nonlinear effects in the system.  In many applications
one thinks of $\epsilon$ as a small parameter in part because numerical
experiments show that in this situation the finite-time formation
of a shock wave (gradient catastrophe) in the formal limiting equation
(obtained simply by setting $\epsilon=0$ in \eqref{BOE}) 
is dispersively regularized by the generation of a
smoothly modulated train of approximately periodic traveling waves, which
correspond to so-called undular bores frequently 
observed in the evolution of physical internal waves.  Snapshots from
the solution of a Cauchy problem for \eqref{BOE} illustrating the
averted shock and onset of an undular bore are shown in 
Figure~\ref{fig:numerics1}.
\begin{figure}[h]
\begin{center}
\includegraphics{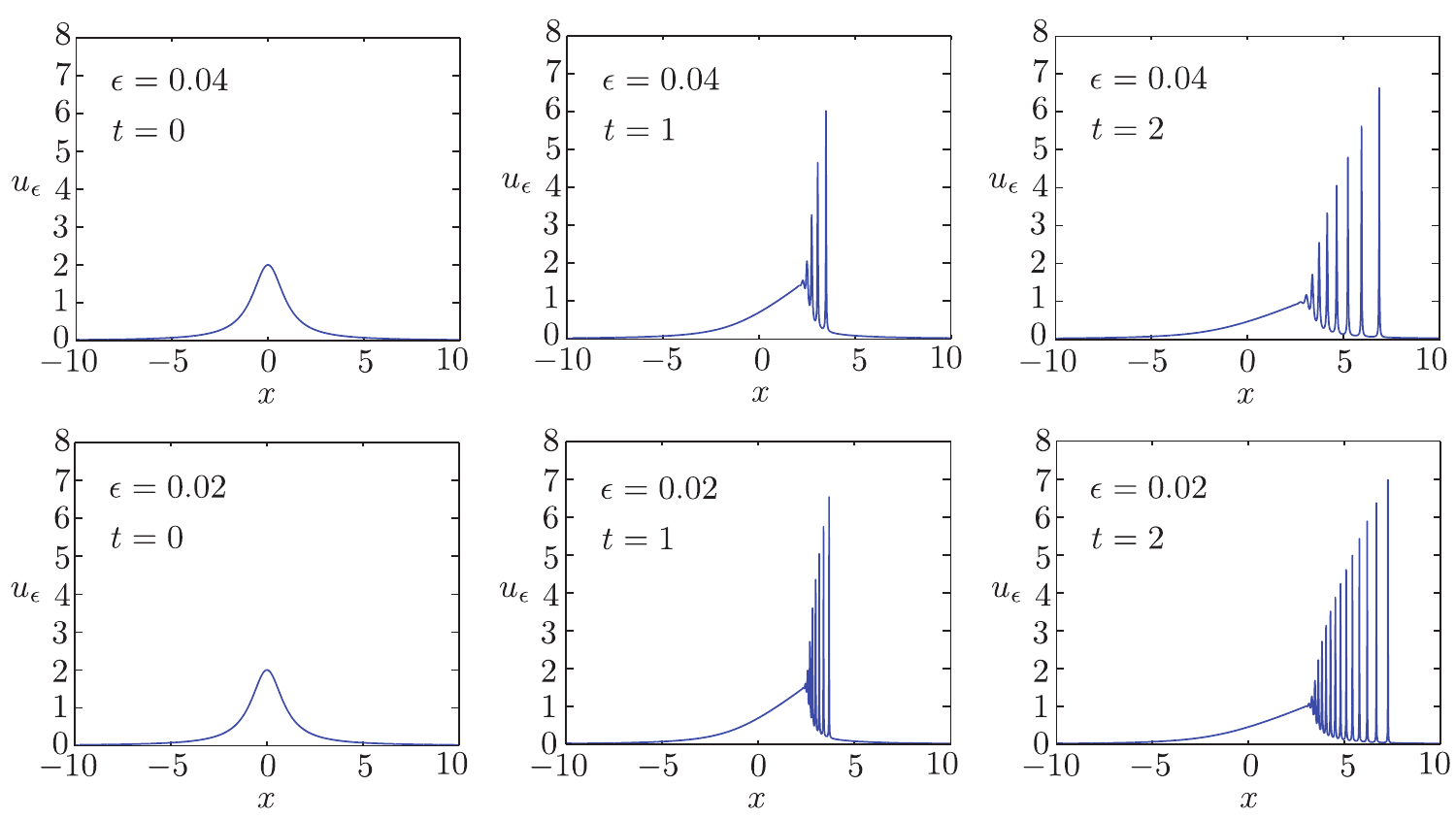}
\end{center}
\caption{\emph{The evolution of a pulse under the BO
    equation.  Top row: $\epsilon=0.04$.  Bottom row: $\epsilon=0.02$.  In
    both cases the initial condition is the same:
    $u_0(x)=2(1+x^2)^{-1}$.}}
\label{fig:numerics1}
\end{figure}
These figures clearly show that the mathematical description of the
undular bore consists of waves of amplitude independent of $\epsilon$
and wavelength approximately proportional to $\epsilon$.  We refer to
the asymptotic analysis of the solution of the Cauchy problem with
$\epsilon$-independent initial data $u_0(x)$ in the limit
$\epsilon\downarrow 0$ as the \emph{zero-dispersion limit}.

\subsection{A related problem and its history}
A more famous nonlinear dispersive wave equation is the Korteweg-de Vries
(KdV) equation
\begin{equation}
\frac{\partial v_\epsilon}{\partial t} + 2v_\epsilon
\frac{\partial v_\epsilon}{\partial x} + 
\frac{1}{3}\epsilon^2 \frac{\partial^3v_\epsilon}{\partial x^3}=0,
\quad x\in\mathbb{R},\quad
t>0,
\label{eq:KdV}
\end{equation}
a model for long surface waves on shallow water among a wide variety
of other physical phenomena.  When $\epsilon>0$ is small, this
equation displays qualitatively similar behavior to that just
illustrated for the BO equation: the dispersive term arrests
the shock in the $\epsilon=0$ equation with the formation of a train
of waves of amplitude approximately independent of $\epsilon$ and
wavelength proportional to $\epsilon$.

The modeling of the zero-dispersion limit for the KdV
equation has a long history going back to the work of Whitham
\cite{Whitham65} who used the method of averaging to propose a
nonlinear hyperbolic system of three partial differential equations to
describe the modulational variables (\emph{e.g.} slowly-varying
amplitude, mean, and wavelength of the wavetrain).  Whitham noted that
the system of modulation equations he obtained had the nongeneric
property that by choice of special dependent variables $v^1$, $v^2$, and
$v^3$, it could be written in so-called Riemann invariant form,
in which the three equations are only coupled through the characteristic
velocities:
\begin{equation}
\frac{\partial v^i}{\partial t} + c^i(v^1,v^2,v^3)
\frac{\partial v^i}{\partial x}=0,\quad i=1,2,3.
\end{equation}
Later, Gurevich and Pitaevskii \cite{Gurevich 1974} considered the problem
of patching together solutions of Whitham's modulational system with
solutions of the formal limiting equation (obtained by setting $\epsilon=0$
in \eqref{eq:KdV}) at two moving
boundary points that delineate the oscillation zone; their goal was to
provide a reasonable global approximation scheme for the solution of 
the initial-value problem for the KdV equation \eqref{eq:KdV}
in the zero-dispersion limit
subject to given initial data $v_\epsilon(x,0)=v_0(x)$ independent of $\epsilon$.

In the meantime, it was of course discovered that the KdV 
equation is a completely integrable system, posessing a compatible structure
now called a \emph{Lax pair} and a coincident solution procedure for addressing
the Cauchy (initial-value) problem: the \emph{inverse-scattering transform}.
This development suggested that the methodology invented by Whitham could
perhaps be placed on completely rigorous mathematical footing.  After
the exact periodic (and quasiperiodic) solutions of the KdV
equation \eqref{eq:KdV} were given a spectral interpretation by 
Its and Matveev \cite{ItsM75} and Dubrovin, Matveev, and Novikov
\cite{DubrovinMN76}, the Whitham modulation equations themselves were
reinterpreted within the framework of integrability by Flaschka, Forest,
and McLaughlin \cite{FlaschkaFM80}.  (In particular this work made clear the
reason why Whitham's equations could be placed in Riemann invariant form;
it is a consequence of integrability.)  

The task that remained in the use of integrable machinery to study the
zero-dispersion limit of the KdV equation was to
rigorously analyze the Cauchy problem using the inverse-scattering
transform.  The first step in this program was taken by Lax and
Levermore \cite{Lax 1983-1} who considered positive initial data
$v_0(x)$ rapidly decaying to zero for large $|x|$.  They used WKB
methods to argue that the Schr\"odinger operator with potential
$-v_0(x)$ that arises in the scattering theory is approximately
reflectionless in the limit $\epsilon\downarrow 0$.  On an
\emph{ad-hoc} basis they replaced the true scattering data by its WKB
analogue, retaining only contributions from a set of $N(\epsilon)\sim
\epsilon^{-1}$ discrete eigenvalues which are approximated by a Bohr-Sommerfeld
quantization rule, which amounts to replacing the solution $v_\epsilon(x,t)$
of the Cauchy problem with another solution $\tilde{v}_\epsilon(x,t)$ 
of \eqref{eq:KdV} having $\epsilon$-dependent initial data close to $v_0$.  
In this situation, the inverse-scattering procedure
reduces to finite-dimensional (of dimension $N(\epsilon)$) linear algebra,
and in fact the solution obtained by Cramer's rule can be reduced to the
determinantal formula
\begin{equation}
\tilde{v}_\epsilon(x,t)=
2\epsilon^2\frac{\partial^2}{\partial x^2}\log(\tau(x,t)),\quad
\tau(x,t)=\det(\mathbb{I}+\mat{G}(x,t)),
\end{equation}
where $\mat{G}(x,t)$ is a positive-definite real symmetric matrix of
dimension $N(\epsilon)\times N(\epsilon)$.  Lax and Levermore then
established the existence of the limit of $2\epsilon^2\log(\tau(x,t))$
as $\epsilon\downarrow 0$, uniformly on compact subsets of the
$(x,t)$-plane.  This yields weak convergence of
$\tilde{v}_\epsilon(x,t)$ by differentiation of the limit function with 
respect to $x$.
The Lax-Levermore method is to expand the determinant $\tau(x,t)$ in
principal minors indexed by subsets of the set of eigenvalues; noting
that each term is positive they showed that the sum of terms is
asymptotically dominated by its largest term, and they further
approximated this discrete optimization problem with an
$\epsilon$-independent (limiting) convex variational problem,
explicitly parametrized by $x$ and $t$, for measures.  The weak
zero-dispersion limit of the Cauchy problem for the KdV
equation is therefore encoded implicitly in the solution of this
variational problem.  The Lax-Levermore method reproduces the
specified initial data $v_0(x)$ at $t=0$ as $\epsilon\downarrow 0$,
which establishes validity, in a certain sense, of the WKB-based
spectral approximation procedure in the first step.  

Later, Venakides \cite{Venakides90} was able to extend the method of
Lax and Levermore to higher order, capturing the form of the
oscillations that are averaged out in the weak limit.  This work at
last made clear that the solution of the Cauchy problem for the
KdV equation with smooth, $\epsilon$-independent initial
data $v_0(x)$ really does generate after some fixed breaking time 
a train of high-frequency waves of
exactly the kind originally considered without complete justification
by Whitham.  More recently, the powerful steepest-descent method 
for matrix Riemann-Hilbert problems developed by Deift and Zhou
was used to analyze the zero-dispersion limit for the KdV
equation \cite{Deift 1997}.  This technique is best viewed as a tool for
converting weak asymptotics (the solution of the Lax-Levermore
variational problem) into strong asymptotics (an improvement of the
Venakides asymptotics in which the phase of the waveform is accurate
to very high order).

\subsection{The zero-dispersion limit of the Benjamin-Ono equation}
It turns out that the BO equation \eqref{BOE} is also an
integrable equation, in the sense that it has a representation as the
compatibility condition of an overdetermined Lax pair of linear
problems \cite{Bock 1979}.  
In fact, both BO and KdV equations
may be viewed as limiting cases (as depth of a fluid layer tends to
infinity and zero, respectively) of the so-called intermediate
long-wave equation \cite{KodamaAS82}, itself an integrable system for
arbitrary layer depth.  However, the integrable structure of the BO
equation is markedly different from that of the KdV equation.
In particular, the nonlocality in the equation due to the presence of
the Hilbert transform is mirrored in a certain nonlocality of the 
scattering and inverse-scattering problems.  In place of the spectral
theory of the selfadjoint Schr\"odinger (Sturm-Liouville) differential operator
one has to work with the spectral theory of the nonlocal
operator 
\begin{equation}
\mathcal{L}:=-i\epsilon\frac{\partial}{\partial x} - \mathcal{C}_+
\circ \ueps \circ \mathcal{C}_+,\quad
\mathcal{C}_+[f](x):=\lim_{\delta\downarrow 0}\frac{1}{2\pi i}
\int_\mathbb{R}\frac{f(y)}
{y-x-i\delta}\,dy.
\label{eq:Loperator}
\end{equation}
Here, the operator $\mathcal{C}_+$ is the selfadjoint orthogonal
projection from $L^2(\mathbb{R})$ onto the Hardy space of the upper
half-plane, the Hilbert space on which $\mathcal{L}$ is selfadjoint,
and $\ueps$ denotes the operator of multiplication by $\ueps(\cdot,t)$.

Certainly a key step forward in the theory of the zero-dispersion
limit was taken by Dobrokhotov and Krichever \cite{Dobrokhotov 1991}
who noted that the second (time evolution) equation in the Lax pair
for the BO equation (see \eqref{Laxpair2} below) is simply a
time-dependent Schr\"odinger equation whose potential is a function
with an analytic continuation from the real $x$-axis into the upper
half-plane and were able to adapt a pre-existing construction of
``integrable'' potentials for this equation to the appropriate
Hardy-space setting.  This allowed them to construct, from the Lax
pair, a large family of periodic traveling wave solutions of the BO
equation \eqref{BOE}, along with quasiperiodic generalizations.
Remarkably, unlike the corresponding exact solutions of the KdV
equation \eqref{eq:KdV} which are highly transcendental objects
constructed from Riemann theta functions of hyperelliptic curves of
arbitrary genus, the periodic and quasiperiodic solutions of the BO
equation turn out to be simple rational functions of $P$ exponential
phases $e^{i(k_jx-\omega_jt)/\epsilon}$.  In the same paper,
Dobrokhotov and Krichever also carried out for the BO equation the
analogue of the calculation of Flaschka, Forest, and McLaughlin
\cite{FlaschkaFM80}, deriving by multiphase averaging a system of
equations governing the modulational variables for a slowly-varying
train of $P$-phase waves.  Here we arrive at a second remarkable fact:
not only can the modulation equations be written in Riemann invariant
form, they are completely diagonal:
\begin{equation}
\frac{\partial u^i}{\partial t} + 2u^i\frac{\partial u^i}{\partial x}=0,
\quad i=1,\dots,2P+1
\end{equation}
(the case of $P=1$ corresponds to simple traveling waves).  This again should
be contrasted with the situation for the KdV equation in which
the characteristic velocities not only provide coupling among the fields but
also are transcendental functions of the fields written in terms of ratios of
complete hyperelliptic integrals.

The analogue for the BO equation of the matching procedure developed
by Gurevich and Pitaevskii \cite{Gurevich 1974} to describe the evolution of
a dispersive shock was independently described by Matsuno 
\cite{Matsuno 1998-1,Matsuno 1998-2} and by Jorge, Minzoni, and Smyth 
\cite{Jorge 1999}.  This matching procedure provides a reasonable approach
to the Cauchy problem for the Benjamin-Ono equation \eqref{BOE} with fixed
initial data $\ueps(x,0)=u_0(x)$ when $\epsilon\ll 1$, but it is based on formal
asymptotics.  In \cite{Matsuno 1998-1}, Matsuno writes:
\begin{quote}
From a rigorously mathematical point of view, however, the various results
presented in this paper should be justified on the basis of an exact method
of solution such as [the inverse-scattering transform], or an analog of
the Lax-Levermore theory for the KdV equation.
\end{quote}
It is our intention in this paper to provide exactly such a justification,
by developing a new method that does for the BO Cauchy problem exactly what
the Lax-Levermore method does for the KdV Cauchy problem.

The main result of our analysis is remarkably easy to state, but first we need 
to recall some basic facts concerning the equation obtained from
\eqref{BOE} simply by setting $\epsilon=0$.  Recall
that while for general sufficiently smooth
initial data $u^\mathrm{B}(x,0)=u_0(x)$ 
the inviscid Burgers equation
\begin{equation}
\frac{\partial u^\mathrm{B}}{\partial t}+2u^\mathrm{B}
\frac{\partial u^\mathrm{B}}{\partial x}=0
\label{eq:inviscidBurgers}
\end{equation}
does not have a global solution as a function due to
gradient catastrophe (shock formation) in finite time, it does have a
global solution as a real multi-sheeted surface over the
$(x,t)$-plane; indeed this is the construction of the method of
characteristics.  The sheets of this surface are obtained as the real
solutions of the implicit equation 
\begin{equation}
u^\mathrm{B}=u_0(x-2u^\mathrm{B}t),
\label{eq:implicitBurgersintro}
\end{equation} 
and by implicit differentiation it
is easy to verify that away from singularities 
each sheet of the surface is a function
$u^\mathrm{B}=u^\mathrm{B}(x,t)$ that satisfies
\eqref{eq:inviscidBurgers}.  
A simple consequence of the Implicit Function Theorem is that
for sufficiently small $|t|$ there is a
unique solution of \eqref{eq:implicitBurgersintro}
for all $x\in\mathbb{R}$. 
New sheets of the multivalued solution are born from \emph{breaking points}
in the $(x,t)$-plane 
that are in one-to-one correspondence with generic inflection points $\xi$ of
$u_0$ for which $u_0'(\xi)\neq 0$ but $u_0''(\xi)=0$.  
If $\xi\in\mathbb{R}$ is such a point, then
the corresponding breaking point is given by 
\begin{equation}
(x_\xi,t_\xi):=\left(\xi-\frac{u_0(\xi)}{u_0'(\xi)},-\frac{1}{2u_0'(\xi)}\right).
\end{equation}
Each such breaking point is the location of 
a pitchfork bifurcation for $u^\mathrm{B}$
with respect to $t$ holding $x-2u_0(\xi)t=\xi$ fixed, with two new branches
emerging as $|t|$ increases.
Thus, assuming that $u_0'$ is a bounded function of total integral zero,
the solution of the Cauchy problem for \eqref{eq:inviscidBurgers}
is classical for 
\begin{equation}
T_-:=-\frac{1}{2\max_{x\in\mathbb{R}}u_0'(x)}<t<-\frac{1}{2\min_{x\in\mathbb{R}}u_0'(x)}=:T_+.
\end{equation}
Note that under our assumptions on $u_0'$ we have $T_-<0<T_+$.  Also, 
$T_-$ is the supremum of all $t_\xi<0$ while $T_+$ is the infimum of all 
$t_\xi>0$.  When we consider the Cauchy problem for $t>0$, we will refer
to $T:=T_+$ as the \emph{breaking time}.

For $t/t_\xi> 1$ there are \emph{caustic curves}
$x^-_\xi(t)<x^+_\xi(t)$ with limiting values as $t\to t_\xi$ given 
by $x^-_\xi(t_\xi)=x^+_\xi(t_\xi)=x_\xi$ 
that bound the triply-folded region emerging from $(x_\xi,t_\xi)$.
The caustic curves correspond to double roots of \eqref{eq:implicitBurgersintro},
and crossing one of them at a generic point results in a change in the
number of sheets by exactly two.  Except along the union of the caustic
curves and the breaking points from which they emerge, the number of
solutions of \eqref{eq:implicitBurgersintro} is always odd, and all are
simple roots.
See Figure~\ref{fig:CausticsAnnotated}.
\begin{figure}[h]
\begin{center}
\includegraphics{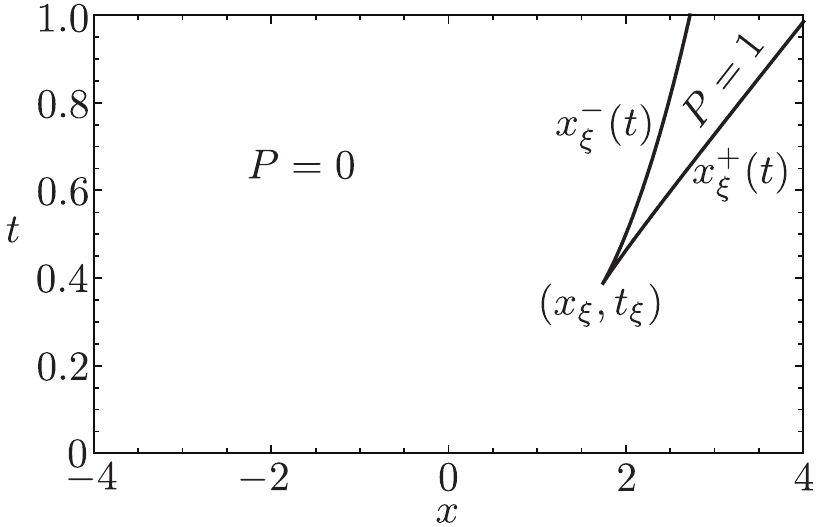}
\end{center}
\caption{\emph{Except along the caustic curves $x=x^-_\xi(t)$ and $x=x^+_\xi(t)$
the number of solutions of \eqref{eq:implicitBurgersintro} is of the
form $2P+1$, and these solutions are simple roots.
For this figure, $u_0(x):=2(1+x^2)^{-1}$.}}
\label{fig:CausticsAnnotated}
\end{figure}

For the initial data
$u_0(x)=2(1+x^2)^{-1}$ used in Figure~\ref{fig:numerics1}, the
breaking time before which there is a unique solution for all
$x\in\mathbb{R}$ and after which there is an expanding interval in
which there are three solutions, is exactly $T=2\sqrt{3}/9\approx
0.3849$.  Snapshots of the evolution of the multivalued solution of
\eqref{eq:inviscidBurgers} for this initial data are shown in
Figure~\ref{fig:numerics2}.
\begin{figure}[h]
\begin{center}
\includegraphics{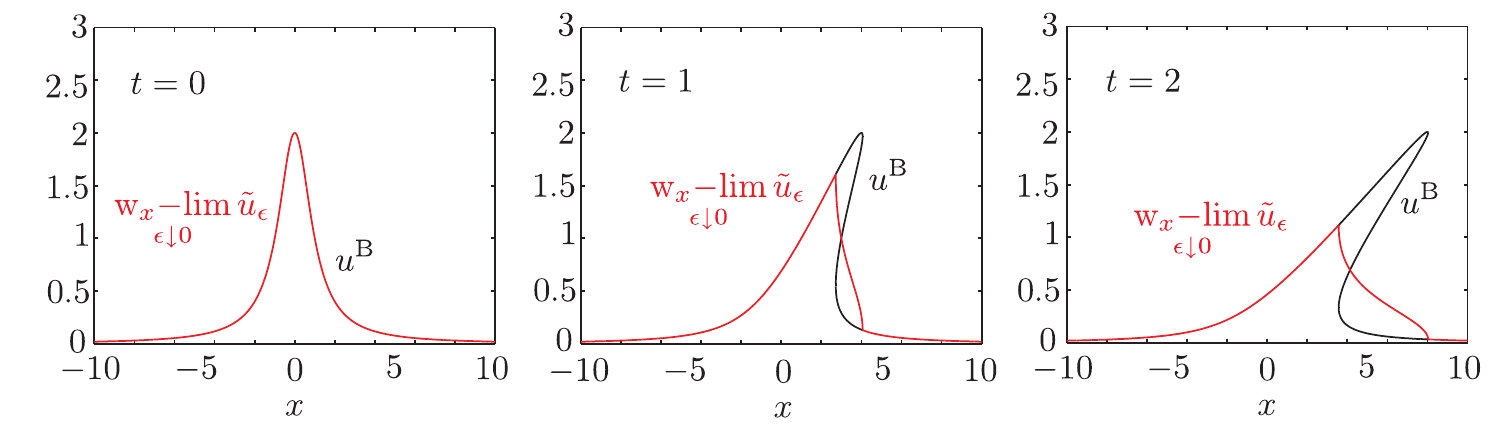}
\end{center}
\caption{\emph{The multivalued solution (black) of \eqref{eq:inviscidBurgers}
and the signed sum of branches (red) corresponding to $u_0(x)=2(1+x^2)^{-1}$.
Left:  $t=0$.  Middle:  $t=1$.  Right:  $t=2$.  Before the breaking time
as well as afterwards but outside the oscillation interval there is 
only 
one solution branch and hence no difference between the red and black curves.}}
\label{fig:numerics2}
\end{figure}
Our result is then the following.
\begin{theorem}
\label{MainTheorem} Let $u^{\mathrm{B}}_{0}(x,t) <
u^{\mathrm{B}}_{1}(x,t)< \cdots < u^{\mathrm{B}}_{2P(x,t)}(x,t)$ be the
branches of the multivalued (method of characteristics) solution of
the inviscid Burgers equation \eqref{eq:inviscidBurgers} subject to an
admissible initial condition $u^\mathrm{B}(x,0) = u_{0}(x)$. Then, the 
weak $L^2(\mathbb{R})$ (in $x$) limit of $\ut(x,t)$ is given by
\begin{equation}
\label{eq:MainResult}
\mathop{\mathrm{w}_x\mathrm{-lim}}_{\epsilon \downarrow
0}\ut(x,t)=\sum_{n=0}^{2P(x,t)}(-1)^{n}u^{\mathrm{B}}_{n}(x,t),
\end{equation}
uniformly for $t$ in arbitrary bounded intervals.  Note that the right-hand
side extends by continuity to the caustic curves.
\end{theorem}
The signed sum of branches that is the weak limit is illustrated with
red curves in Figure~\ref{fig:numerics2} for the same initial data as
in Figure~\ref{fig:numerics1}.  Of course convergence in the weak
$L^2(\mathbb{R})$ (in $x$) topology means that for every 
$v\in L^2(\mathbb{R})$, we have
\begin{equation}
\lim_{\epsilon\downarrow 0}\int_{\mathbb{R}}\ut(x,t)v(x)\,dx = 
\int_{\mathbb{R}}
\left[\sum_{n=0}^{2P(x,t)}(-1)^{n}u^{\mathrm{B}}_{n}(x,t)\right]v(x)\,dx
\end{equation}
with the limit being uniform with respect to $t$ in arbitrary bounded
intervals.  Thus, the weak limit essentially smooths out the rapid
oscillations seen in Figure~\ref{fig:numerics1} and (if we think of
$v$ as the indicator function of a mesoscale interval) represents a
kind of local average in $x$.  What it means for an initial condition
to be admissible will be explained later (see
Definition~\ref{def:admissible}).  Here $\ut(x,t)$ is not
exactly the solution of the Cauchy problem for the BO equation
\eqref{BOE} with fixed initial data $u_0(x)$, but it is for every
$\epsilon>0$ an exact solution of \eqref{BOE} that satisfies an
$\epsilon$-dependent initial condition that converges (in the strong
$L^2$ sense, see Corollary~\ref{corr:strong} below) to $u_0(\cdot)$ as
$\epsilon\downarrow 0$.  See Definition~\ref{def:modified} for more
details.  This modification of the initial data is an analogue of the
replacement of the true scattering data by its reflectionless WKB
approximation in the Lax-Levermore theory.  

For $t$ before the breaking time $T$ for Burgers' equation, the weak limit
guaranteed by Theorem~\ref{MainTheorem} may be strengthened as follows.
\begin{corollary}
Suppose that $0\le t<T$, so that $P(x,t)=0$ for all
$x\in\mathbb{R}$ (that is, the solution $u^\mathrm{B}=u^\mathrm{B}_0(x,t)$ 
of Burgers' equation with initial data $u_0(x)$ is
classical).  Then 
\begin{equation}
\lim_{\epsilon\downarrow 0}\ut(x,t) = u^\mathrm{B}_0(x,t)
\end{equation}
with the limit being in the (strong) $L^2(\mathbb{R}_x)$ topology.
\label{corr:strong}
\end{corollary}

It should be pointed out that the weak limit formula \eqref{eq:MainResult}
is much more explicit than the corresponding formula found by Lax and Levermore
\cite{Lax 1983-1} for the weak zero-dispersion limit of the Cauchy problem
for the KdV equation.  Indeed, the latter requires the solution, for each
$x$ and $t$, of a constrained functional variational problem, which can be 
solved in closed form only for the simplest initial data.  As we will now
see, there are several classical wave propagation problems whose asymptotic
behavior can be reduced to the multivalued solution of Burgers' equation;
however, even these simple problems involve more complicated schemes for
combining the solution branches than that exhibited in the simple formula
\eqref{eq:MainResult}.

After we introduce the necessary framework for our study (the inverse-scattering
transform for the BO equation) in 
\S\ref{sec:background}, we will analyze the direct scattering map in the zero-dispersion limit in \S\ref{sec:direct}.
Then we will prove Theorem~\ref{MainTheorem} in \S\ref{sec:inverse} by 
carrying out a detailed analysis of the inverse scattering map applied
to the asymptotic formulae for scattering data obtained in \S\ref{sec:direct}.
In \S\ref{sec:corollary} we prove Corollary~\ref{corr:strong}, and then
in \S\ref{sec:numerics} we will illustrate our results with numerical 
calculations and address the relation between $\ut(x,t)$ and the
true solution $\ueps(x,t)$ of the Cauchy problem for the BO equation 
with initial
data $u_0$.  Some comments about our continuing work can be found in 
the conclusion, \S\ref{sec:conclusion}.  But before we embark on our
study of the BO equation, we pause to consider some familiar analogues
of Theorem~\ref{MainTheorem}.

\subsection{Elementary examples}
The key role played in the zero-dispersion limit of the BO Cauchy problem
by the multivalued solution of the equation \eqref{eq:inviscidBurgers}
with the same initial
data is reminiscent of two basic example problems from the theory of linear
and nonlinear waves.
\subsubsection{The zero-viscosity limit of the viscous Burgers equation}
The Burgers equation with viscosity $\epsilon>0$ is the nonlinear wave equation
\begin{equation}
\frac{\partial w_\epsilon}{\partial t} + 2w_\epsilon
\frac{\partial w_\epsilon}{\partial x} -
\epsilon \frac{\partial^2w_\epsilon}{\partial x^2}=0,
\quad x\in\mathbb{R},\quad t>0
\end{equation}
and we take fixed initial data $w_\epsilon(x,0)=u_0(x)$.  As is well-known, this
Cauchy problem is solved by the Cole-Hopf transformation, leading to
the exact solution formula
\begin{equation}
w_\epsilon(x,t)=\frac{1}{2t}\frac{\displaystyle\int_\mathbb{R}
e^{R(\xi;x,t)/\epsilon}(x-\xi)\,d\xi}
{\displaystyle\int_\mathbb{R}
e^{R(\xi;x,t)/\epsilon}\,d\xi},\quad t>0,
\end{equation}
where the exponent function is defined as
\begin{equation}
R(\xi;x,t):=-\int_0^\xi u_0(\eta)\,d\eta -\frac{(x-\xi)^2}{4t}.
\end{equation}
One examines the asymptotic behavior in the limit $\epsilon\downarrow
0$ by using Laplace's method to analyze the integrals (see
\cite{Miller}, \S3.6).  The dominant contributions to the integrals
come from neighborhoods of points $\xi=\xi(x,t)\in\mathbb{R}$ at which
$R(\xi;x,t)$ achieves its maximum value.  The critical points of $R$
satisfy $\xi=x-2u_0(\xi)t$.  Writing $u^{\mathrm{B}}=u_0(\xi)$ and
applying $u_0$ to both sides gives the equation
\eqref{eq:implicitBurgersintro}, so the critical points
correspond to the sheets of the multivalued solution of the (inviscid)
Burgers equation \eqref{eq:inviscidBurgers} 
with initial data $u_0$.  It is easy to check that if
$x$ and $t$ are such that there is just one sheet, then the unique
critical point is the global maximizer of $R$ and Laplace's method
gives the result that $w_\epsilon(x,t)$ converges (strongly, pointwise
in $x$ and $t$) to $u^\mathrm{B}(x,t)$.  On the other hand, if there are
$2P+1>1$ sheets, then for generic $(x,t)$ exactly one of them
corresponds to the global maximum of $R$, and Laplace's method
predicts that $w_\epsilon(x,t)$ will converge to the maximizing sheet.
Shocks appear in the small viscosity limit as curves in the
$(x,t)$-plane along which there are jump discontinuities of the
pointwise limit corresponding to sudden changes in the choice of sheet
that maximizes the exponent $R$.  To summarize, we have the formula
\begin{equation}
\lim_{\epsilon\downarrow 0}w_\epsilon(x,t)=u^\mathrm{B}_n(x,t),\quad
n=\mathop{\mathrm{argmax}}_{m=0,\dots,2P(x,t)+1} R(x-2u^\mathrm{B}_m(x,t)t;x,t)
\end{equation}
for $(x,t)$ not on a shock.

Thus, one sees that for the zero-viscosity limit of the viscous
Burgers equation, different sheets of the multivalued solution of the
formal limiting Cauchy problem (set $\epsilon=0$) provide the strong
limit of $w_\epsilon(x,t)$ for different $x$ and $t$.  However, the
choice of sheet requires the solution of a discrete maximization problem
parametrized by $x$ and $t$, making the limiting behavior harder
to calculate than the weak zero-dispersion limit of the BO equation.
\subsubsection{The semiclassical limit of the free linear Schr\"odinger equation}
In this problem, one considers the equation
\begin{equation}
i\epsilon\frac{\partial \psi_\epsilon}{\partial t} +
2\epsilon^2\frac{\partial^2\psi_\epsilon}{\partial x^2}=0,\quad 
x\in\mathbb{R},\quad
t>0
\label{eq:Schroedinger}
\end{equation}
for small $\epsilon>0$, subject to initial data of WKB form
\begin{equation}
\psi_\epsilon(x,0)=A(x)e^{iS(x)/\epsilon}
\end{equation}
with $A$ and $S$ real-valued and independent of $\epsilon$.  For suitable
$A$ and $S$, the solution to
this problem can be written as an integral
\begin{equation}
\psi_\epsilon(x,t)=
\frac{e^{-i\pi/4}}{\sqrt{8\pi\epsilon t}}\int_\mathbb{R}
e^{iI(\xi;x,t)/\epsilon}A(\xi)\,d\xi,\quad t>0,\quad I(\xi;x,t):=S(\xi)
+\frac{(x-\xi)^2}{8t}.
\end{equation}
The dominant contributions to the solution are calculated via the
method of stationary phase (see \cite{Miller}, \S5.6), and these come
from small neighborhoods of points $\xi$ satisfying $I'(\xi;x,t)=0$,
that is, solutions $\xi$ of the implicit equation $\xi=x-4S'(\xi)t$.
Evaluating the function $2S'(\cdot)$ on both sides of this equation
and making the substitution $u^\mathrm{B}=2S'(\xi)$, one arrives at
the equivalent form \eqref{eq:implicitBurgersintro} where
$u_0(x):=2S'(x)$.  Thus, the branches of the multivalued solution of
Burgers' equation \eqref{eq:inviscidBurgers} with initial condition $u_0$
correspond to stationary phase points $\xi$ that
yield the leading term of the solution $\psi_\epsilon(x,t)$ in the
semiclassical limit $\epsilon\downarrow 0$.  Unlike in the analysis of
Laplace-type integrals, where only the critical points corresponding
to maxima matter in the limit, for oscillatory integrals all
stationary phase points contribute to the leading-order behavior, and
therefore we have an asymptotic representation of $\psi_\epsilon(x,t)$
as a sum over branches $u^\mathrm{B}_n(x,t)$ of the multivalued solution of
Burgers' equation with initial data $u_0$:
\begin{equation}
\psi_\epsilon(x,t)=\sum_{n=0}^{2P(x,t)}M_n(x,t)e^{i\theta_k(x,t;\epsilon)} + O(\epsilon),\quad t>0,
\label{eq:psiSP}
\end{equation}
where $M_n(x,t)$ are slowly-varying positive amplitudes given by
\begin{equation}
M_n(x,t):=A(x-2u_n^\mathrm{B}(x,t)t)\sqrt{\left|1-2t\frac{\partial 
u_{n}^\mathrm{B}}{\partial x}(x,t)\right|},
\end{equation}
and $\theta_k(x,t;\epsilon)$ are rapidly-varying real phases given by
\begin{equation}
\theta_n(x,t;\epsilon):=\frac{1}{\epsilon}I(x-2u_n^\mathrm{B}(x,t)t;x,t) + 
\frac{\pi}{4}
\left(\mathrm{sgn}\left(1-2t\frac{\partial u_{n}^\mathrm{B}}{\partial x}(x,t)
\right)-1\right),
\end{equation}
for $n=0,\dots,2P(x,t)$.

A more explicit connection with the multivalued solution of Burgers'
equation may be obtained by introducing the quantity
\begin{equation}
w_\epsilon(x,t):=2\epsilon\frac{\partial}{\partial x}
\Im\{\log(\psi_\epsilon(x,t))\}
\end{equation}
which is the fluid velocity in Madelung's interpretation of the wave
function $\psi_\epsilon$ as describing a quantum-corrected fluid
motion.  Under the condition that the error term in \eqref{eq:psiSP}
becomes $O(1)$ after differentiation with respect to $x$, some easy
calculations show that \eqref{eq:psiSP} implies
\begin{equation}
w_\epsilon(x,t)=\Re\left\{
\frac{\displaystyle\sum_{n=0}^{2P(x,t)}u_n^B(x,t)M_n(x,t)e^{i\theta_n(x,t;\epsilon)}}
{\displaystyle\sum_{n=0}^{2P(x,t)}M_n(x,t)e^{i\theta_n(x,t;\epsilon)}}\right\}
+O(\epsilon).
\end{equation}
It is then easy to see that if $P(x,t)=0$, $w_\epsilon(x,t)$ converges
strongly pointwise to $u_0^B(x,t)$, the unique solution (for this $x$
and $t$, anyway) of Burgers' equation.  On the other hand, if
$P(x,t)>0$, then there are interference effects among the terms in the
sums and these lead to rapid oscillations with the effect that
$w_\epsilon(x,t)$ no longer converges in the pointwise sense as
$\epsilon\downarrow 0$.  However, it does converge in the weak
topology.  The weak limit may be computed by multiphase averaging,
which we illustrate in the case $P(x,t)=1$.  The procedure is to
average the leading term in $w_\epsilon(x,t)$ over an interval in 
$x$ centered at the point of interest of radius, say, $\epsilon^{p}$ for
some $p\in (0,1)$, and then pass to the limit $\epsilon\downarrow 0$.
This produces the desired local average over rapid oscillations of
wavelength or period proportional to $\epsilon$.  Under an ergodic
hypothesis that is valid on a set of full measure in the $(x,t)$-plane,
this procedure is equivalent to holding $u^\mathrm{B}_n$ and $M_n$ fixed and 
averaging (with uniform measure) over the torus of relative angles
$\phi_1:=\theta_1-\theta_0$ and $\phi_2:=\theta_2-\theta_0$.  The double
integrals can be evaluated explicitly, with the result that 
\begin{equation}
\mathop{\mathrm{w}_x\mathrm{-lim}}_{\epsilon\downarrow 0}w_\epsilon(x,t)=
\sum_{n=0}^2 c_n(x,t)u_n^\mathrm{B}(x,t),
\end{equation}
where $c_n(x,t)$, $n=0,1,2$, are nonnegative coefficients with the
property that $c_0(x,t)+c_1(x,t)+c_2(x,t)=1$.  Specifically, the
coefficients only depend on $x$ and $t$ through $M_0(x,t)$,
$M_1(x,t)$, and $M_2(x,t)$.  If any of these, say $M_n$, exceeds the
sum of the other two, then $c_n=1$ and the two other coefficients
vanish.  Thus the weak limit produces in this case exactly the branch
$u_n^B(x,t)$ through multiphase averaging\footnote{ The strict
  inequality, $M_0(x,t)>M_1(x,t)+M_2(x,t)$ or a permutation thereof,
  defining this situation is an open condition on
  $(x,t)\in\mathbb{R}^2$, and therefore (depending on initial conditions) 
there can exist open domains
  in the $(x,t)$-plane on which the weak limit of $w_\epsilon(x,t)$ is
  given by a single branch of the solution of the inviscid Burgers
  equation while $w_\epsilon(x,t)$ itself exhibits wild oscillations.  
Interestingly, this is precisely the conjecture made by von Neumann
regarding grid-scale oscillations observed in the numerical solution
of Burgers' equation via a finite-difference scheme (which may be
viewed as a dispersive regularization of the equation).  While 
many model equations for finite-difference schemes (the KdV equation
is one example) do not yield such a simple interpretation of the weak
limit \cite{Lax86}, it would seem that von Neumann's conjecture can
hold true if the Schr\"odinger equation is viewed as a dispersive
correction to Burgers' equation.}.  On the other hand, if none of the
$M_n$ exceeds the sum of the other two, then $M_0$, $M_1$, and $M_2$
are the side lengths of a triangle, and the weak limit is a genuine
weighted average of the three branches, with weights proportional to
the opposite angles:
\begin{equation}
\begin{split}
c_0&=\frac{1}{\pi}\arccos\left(\frac{M_1^2+M_2^2-M_0^2}{2M_1M_2}\right)\\
c_1&=\frac{1}{\pi}\arccos\left(\frac{M_0^2+M_2^2-M_1^2}{2M_0M_2}\right)\\
c_2&=\frac{1}{\pi}\arccos\left(\frac{M_0^2+M_1^2-M_2^2}{2M_0M_1}\right).
\end{split}
\end{equation}
The most significant aspect of this analysis is that the weak limit
depends on information other than just the initial condition $u_0$
for Burgers' equation since the functions $M_n(x,t)$ involve also the
initial wave function amplitude $A$.  This makes the evaluation of the
weak limit a more complicated procedure than in the case of the BO
equation.

A further connection between the BO equation 
\eqref{BOE} and the linear Schr\"odinger equation \eqref{eq:Schroedinger}
in the zero-dispersion limit will be made in \S\ref{sec:reflectionless}.

\section{Relevant Aspects of the Inverse Scattering Transform for the BO Cauchy Problem}
\label{sec:background}
\subsection{The Lax pair for the BO equation and its basic properties}
The Lax pair \cite{Bock 1979}, whose compatibility condition is the BO
equation \eqref{BOE}, consists of the two equations
\begin{equation}
\label{Laxpair1} i\epsilon \frac{\partial w^+}{\partial x}+\lambda (w^{+}-w^{-})=-\ueps w^{+}
\end{equation}
\begin{equation}
\label{Laxpair2} i\epsilon \frac{\partial w^\pm}{\partial t}-
2i\lambda \epsilon \frac{\partial w^\pm}{\partial x}+\epsilon^2
\frac{\partial^2 w^\pm}{\partial x^2}-2i\mathcal{C}_{\pm}
\left[\epsilon \frac{\partial \ueps}{\partial x}\right]w^{\pm}=0,
\end{equation}
where $\lambda\in\mathbb{C}$ is a spectral parameter, $\ueps=\ueps(x,t)$ is a
solution of \eqref{BOE}, and $w^\pm=w^\pm(x,t;\lambda)$ are functions that
are required to be, for each fixed $t$ and $\lambda$, the boundary values on
the real $x$-axis of functions analytic in the upper ($+$) and lower
($-$) half complex $x$-plane.  Also, $\pm \mathcal{C}_\pm =
\tfrac{1}{2}\mathbb{I}\mp\tfrac{1}{2}i\mathcal{H}$ are the orthogonal
and complementary ($\mathcal{C}_+-\mathcal{C}_-=\mathbb{I}$, the
Plemelj formula) projections from $L^2(\mathbb{R})$ onto its upper and
lower Hardy subspaces $\mathbb{H}^\pm(\mathbb{R})$.  
From the point of view of the inverse-scattering
transform, \emph{i.e.\@} using the Lax pair as a tool to solve the Cauchy 
problem, equation \eqref{Laxpair1} may
be considered for fixed time $t$ and defines the scattering data
associated with $\ueps(x,t)$ at time $t$.  The function $w^-$
may be viewed as a kind of Lagrange multiplier present to satisfy the
constraint that $w^+$ be an ``upper'' function.  In fact, if
$w^\pm\in \mathbb{H}^\pm$ then by applying $\mathcal{C}_+$ to \eqref{Laxpair1}
and using the projective identities $\mathcal{C}_+[w^+]=w^+$ and
$\mathcal{C}_+[w^-]=0$, \eqref{Laxpair1}
can be written in the form of an eigenvalue problem
\begin{equation}
\mathcal{L}w^+ = \lambda w^+,\quad w^+\in H^+(\mathbb{R}),
\end{equation}
where $\mathcal{L}$ is the nonlocal selfadjoint operator \eqref{eq:Loperator}.
Equation \eqref{Laxpair2}
determines the (trivial, as we will recall) time dependence of the
scattering data.

\subsection{Scattering data}
The theory of the inverse-scattering transform solution of the Cauchy
problem for the BO equation was first developed by Fokas and Ablowitz
\cite{Fokas 1983}.  Certain analytical details of the theory were
clarified by Coifman and Wickerhauser \cite{Coifman 1990}, and more
recently Kaup and Matsuno \cite{Kaup 1998} found conditions on the
scattering data consistent with real-valued solutions of \eqref{BOE}.
As an operator on $\mathbb{H}^+(\mathbb{R})$, the essential spectrum
of $\mathcal{L}$ is the positive real $\lambda$-axis (for suitable $\ueps$,
$\mathcal{L}$ is a relatively compact perturbation of the ``free''
operator corresponding to $\ueps\equiv 0$).  For each fixed $t$ and each
real $\lambda >0$, there exists a unique solution $w^+=M$ of
\eqref{Laxpair1} with the property that (remarkably, despite the
nonlocal nature of the problem) it is determined by its asymptotic
behavior as $x\to -\infty$ on the real line: $M(x,t;\lambda)=1+o(1)$
as $x\to -\infty$.  As the problem is nonlocal, $M$ cannot be
characterized by a Volterra-type integral equation, but Fokas and
Ablowitz \cite{Fokas 1983} gave a Fredholm-type equation whose unique
solution is $M$.  The \emph{reflection coefficient} for the problem is
then defined for positive real $\lambda$ by the formula \cite{Fokas
  1983,Kaup 1998}
\begin{equation}
\beta(\lambda,t):= \frac{i}{\epsilon}
\int_\mathbb{R}\ueps(x,t)M(x,t;\lambda)e^{-i\lambda x/\epsilon}\,dx,\quad
\lambda>0.
\end{equation}
For each fixed $x$ and $t$ the function $M(x,t;\lambda)$ can be shown
to be the boundary value taken on the positive half-line $\lambda\in
\mathbb{R}_+$ from the upper half $\lambda$-plane of a function
$W(x,t;\lambda)$ that is meromorphic in the complex $\lambda$-plane
with $\mathbb{R}_+$ (a branch cut) deleted.  Fokas and Ablowitz refer
to the boundary value taken by $W(x,t;\lambda)$ on the positive
half-line from the \emph{lower} half-plane as
$\overline{N}(x,t;\lambda)$.  The poles of $W(x,t;\lambda)$ are all on
the negative real $\lambda$-axis (by self-adjointness of
$\mathcal{L}$) and correspond to the point spectrum of $\mathcal{L}$.
It turns out that one of the consequences of the Lax pair equation
\eqref{Laxpair2} is that the point spectrum is independent of time
$t$.  In \cite{Fokas 1983} it is shown that in the generic case when
$\lambda_n<0$ is a simple pole of $W$, the first two terms in the
Laurent expansion of $W$ at $\lambda=\lambda_n$ are both proportional
to the same function $\Phi_n(x,t)\in \mathbb{H}^+(\mathbb{R}_x)$,
which is an eigenfunction of $\mathcal{L}$ with eigenvalue
$\lambda=\lambda_n$.  The ratio of these two terms is in fact linear
in $x$:
\begin{equation}
W(x,t;\lambda)=-i\epsilon\frac{\Phi_n(x,t)}{\lambda-\lambda_n} + 
(x+\alpha_n(t))\Phi_n(x,t)
+ O(\lambda-\lambda_n),\quad \lambda\to \lambda_n.
\label{eq:Laurent}
\end{equation}
Kaup and Matsuno \cite{Kaup 1998} showed that for real $u$, the
complex-valued phase shift $\alpha_n(t)$ may be written in the form
\begin{equation}
\alpha_n(t) = \gamma_n(t) -\frac{i}{2\lambda_n},\quad \gamma_n(t)\in\mathbb{R}.
\end{equation}
The set of \emph{scattering data} corresponding to the potential 
$\ueps(\cdot,t)$ then consists of:
\begin{itemize}
\item The reflection coefficient $\beta(\lambda,t)$ for $\lambda>0$.  We write
$\beta(\lambda)$ for $\beta(\lambda,0)$.
\item The negative discrete eigenvalues 
$\{\lambda_n\}_{n=1}^N$, $\lambda_1<\lambda_2<\cdots <\lambda_N<0$.
\item The real \emph{phase constants} $\{\gamma_n(t)\}_{n=1}^N$.  We write
$\gamma_n$ for $\gamma_n(0)$.
\end{itemize}

\subsection{Time dependence of the scattering data and the inverse 
scattering transform}
As time varies, one may expect the scattering data to vary, but the
time dependence as implied by \eqref{Laxpair2} turns out to be very
simple.  As pointed out above, the discrete eigenvalues
$\{\lambda_n\}_{n=1}^N$ are constants of the motion, and Fokas and
Ablowitz \cite{Fokas 1983} showed that
\begin{equation}
\beta(\lambda,t) = \beta(\lambda)e^{i\lambda^2t/\epsilon},\quad \lambda>0,
\label{eq:timevariationbeta}
\end{equation}
and
\begin{equation}
\gamma_n(t)=\gamma_n + 2\lambda_nt,\quad n=1,\dots,N.
\label{eq:timevariationgamman}
\end{equation}

The inverse-scattering procedure for solving the Cauchy problem for
the BO equation with suitable real initial data $u_0(x)$ is then to
calculate the scattering data at time $t=0$ from $u_0$, evolve the
scattering data forward in time $t$ by the explicit formulae
\eqref{eq:timevariationbeta} and \eqref{eq:timevariationgamman}, and
then solve the inverse problem of constructing $\ueps(\cdot,t)$ from the
scattering data at time $t$.  Generally, this requires solving a
scalar Riemann-Hilbert problem for $W(x,t;\lambda)$ in the complex
$\lambda$-plane.  This Riemann-Hilbert problem is quite interesting as
it involves a jump condition across the continuous spectrum
$\lambda>0$ in which the boundary value from above,
$W=M(x,t;\lambda)$, is proportional to an integral from $\lambda'=0$
to $\lambda'=\lambda$ of the boundary value from below,
$W=\overline{N}(x,t;\lambda')$.  Thus the jump condition is nonlocal,
a fact that makes the inverse problem almost completely analogous to
the direct problem \eqref{Laxpair1} which, after integration becomes a
nonlocal Riemann-Hilbert problem of exactly the same type in the
complex $x$-plane.  This fact should perhaps be contrasted with the
situation for the KdV equation where the direct and inverse problems
are of quite different natures.  This remarkable symmetry between the
forward and inverse problems for the BO equation is a theme that will
be touched upon again in this paper in some detail.
\label{sec:timedependence}

\subsection{The reflectionless inverse scattering transform}
\label{sec:reflectionless}
If $\beta(\lambda)\equiv 0$ (\emph{i.e.\@} the problem is
\emph{reflectionless}), then the boundary values taken by $W(x,t;\lambda)$
on the positive half-line agree, so $W(x,t;\lambda)$ is a meromorphic
function on the whole complex $\lambda$-plane with simple poles at the
negative real eigenvalues.  The condition \eqref{eq:Laurent} 
 and the normalization
condition that $W(x,t;\lambda)\to 1$ as $\lambda\to\infty$ then provides
sufficient information to reconstruct $W(x,t;\lambda)$ from the discrete
data $\{\lambda_n\}_{n=1}^N$ and $\{\gamma_n\}_{n=1}^N$.  Via a partial-fractions
ansatz for $W(x,t;\lambda)$, this amounts to a
solving a linear algebra problem in dimension $N$.  Once $W$ is determined
in this way, one obtains $\mathcal{C}_+[\ueps(\cdot,t)]$ by the formula
\begin{equation}
\mathcal{C}_+[\ueps(\cdot,t)](x)=\lim_{\lambda\to\infty}\lambda(1-W(x,t;\lambda)).
\end{equation}
Since $\ueps$ is real, one then has
\begin{equation}
\ueps(x,t)=2\Re\{\mathcal{C}_+[\ueps(\cdot,t)](x)\}.
\end{equation}
This procedure clearly leads to a determinantal formula for $\ueps(x,t)$
in the reflectionless case.  It turns out to be the same (multisoliton)
formula that Matsuno \cite{Matsuno 1979} had obtained, before the relevant
inverse-scattering transform was discovered, by applying 
Hirota's bilinear method to the BO equation:
\begin{equation}
\ueps(x,t)=2\epsilon\frac{\partial}{\partial x}
\Im\left\{\log(\tau_\epsilon(x,t))\right\},
\label{eq:Hirotatransform}
\end{equation}
with the ``tau-function''
\begin{equation}
\tau_\epsilon(x,t):=\det(\mathbb{I}+i\epsilon^{-1}\mat{A}_\epsilon),
\label{eq:taudet}
\end{equation}
where $\mat{A}_\epsilon=\mat{A}_\epsilon(x,t)$ is an $N\times N$ Hermitean 
matrix with constant off-diagonal
elements
\begin{equation}
A_{nm}:=\frac{2i\epsilon\sqrt{\lambda_n\lambda_m}}{\lambda_n-\lambda_m},\quad n\neq m,
\label{eq:Aoffdiag}
\end{equation}
and diagonal elements depending explicitly $x$ and $t$:
\begin{equation}
A_{nn}:=-2\lambda_n(x+2\lambda_nt+\gamma_n).
\label{eq:Adiag}
\end{equation}
In \eqref{eq:Aoffdiag} we mean the positive square root of the
positive product $\lambda_n\lambda_m$.  For the purposes of this
paper, we will only require this reflectionless version of the
inverse-scattering transform.

In his paper \cite{Matsuno 1979}, Matsuno noted that regardless of the
value of $N$, the complex determinant $\tau_\epsilon(x,t)$ satisfies the real
equation (Hirota bilinear form of the BO equation)
\begin{equation}
\left(i\epsilon\frac{\partial \tau_\epsilon}{\partial t} + 2\epsilon^2
\frac{\partial^2\tau_\epsilon}{\partial x^2}\right)\tau_\epsilon^* +
\left(-i\epsilon\frac{\partial \tau_\epsilon^*}{\partial t} + 
2\epsilon^2\frac{\partial^2\tau_\epsilon^*}{\partial x^2}\right)\tau_\epsilon = 
\epsilon\frac{\partial}{\partial x}
\left(\epsilon\frac{\partial \tau_\epsilon}{\partial x}\tau_\epsilon^* +
\epsilon\frac{\partial \tau_\epsilon^*}{\partial x}\tau_\epsilon\right).
\label{eq:Hirota}
\end{equation}
The terms on the left-hand side should be compared with the linear
Schr\"odinger equation \eqref{eq:Schroedinger}.  If one makes a formal
WKB ansatz of the form $\tau_\epsilon(x,t)=A(x,t)e^{iS(x,t)/\epsilon}$, then 
the terms on the right-hand side of \eqref{eq:Hirota} are formally small
compared with those on the left-hand side, and to leading
order in $\epsilon$ \eqref{eq:Hirota} simply reduces to 
the inviscid Burgers equation \eqref{eq:inviscidBurgers}
with $u^\mathrm{B}=2\partial S/\partial x$ 
(as is consistent with \eqref{eq:Hirotatransform}).  

\subsection{Conservation laws and trace formulae}
As with the KdV equation, the time evolution of BO equation preserves
an infinite number of functionals of $u$.  These were first found by
Nakamura \cite{Nakamura 1979}.  The equivalent representation of these
functionals in terms of the time-independent portion of the scattering
data, \emph{i.e.\@} the eigenvalues $\{\lambda_n\}_{n=1}^N$ and the
modulus of the reflection coefficient $|\beta(\lambda)|^2$ for $\lambda>0$, was
obtained by Kaup and Matsuno \cite{Kaup 1998}.  These identities amount
to a hierarchy of \emph{trace formulae} for the operator $\mathcal{L}$.

The conservation laws take the form $d I_k/dt=0$, $k=1,2,3,\dots$.  
The integrals $I_k$ may be generated by the following recursive procedure:
first set $\rho_1:=1$ and then define
\begin{equation}
\rho_{k+1}(x,t):=\mathcal{C}_+[\ueps(\cdot,t)\rho_k(\cdot,t)](x) 
+i\epsilon\frac{\partial \rho_k}{\partial x}(x,t),\quad k=1,2,3,\dots.
\label{eq:rhorecurrence}
\end{equation}
Then, the integrals of motion are
\begin{equation}
I_k(t):=\int_\mathbb{R}\ueps(x,t)\rho_k(x,t)\,dx,\quad k=1,2,3,\dots.
\label{eq:integralsofmotion}
\end{equation}
The equivalent spectral representation given in
\cite{Kaup 1998} is
\begin{equation}
I_k(t) = 2\pi\epsilon\sum_{n=1}^N
(-\lambda_n)^{k-1} +\frac{(-1)^k\epsilon}{2\pi}\int_0^{+\infty}
|\beta(\lambda)|^2\lambda^{k-2}\,d\lambda,\quad k=1,2,3,\dots.
\label{eq:equivalentintegrals}
\end{equation}
In view of the results presented in \S\ref{sec:timedependence}, the
latter representation makes clear the fact that $dI_k/dt=0$.

The first two conserved quantities are quite simple and in fact they
are the only ones in the
hierarchy having local densities:
\begin{equation}
I_1:=\int_\mathbb{R}\ueps(x,t)\,dx \quad\text{and}\quad
I_2:=\frac{1}{2}\int_\mathbb{R}\ueps(x,t)^2\,dx.
\label{eq:I2}
\end{equation}

\section{The Scattering Data in the Zero-Dispersion Limit}
\label{sec:direct}
In this section we consider the following problem.  Given a suitable
function $u_0(x)$ representing the initial condition for the BO
equation, we wish to determine an asymptotic approximation, valid when
$\epsilon>0$ is small, to the scattering data
$\{\beta(\lambda),\{\lambda_n\}_{n=1}^N,\{\gamma_n\}_{n=1}^N\}$
corresponding to $u_0$.  Even though $u_0$ is held fixed as $\epsilon$
tends to zero, the scattering data will depend on $\epsilon$ as this
parameter appears in the equation \eqref{Laxpair1}.  As the operator
$\mathcal{L}$ is nonlocal, we cannot rely on the WKB method as is so
useful for analysis of differential operators (for example, the
analysis of Lax and Levermore \cite{Lax 1983-1} was based on the WKB
analysis of the Schr\"odinger operator that arises in the scattering
theory for the KdV equation).
\subsection{Admissible initial conditions}
The type of initial data for the BO equation \eqref{BOE} that we 
will consider for the rest of this paper is the following.
Many of these conditions are imposed for our convenience; we make no 
claim that they are necessary.  
\begin{definition}
A function $u_0:\mathbb{R}\to \mathbb{R}$ will be called an 
\emph{admissible initial condition} if it has the following properties:

\noindent\textbf{\textit{Smoothness:}}  $u_{0}\in C^{(3)}(\mathbb{R})$.

\noindent\textbf{\textit{Positivity:}}  $u_0(x)>0$ for all $x\in\mathbb{R}$.

\noindent\textbf{\textit{Existence of a Unique Critical Point:}}  There is a 
unique point $x_0\in\mathbb{R}$ for which $u_0'(x_0)=0$ and
\begin{equation}
\label{u0 maximum point condition}u''_{0}(x_{0})< 0,
\end{equation}
making $x_0$ the global maximizer of $u_0$.

\noindent\textbf{\textit{Tail Behavior:}}
$\lim_{x\to\pm\infty}u_0(x)=0$, and 
\begin{equation}
\label{u0 infinity boundary condition} \lim_{x\to \pm
\infty}|x|^{q+1}u'_{0}(x)=C_{\pm} \quad \text{for some $q>1$,}
\end{equation}
where $C_{+}<0$ and $C_{-}>0$ are constants. These two conditions together
imply that an admissible initial condition $u_0$ also satisfies 
\begin{equation}
\label{u0 infinity boundary condition1} \lim_{x\to \pm
\infty}|x|^{q}u_{0}(x)=\mp \frac{C_{\pm}}{q}.
\end{equation}

\noindent\textbf{\textit{Inflection Points:}}  In each bounded interval 
there exist at most finitely many points $x=\xi$ at which $u_0''(\xi)=0$, 
and each is a simple inflection point:  $u_0'''(\xi)\neq 0$.
\label{def:admissible}
\end{definition}
Corresponding to an admissible initial condition $u_0$ we define a positive
constant $L$ by
\begin{equation}
L:=\mathop{\mathrm{max}}_{x\in \mathbb{R}}u_{0}(x),
\label{eq:Ldef}
\end{equation}
and we let the \emph{mass} $M$ be defined by
\begin{equation}
M:=\frac{1}{2\pi}\int_\mathbb{R}u_0(x)\,dx.
\label{eq:massdef}
\end{equation}
Note that the mass is guaranteed to be finite according to 
\eqref{u0 infinity boundary condition1} since $u_0$ is bounded.
Also, if $u_{0}$ is an admissible initial condition, we can define
\emph{turning points} $x_\pm:[-L,0)\to\mathbb{R}$ as two monotone branches of
the inverse function of $u_0$: $u_{0}(x_\pm(\lambda))=-\lambda$ and $x_{-}(\lambda)\le
x_{0}\le x_+(\lambda)$ for $-L\le \lambda<0$.  See 
Figure~\ref{small_dispersion_paper_picture1}.
\begin{figure}[h]
\begin{center}
\includegraphics{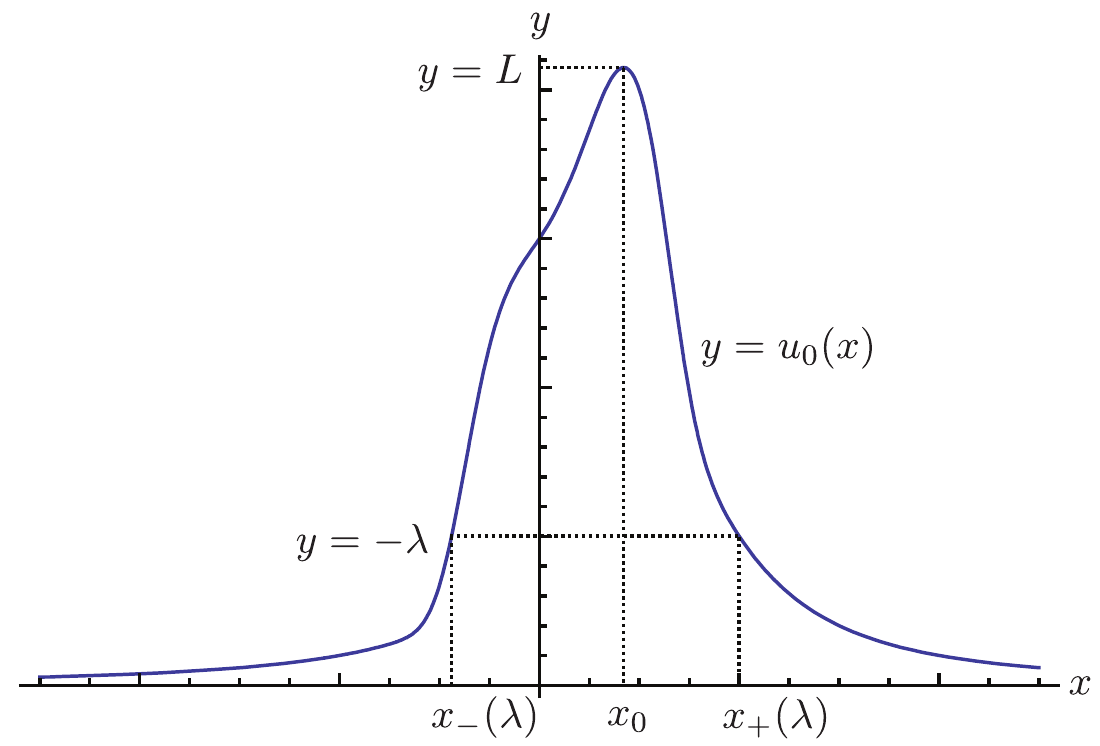}
\end{center}
\caption{\emph{The graph of an admissible
initial condition and the turning points $x_{\pm}(\lambda)$.}}
\label{small_dispersion_paper_picture1}
\end{figure}
\subsection{Matsuno's method}
\label{sec:MatsunosMethod}
In two papers \cite{Matsuno 1981, Matsuno 1982}, Matsuno proposed a
remarkable method to approximate, in the limit $\epsilon \downarrow
0$, the time-independent components of the scattering data for
suitable $u_0$.  His method was based on the conservation laws for the
quantities \eqref{eq:integralsofmotion}.  With the use of the more
recently obtained trace formulae equating $I_k$ as given by
\eqref{eq:integralsofmotion} with the equivalent formulae
\eqref{eq:equivalentintegrals} \cite{Kaup 1998}, several heuristic
aspects of the original method given in \cite{Matsuno 1981, Matsuno
  1982} can be placed on more rigorous footing.

The first key observation made in \cite{Matsuno 1981, Matsuno 1982}
is that if $u_0$ is a smooth function independent of $\epsilon$, then
by evaluating the integrals $I_k(t)$ at time $t=0$, one sees that they
have limiting values as $\epsilon\downarrow 0$.  These limits may be obtained
simply by solving the recurrence relation \eqref{eq:rhorecurrence}
with $\epsilon=0$:
\begin{equation}
\lim_{\epsilon\downarrow 0} I_k = \int_\mathbb{R}
u_0(x)\mathcal{C}_+[u_0\mathcal{C}_+[u_0\mathcal{C}_+[\cdots u_0
\mathcal{C}_+[u_0] \cdots]]](x)\,dx,\quad k=1,2,3,\dots,
\label{eq:limIkfirst}
\end{equation}
where the Cauchy projector $\mathcal{C}_+$ occurs $k-1$ times in the
integrand.  
With the use of an identity valid for reasonable complex-valued functions
$u_0(\cdot)$ and suggested by comparing the conserved quantities generated
from the Kaup-Matsuno iteration scheme \eqref{eq:rhorecurrence} with those
generated via the older scheme of Nakamura \cite{Nakamura 1979}, one sees
that the right-hand side of \eqref{eq:limIkfirst}
can be equivalently
written in the simple form
\begin{equation}
\lim_{\epsilon\downarrow 0}I_k = \frac{1}{k}\int_\mathbb{R}u_0(x)^k\,dx,
\quad k=1,2,3,\dots.
\label{eq:limIk}
\end{equation}
On the basis of heuristic physical arguments, in \cite{Matsuno 1981,
  Matsuno 1982} Matsuno supposed that for smooth \emph{positive}
initial data $u_0$, all moments of the reflection coefficient remain
bounded as $\epsilon\downarrow 0$.  Adopting this hypothesis, a comparison
of \eqref{eq:limIk} with \eqref{eq:equivalentintegrals} then shows that
\begin{equation}
\lim_{\epsilon\downarrow 0} \epsilon\sum_{n=1}^N(-\lambda_n)^{k-1} = 
\frac{1}{2\pi k}\int_\mathbb{R}u_0(x)^k\,dx, \quad k=1,2,3,\dots.
\label{eq:limitingmoments}
\end{equation}
In particular, taking $k=1$ one obtains
\begin{equation}
\lim_{\epsilon\downarrow 0}\epsilon N = M,
\end{equation}
where the mass $M$ is defined by \eqref{eq:massdef}, 
so the number of eigenvalues is asymptotically proportional to $1/\epsilon$.

These calculations suggest that the normalized counting measure of
eigenvalues may have a limit in a certain sense as $\epsilon\downarrow
0$, perhaps as an absolutely continuous measure with density
$F(\lambda)$.  Matsuno calculated this density by replacing the
left-hand side of \eqref{eq:limitingmoments} with an integral against
the unknown density $F(\lambda)$:
\begin{equation}
\int_{-\infty}^0(-\lambda)^{k-1}F(\lambda)\,d\lambda = \frac{1}{2\pi k}
\int_\mathbb{R}u_0(x)^k\,dx,\quad k=1,2,3,\dots.
\label{eq:momentrelations}
\end{equation}
The problem that remains is then the classical one 
of constructing the density $F(\lambda)$ from its moments, which are known
if the initial condition $u_0$ is given.  

Matsuno showed that, remarkably, this moment problem can be solved explicitly.
He introduced the characteristic function (Fourier transform) of $F$:
\begin{equation}
\hat{F}(\xi):=\int_{-\infty}^0 F(\lambda)e^{-i\xi\lambda}\,d\lambda,
\end{equation}
in terms of which, the moment relations \eqref{eq:momentrelations} become
\begin{equation}
\frac{d^{k-1}\hat{F}}{d\xi^{k-1}}(0)=\frac{i^{k-1}}{2\pi k}\int_\mathbb{R}
u_0(x)^k\,dx,\quad k=1,2,3,\dots.
\end{equation}
Recalling the constants $L$ and $M$ defined by \eqref{eq:Ldef} and
\eqref{eq:massdef} respectively, it is easy to obtain the estimate
\begin{equation}
\left|\frac{1}{(k-1)!}\frac{d^{k-1}\hat{F}}{d\xi^{k-1}}(0)\right|\le
\frac{ML^{k-1}}{k!}
\end{equation}
from which it follows that $\hat{F}(\xi)$ is an entire function and hence
is equal to its Taylor series about $\xi=0$:
\begin{equation}
\begin{split}
\hat{F}(\xi) &= \sum_{k=1}^\infty\frac{1}{(k-1)!}\frac{d^{k-1}\hat{F}}
{d\xi^{k-1}}(0)\xi^{k-1}\\
& = \sum_{k=1}^\infty \frac{(i\xi)^{k-1}}{2\pi k!}
\int_\mathbb{R}u_0(x)^k\,dx\\
&=\frac{1}{2\pi i\xi}\sum_{k=1}^\infty\int_\mathbb{R}
\frac{[i\xi u_0(x)]^k}{k!}\,dx.
\end{split}
\end{equation}
The combined sum and integral is absolutely convergent for all
$\xi\in\mathbb{C}$, so the order of operations may be reversed:
\begin{equation}
\hat{F}(\xi) = \frac{1}{2\pi i\xi}\int_\mathbb{R}
\left(\sum_{k=1}^{\infty}\frac{[i\xi u_0(x)]^k}{k!}\right)\,dx = 
\frac{1}{\pi \xi}\int_\mathbb{R}
e^{i\xi u_0(x)/2}\sin\left(\tfrac{1}{2}\xi u_0(x)\right)
\,dx.
\end{equation}
By Fourier inversion, 
\begin{equation}
F(\lambda) = \frac{1}{2\pi}\lim_{R\uparrow \infty}\int_{-R}^{+R}\hat{F}(\xi)e^{i\xi\lambda}\,
d\xi = \lim_{R\uparrow\infty}\int_{-R}^{+R}
\int_\mathbb{R}\frac{e^{i\xi(\lambda+u_0(x)/2)}}{2\pi^2\xi}\sin
\left(\tfrac{1}{2}\xi u_0(x)\right)
\,dx\,d\xi.
\end{equation}
Applying Fubini's Theorem to reverse the order of integration and
then passing to the limit $R\uparrow\infty$, the integral over $\xi$
can be evaluated as the indicator function of an interval:
\begin{equation}
F(\lambda)=\frac{1}{2\pi}\int_\mathbb{R}\chi_{[-u_0(x),0]}(\lambda)\,
dx. 
\end{equation}
This formula shows that $F(\lambda)\equiv 0$ for $\lambda>0$ or
$\lambda<-L$.  By a ``layer-cake'' argument we may simplify this
formula for $\lambda\in (-L,0)$ as
\begin{equation}
F(\lambda)= \frac{1}{2\pi}\int_{\{x\in\mathbb{R},\;u_0(x)>-\lambda\}}dx,
\quad -L<\lambda<0.
\label{eq:MatsunoBS}
\end{equation}
This is Matsuno's remarkable result.  We have presented Matsuno's
method in some detail because it turns out that a key calculation in
our analysis of the \emph{inverse} problem in the zero-dispersion
limit reduces to almost the same steps, as we will see shortly.  This
is worth emphasizing because it provides further evidence that for the
BO equation, scattering and inverse-scattering are mathematically very
similar operations.

Matsuno's formula \eqref{eq:MatsunoBS} could perhaps be compared with
the Bohr-Sommerfeld formula that gives the density of eigenvalues of
the Schr\"odinger operator in the zero-dispersion theory of the KdV
equation \cite{Lax 1983-1}; aside from a constant factor the Bohr-Sommerfeld
formula replaces the unit integrand in \eqref{eq:MatsunoBS} with the
positive square root $\sqrt{u_0(x)+\lambda}$.

While quite severe hypotheses on $u_0$ are required for all of the arguments
to go through, the formula \eqref{eq:MatsunoBS} makes sense under much
weaker conditions.  In particular, we may interpret \eqref{eq:MatsunoBS}
for an admissible initial condition, in which case we may express $F(\lambda)$
directly in terms of the turning points $x_\pm(\lambda)$:
\begin{equation}
F(\lambda):=\frac{1}{2\pi}\left(x_+(\lambda)-x_-(\lambda)\right),\quad
-L\le \lambda<0.
\label{eq:Fadmissibledef}
\end{equation}
We take \eqref{eq:Fadmissibledef} as a definition valid for admissible
initial conditions $u_0$.  Note that
\begin{equation}
\int_{-L}^0F(\lambda)\,d\lambda = M,
\label{eq:intF}
\end{equation}
where the mass $M$ is defined by \eqref{eq:massdef}.

\subsection{Formula for phase constants}
The WKB methods recalled by Lax and Levermore \cite{Lax 1983-1} to
analyze the Schr\"odinger equation in the forward problem for the
zero-dispersion limit of the KdV equation were sufficiently powerful
to provide asymptotic formulae for both the discrete spectrum (the
Bohr-Sommerfeld formula that is the analogue in the KdV theory of the
function $F(\lambda)$ obtained by Matsuno) and also for the ``norming
constants'' that in the KdV theory are the analogues of the phase
constants $\{\gamma_n\}_{n=1}^N$ in the BO theory.  However, we have
not found a way to apply these methods to the nonlocal operator
$\mathcal{L}$, and unfortunately Matsuno's method does not provide
approximations of the phase constants $\{\gamma_n\}_{n=1}^N$ since
they do not enter into the trace formulae.  

Our contribution to the theory of the spectral analysis of the nonlocal
operator $\mathcal{L}$ in the zero-dispersion limit is to provide a new
asymptotic formula for the phase constants.  It is difficult to motivate
the formula as it arises from the analysis of the inverse problem that
we will describe in the next section, but it is nonetheless quite easy
to present.  If $\lambda<0$ is an eigenvalue of $\mathcal{L}$ with potential
$u$ given by an admissible initial condition $u_0$, then our approximation
to the corresponding phase constant is given in terms of the
turning points $x_\pm(\lambda)$ as follows:
\begin{equation}
\label{formula for gamma in terms of x}
\gamma\approx \gamma(\lambda):=
-\frac{1}{2}(x_{+}(\lambda)+x_{-}(\lambda)),
\quad -L\le k<0.
\end{equation}
\begin{remark}
  Our choice of $\gamma(\lambda)$ in terms of $u_0$ is specifically
  designed to ensure the convergence of $\ut(x,t)$ (to be
  defined precisely in Definition~\ref{def:modified} below) at $t=0$
  to the given $\epsilon$-independent initial condition $u_0$.
\end{remark}
\subsection{Modification of the  Cauchy data}
Based on the above considerations, we may now make very precise
definitions of formal (not rigorously justified) approximations
of the scattering data corresponding to an admissible condition
$u_0$.  The first approximation is to neglect the reflection coefficient
by setting
\begin{equation}
\tilde{\beta}(\lambda):=0,\quad\lambda>0.
\end{equation}
Next we define the exact number of approximate eigenvalues (hopefully also
the approximate number of exact eigenvalues) by setting
\begin{equation}
N(\epsilon) := \left\lfloor \frac{M}{\epsilon}\right
\rfloor,
\label{eq:Numberofeigs}
\end{equation}
which in particular implies that
\begin{equation}
\lim_{\epsilon\downarrow 0}\epsilon N(\epsilon) = M.
\label{eq:epsilonNepsilon}
\end{equation}
Then we define approximations to the eigenvalues themselves as an
ordered set of numbers
$\{\lamt_n\}_{n=1}^{N(\epsilon)}\subset(-L,0)$ obtained by
quantizing the Matsuno eigenvalue density given by
\eqref{eq:Fadmissibledef}:
\begin{equation}
\label{value of eigenvalue given by a fomula}
\int_{-L}^{\lamt_{n}}F(\lambda)\,d\lambda=\epsilon \left(n-\frac{1}{2}
\right),\quad
n=1,2,\cdots N(\epsilon).
\end{equation}
Finally, we define approximations to the corresponding phase constants
as numbers $\{\gamt_n\}_{n=1}^{N(\epsilon)}$ given precisely by
\begin{equation}
\gamt_n:=\gamma(\lamt_n),\quad n=1,\dots,N(\epsilon).
\label{eq:gammandef}
\end{equation}
where $\gamma(\cdot)$ is defined by \eqref{formula for gamma in terms of x}.

Now in our analysis of the Cauchy problem for the BO equation with
admissible initial data $u_0$ we take a sideways step that is not
\emph{a priori} justified:  we simply replace the true solution $\ueps(x,t)$
of the Cauchy problem with a family $\ut(x,t)$ of exact solutions
of the BO equation \eqref{BOE} with the property that for each $\epsilon>0$
the scattering data for $\ut(x,t)$ at time $t=0$ is \emph{exactly}
the approximate scattering data just defined.  This step was also an important
part of the method of Lax and Levermore \cite{Lax 1983-1}.  We formalize
this modification of the initial data in the following definition.

\begin{definition}
\label{def:modified}
Let $u_0$ be an admissible initial condition.  Then, by $\ut(x,t)$
we mean the exact solution of the BO equation \eqref{BOE} given for
each $\epsilon>0$ by the reflectionless inverse-scattering formula 
\begin{equation}
\ut(x,t):=2\epsilon\frac{\partial}{\partial x}
\Im\{\log(\tilde{\tau}_\epsilon(x,t))\},
\end{equation}
where
\begin{equation}
\tilde{\tau}_\epsilon(x,t):=\det\left(\mathbb{I}+i\epsilon^{-1}\tilde{\mat{A}}_\epsilon\right)
\end{equation}
and where $\tilde{\mat{A}}_\epsilon=\tilde{\mat{A}}_\epsilon(x,t)$ is an
$N(\epsilon)\times N(\epsilon)$ Hermitean matrix with elements
\begin{equation}
\tilde{A}_{nm}:=\frac{2i\epsilon\sqrt{\lamt_n\lamt_m}}{\lamt_n-\lamt_m},\quad
n\neq m
\end{equation}
and
\begin{equation}
\tilde{A}_{nn}:=-2\lamt_n(x+2\lamt_nt+\gamt_n)=
-2\lamt_n(x+2\lamt_n+\gamma(\lamt_n)).
\end{equation}
Here
the number $N(\epsilon)$ is defined by \eqref{eq:Numberofeigs} and
the components of the scattering data $\{\lamt_n\}_{n=1}^{N(\epsilon)}$ and 
$\{\gamt_n\}_{n=1}^{N(\epsilon)}$ are given explicitly by 
\eqref{value of eigenvalue given by a fomula} and \eqref{eq:gammandef} 
respectively.
\end{definition}

While it is not the case that $\ut(x,0)=u_0(x)$ in general, the
relevance of this definition in connection with the Cauchy problem
with initial condition $u_0$ is a consequence of
Corollary~\ref{corr:strong} which guarantees convergence in the mean square 
sense of $\ut(\cdot,0)$ to $u_0(\cdot)$ as
$\epsilon\downarrow 0$.  

The proof of Theorem~\ref{MainTheorem} will be given below in
\S\ref{sec:inverse}.  Before embarking on that we note that
Definition~\ref{def:admissible} implies a number
of properties of the functions $F$ and $\gamma$ 
that will be useful later, so we take the opportunity to record these here.  
Note that $F$ and $\gamma$ will frequently occur in the context
of the following functions:
\begin{equation}
D(\lambda;x,t):=-2\lambda(x+2\lambda t+\gamma(\lambda)),\quad -L<\lambda<0,
\label{eq:Dfuncdef}
\end{equation}
and 
\begin{equation}
\varphi(\lambda):=\sqrt{-\lambda F(\lambda)},\quad -L<\lambda<0.
\label{eq:varphidef}
\end{equation}
\begin{lemma}
\label{Lemma for function F}
Let $u_0$ be an admissible initial condition with decay exponent $q>1$,
and let $F:[-L,0)\to\mathbb{R}$ be defined by \eqref{eq:Fadmissibledef}
and $\gamma:[-L,0)\to\mathbb{R}$ be defined by 
\eqref{formula for gamma in terms of x}.  Then $F$ and $\gamma$ both belong
to $C^{(1)}(-L,0)$ and $F$ and $F'$ are strictly positive on this open interval.
Also, there exists a sufficiently small constant $\delta>0$ and positive
constants $C_{-L}$ and $C_0$ such that
\begin{equation}
\label{Lemma for function F equation 2}
\frac{1}{2}C_{-L}\sqrt{L+\lambda}<F(\lambda)<C_{-L}\sqrt{L+\lambda},
\end{equation}
and
\begin{equation}
\label{Lemma for function F equation 1}
\frac{1}{4}\frac{C_{-L}}{\sqrt{L+\lambda}}<F'(\lambda)<
\frac{1}{2}\frac{C_{-L}}{\sqrt{L+\lambda}}
\end{equation}
both hold for $-L<\lambda<-L+\delta$, while 
\begin{equation}
\label{Lemma for function F equation 4}
\frac{1}{2}C_{0}(-\lambda)^{-1/q}<F(\lambda)<C_{0}(-\lambda)^{-1/q},
\end{equation}
and 
\begin{equation}
\label{Lemma for function F equation 3}
\frac{1}{2}\frac{C_{0}}{q}(-\lambda)^{-1/q-1}<F'(\lambda)<
\frac{C_{0}}{q}(-\lambda)^{-1/q-1}
\end{equation}
both hold for $-\delta<\lambda<0$.  
Also, 
\begin{equation}
|\gamma(\lambda)+x_0|\le \pi F(\lambda)\quad\text{and}\quad
|\gamma'(\lambda)|\le \pi F'(\lambda),\quad -L\le\lambda <0,
\label{eq:gammaFbounds}
\end{equation}
inequalities that when combined with 
\eqref{Lemma for function F equation 2}--\eqref{Lemma for function F equation 3}
imply obvious upper bounds for $|\gamma(\lambda)+x_0|$
and $|\gamma'(\lambda)|$.

In particular, these estimates
show that $F(\lambda)$ is integrable, and $\varphi(\lambda)$ and
$D(\lambda;x,t)$ (and hence also $\lambda x_\pm(\lambda)$) 
are bounded, and that with $\sigma=\min(\tfrac{1}{2},1-\tfrac{1}{q})\in (0,1)$,
$\varphi(\cdot)$ is H\"older continuous with exponent $\sigma/2$ while
$D(\cdot;x,t)$ is H\"older continuous with exponent $\sigma$ uniformly
for $(x,t)$ in compact sets, on $(-L,0)$.
\end{lemma}
\begin{proof}
The turning points $x_{\pm}(\lambda)$ are clearly of class $C^{(1)}(-L,0)$,
by definition $x_+(\lambda)>x_0>x_-(\lambda)$ on this open interval,
and moreover 
$x_{ +}(\lambda)$ is strictly increasing while $x_{-}(\lambda)$ is 
strictly decreasing on
 $(-L, 0)$.   These facts immediately imply the desired
basic smoothness properties of $F$ and $\gamma$, and the positivity 
and monotonicity of $F$, as well as the inequalities \eqref{eq:gammaFbounds}.  

Since $u_{0}(x_{0})=L$ and $u_0'(x_{0})=0$, the
$C^{(2)}(\mathbb{R})$ function $u_0$ satisfies
\begin{equation}
\lim_{x\to x_{0}}\frac{u_{0}(x)-L}{(x-x_{0})^{2}}=\frac{u''_{0}(x_{0})}{2}
\quad\text{and}\quad
\lim_{x\to x_{0}}\frac{u'_{0}(x)}{x-x_{0}}=u''_{0}(x_{0}).
\end{equation}
Using these together with the inequality $u''_{0}(x_{0})<0$, the
definition of $x_\pm(\lambda)$ as branches of the inverse function of
$u_0$ shows that 
\begin{equation}
\lim_{\lambda\downarrow -L}\frac{\pm(x_{\pm}(\lambda)-x_{0})}
{\sqrt{L+\lambda}}=\sqrt{\frac{2}{-u''_{0}(x_{0})}}\quad\text{and}\quad
\lim_{\lambda\downarrow -L}\pm x'_{\pm}(\lambda)\sqrt{L+\lambda}=
\sqrt{\frac{1}{-2u''_{0}(x_{0})}}.
\end{equation}
Using these relations in \eqref{eq:Fadmissibledef} and 
\eqref{formula for gamma in terms of x} establishes the existence of the limits
\begin{equation}
\lim_{\lambda\downarrow -L}\frac{F(\lambda)}{\sqrt{L+\lambda}}=\frac{1}{\pi}\sqrt{\frac{2}{-u''_{0}(x_{0})}}\quad\text{and}\quad
\lim_{\lambda\downarrow -L}F'(\lambda)\sqrt{L+\lambda}
=\frac{1}{2\pi}\sqrt{\frac{2}{-u''_{0}(x_{0})}},
\end{equation}
which prove the two-sided estimates 
\eqref{Lemma for function F equation 2} and 
\eqref{Lemma for function F equation 1}.

Next, note that the decay conditions \eqref{u0 infinity boundary condition} 
and \eqref{u0 infinity boundary condition1} for 
$u_0$ and its derivative together imply that 
\begin{equation}
\label{Lemma for function F proof 1}
\lim_{\lambda\uparrow 0}x_{\pm}(\lambda)(-\lambda)^{\frac{1}{q}}=
\pm\left(\mp\frac{C_{\pm}}{q}\right)^{\frac{1}{q}}\quad\text{and}\quad
\lim_{\lambda\uparrow 0}x_\pm'(-\lambda)^{\tfrac{1}{q}+1} = \pm\frac{1}{q}\left(\mp\frac{C_\pm}{q}\right)^{\tfrac{1}{q}},
\end{equation}
where $\mp C_\pm$ are the positive constants 
in \eqref{u0 infinity boundary condition} 
and \eqref{u0 infinity boundary condition1}.
It follows from \eqref{eq:Fadmissibledef} that
\begin{equation}
\label{Lemma for function F proof 3}
\lim_{\lambda\uparrow 0}F(\lambda)(-\lambda)^{\frac{1}{q}}=
\frac{1}{2\pi}\left[\left(-\frac{C_{+}}{q}\right)^{\frac{1}{q}}+\left(\frac{C_{-}}{q}\right)^{\frac{1}{q}}\right],
\end{equation}
which proves \eqref{Lemma for function F equation 4}
and 
\begin{equation}
\label{Lemma for function F proof 7}
\lim_{\lambda\uparrow 0}F'(\lambda)(-\lambda)^{\frac{1}{q}+1}=
\frac{1}{2\pi q}\left[\left(-\frac{C_{+}}{q}\right)^{\frac{1}{q}}+\left(\frac{C_{-}}{q}\right)^{\frac{1}{q}}\right],
\end{equation}
which proves \eqref{Lemma for function F equation 3}.  
\end{proof}

\section{The Inverse-Scattering Problem in the Zero-Dispersion Limit}
\label{sec:inverse}
In this section, we provide the proof of Theorem~\ref{MainTheorem}.
\subsection{Basic strategy.  Outline of proof}
According to Definition~\ref{def:modified}, $\ut(x,t)$ is expressed
in terms of the determinant $\tilde{\tau}_\epsilon$ as follows:
\begin{equation}
\ut(x,t)=\frac{\partial \Ut}{\partial x}(x,t),\quad
\Ut(x,t)=2\epsilon\Im\{\log(\tilde{\tau}_\epsilon(x,t))\}.
\label{eq:U}
\end{equation}
As the logarithm of a complex-valued
quantity is involved, $\Ut(x,t)$ is only defined modulo $4\pi\epsilon$
for each $(x,t)$, and naturally one should choose the appropriate branch
for each $(x,t)$ to achieve continuity.  We do this concretely in
equation \eqref{small dispersion matrix calculation 3_1} below.

At this very early point our
analysis must take a very different path than that followed by Lax and
Levermore \cite{Lax 1983-1} in their study of the zero-dispersion
limit for the KdV equation.  Indeed, the expansion of $\tilde{\tau}_\epsilon$ 
in principal minors that is at the heart of the Lax-Levermore method would
be a poor choice in this situation.  One reason for this is simply that the
principal-minors expansion of $\tilde{\tau}_\epsilon(x,t)$ consists of
complex-valued terms of indefinite phase, so the sum cannot be easily
estimated by its largest term.  But a more important reason is that
the formula \eqref{eq:U} for $\Ut(x,t)$ involves not
$\log(\tilde{\tau}_\epsilon)$ but rather $\Im\{\log(\tilde{\tau}_\epsilon)\}$, 
that is, we require an
estimate of the \emph{phase} of the determinant and we are not
interested in its magnitude.

So instead of expanding the determinant as a sum, we write it as a product.
Let $\{\alpha_n\}_{n=1}^{N(\epsilon)}$ be the real eigenvalues
of $\tilde{\mat{A}}_\epsilon(x,t)$.  Then the corresponding eigenvalues of 
$\mathbb{I}+i\epsilon^{-1}\tilde{\mat{A}}_\epsilon(x,t)$ are 
of course $\{1+i\epsilon^{-1}\alpha_n\}_{n=1}^{N(\epsilon)}$, so we may
expand $\tilde{\tau}_\epsilon$ as a product over eigenvalues in the form:
\begin{equation}
\label{small dispersion matrix calculation 2}
\tilde{\tau}_\epsilon(x,t)=\prod_{n=1}^{N(\epsilon)}
\left(1+i\epsilon^{-1}\alpha_{n}\right).
\end{equation}
This yields a suggestive formula for $\Ut(x,t)$ in terms of the
eigenvalues of $\tilde{\mat{A}}_\epsilon$:
\begin{equation}
\label{small dispersion matrix calculation 3_1} \Ut(x,t):=
\epsilon\sum_{n=1}^{N(\epsilon)}2\arctan\left(\epsilon^{-1}\alpha_{n}\right).
\end{equation}
Here $-\pi/2<\arctan(\cdot)<\pi/2$, so in particular by this
definition we have made an unambiguous choice of the branch of the
logarithm.  This formula seems at first not to be of much use because,
unlike the principal minor determinants in the Lax-Levermore method
which can be written explicitly in terms of the matrix elements, the
eigenvalues of $\tilde{\mat{A}}_\epsilon$ are only implicitly known.  However,
numerical experiments suggest that some structure emerges in the limit
$\epsilon\downarrow 0$.  Indeed, the plots shown in
Figure~\ref{fig:Hist}
\begin{figure}[h]
\begin{center}
\includegraphics{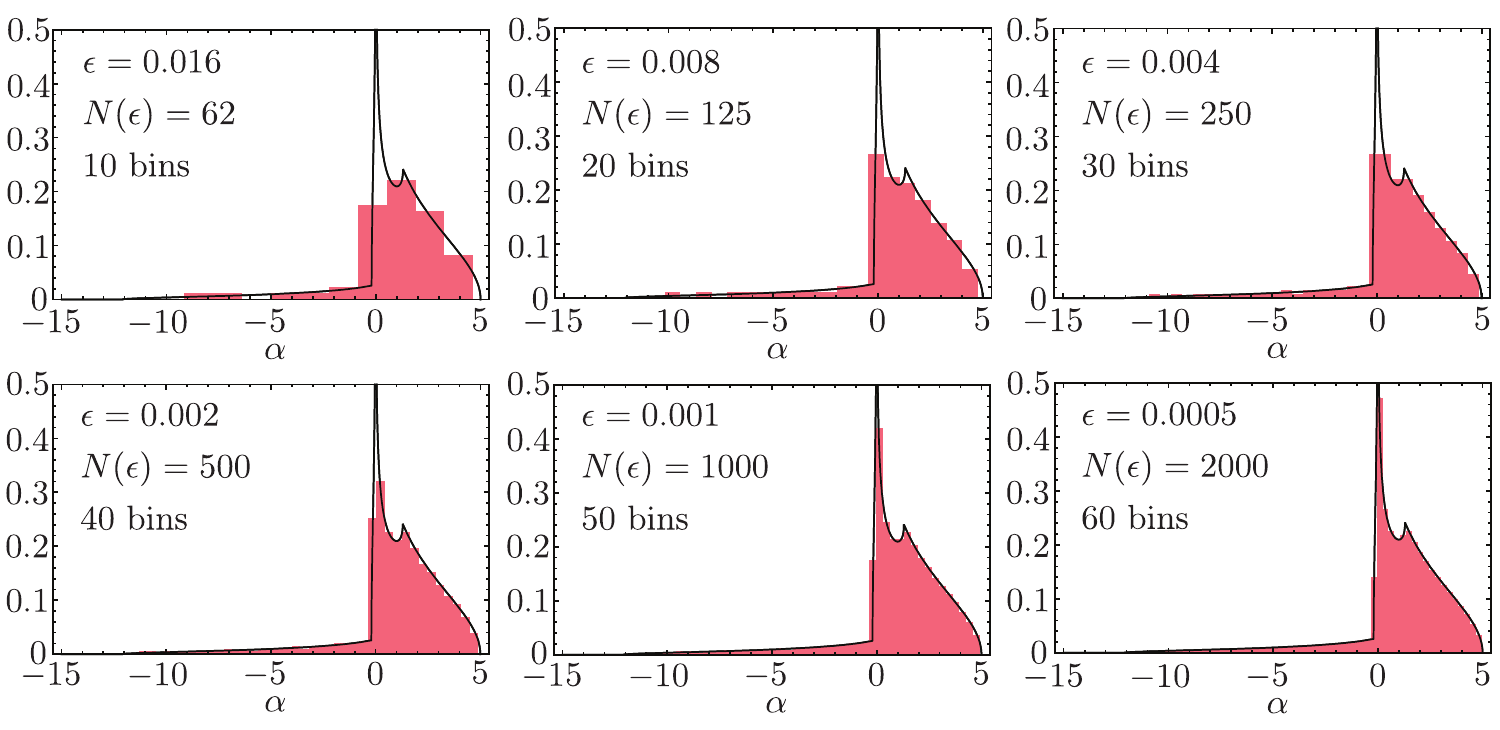}
\end{center}
\caption{\emph{Histograms of eigenvalues of $\tilde{\mathbf{A}}_\epsilon$ 
corresponding
to the initial condition $u_0(x):=2(1+x^2)^{-1}$, $x=5$, and $t=2$, 
normalized to
have total area $M=1$, compared with the density $G(\alpha;x,t)$
of the limiting absolutely continuous measure $\mu$.}}
\label{fig:Hist}
\end{figure}
provide good evidence that the normalized (to mass $M$) 
counting measures $\mu_\epsilon$
given for $\epsilon>0$ by 
\begin{equation}
\mu_{\epsilon}:=\frac{M}{N(\epsilon)}\sum_{n=1}^{N(\epsilon)}\delta_{\alpha_n},\quad
\text{$\{\alpha_n\}_{n=1}^{N(\epsilon)}$ eigenvalues of $\tilde{\mat{A}}_{\epsilon}$}
\label{eq:ncm}
\end{equation}
might converge in some sense to
a measure $\mu$ having a density $G(\alpha;x,t)$.
This convergence suggests further that the formula
\eqref{small dispersion matrix calculation 3_1}
could be interpreted as a Riemann sum, for the integral of 
$\pi\,\mathrm{sgn}(\alpha)$ 
(the pointwise limit as $\epsilon\downarrow 0$ of the summand)
against the limiting measure $\mu$.  We will prove that indeed
$\Ut(x,t)$ converges, uniformly with respect to $x$ and $t$ in
compact sets, to a limit function $U(x,t)$ given by such an integral in
the limit $\epsilon\downarrow 0$.  

To obtain an effective formula for $U(x,t)$ we need to analyze the asymptotic
behavior of the measures $\mu_\epsilon$.
This part of our analysis is modeled after the
work of Wigner \cite{Wigner55,Wigner58} 
on the statistical distribution of eigenvalues of random Hermitian
matrices with independent and identically distributed matrix elements.
Like Wigner, we use the method of moments because 
while the measures themselves are not easy to
express in terms of the matrix elements, their moments are:
\begin{equation}
\int_\mathbb{R}\alpha^p\,d\mu_\epsilon(\alpha) = 
\frac{M}{N(\epsilon)}\sum_{n=1}^{N(\epsilon)}\alpha_n^p
= \frac{M}{N(\epsilon)}\mathrm{tr}(
\tilde{\mat{A}}_\epsilon^p),\quad p=0,1,2,\dots.
\label{eq:trace}
\end{equation}
We prove the existence of the limit of the right-hand side
in equation \eqref{eq:trace} as $\epsilon\downarrow 0$ for every $p$
using the fact that for small $\epsilon$ the matrix $\tilde{\mat{A}}_\epsilon$
concentrates near the diagonal, where it can be approximated by the
product of a diagonal matrix and the Toeplitz matrix corresponding to
the symbol $f(\theta):=i(\pi-\theta)$, $0<\theta<2\pi$ (of singular
Fisher-Hartwig type due to jump discontinuities).  The result of
this asymptotic analysis of moments is the following Proposition, 
the proof of which will be given below in \S\ref{sec:moments}.
\begin{prop}
\label{asymptotic preperties for trace 3} For each nonnegative
integer $p$, 
\begin{equation}
\label{asymptotic preperties for trace 3 1}
\lim_{\epsilon\downarrow 0}\int_{\mathbb{R}}\alpha^p\,d\mu_\epsilon(\alpha)=Q_p,
\end{equation}
with the limit being uniform with respect to $(x,t)$ in any compact set,
where
\begin{equation}
Q_p:=
\frac{1}{2\pi(p+1)}\int_{-L}^0
\left[\left(x+2\lambda t-x_-(\lambda)\right)^{p+1}-
\left(x+2\lambda t -x_+(\lambda)\right)^{p+1}\right](-2\lambda)^p\,d\lambda.
\label{eq:hp}
\end{equation}
\end{prop}

Given these limiting moments, the next task is to establish the existence
of a corresponding limiting measure $\mu$ with these moments, and to prove
the existence of the limit $\Ut(x,t)\to U(x,t)$.  A remarkable feature
of this analysis is that the solution of the moment problem for $\mu$
is carried out by virtually the same procedure as Matsuno used to obtain
the function $F(\lambda)$ from $u_0$ (see \S\ref{sec:MatsunosMethod}).
Our result is the following Proposition, that will be proved in all
details in \S\ref{sec:measures}.
\begin{prop}
\label{limit of sx theorem} Uniformly for $(x,t)$ in compact sets,
\begin{equation}
\label{limit of sx 1} \lim_{\epsilon \downarrow
0}\Ut(x,t)=U(x,t),
\end{equation}
where
\begin{equation}
U(x,t):=\int_\mathbb{R}\pi
\,\mathrm{sgn}(\alpha)\,d\mu(\alpha)
\label{eq:Uformula}
\end{equation}
and where $\mu$ is an absolutely continuous measure
of mass $M$ with density $G(\alpha;x,t)$, and
\begin{equation}
\label{limit of sx
2}G(\alpha;x,t):=-\frac{1}{4\pi}\int_{-L}^{0}
\chi_{[-2\lambda(x+2\lambda t-x_+(\lambda)),-2\lambda(x+2\lambda t-x_-(\lambda))]}(\alpha)\frac{d\lambda}
{\lambda}.
\end{equation}
Here, $\chi_{[a,b]}(z)$ denotes the indicator function of the interval $[a,b]$.
\end{prop}
The limiting measure
$\mu$ is the closest analogue in the zero-dispersion theory of the
BO equation of the equilibrium (or extremal) measure arising in the Lax-Levermore theory of the KdV equation.
But a significant difference is that in this case the measure $\mu$
is specified \emph{explicitly} rather than implicitly as the solution of a
variational problem.

The region of integration in the double integral obtained by combining
\eqref{limit of sx 2} with \eqref{eq:Uformula} is illustrated for
three different values of $(x,t)$ in Figure~\ref{fig:IntRegion}.
\begin{figure}[h]
\begin{center}
\includegraphics{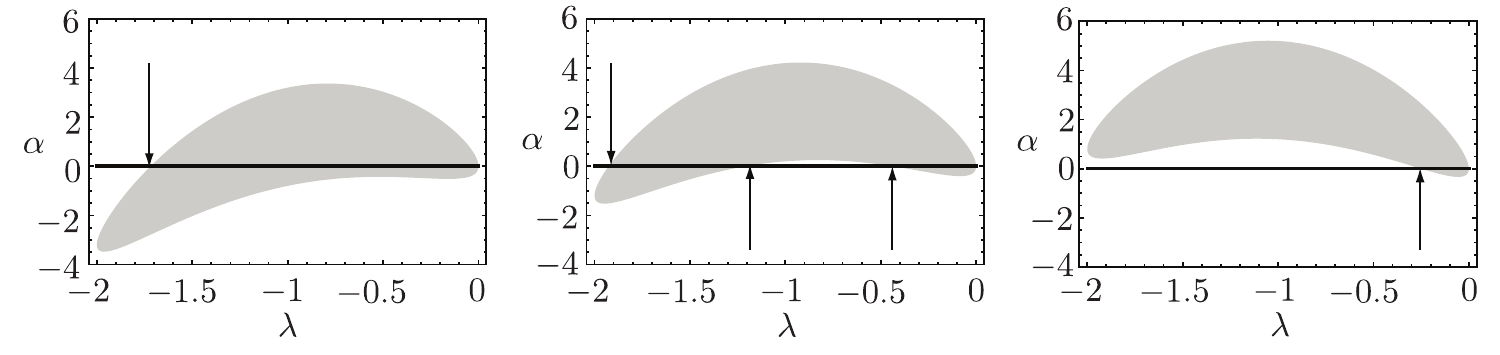}
\end{center}
\caption{\emph{The region of integration 
$-2\lambda(x+2\lambda t-x_+(\lambda))<\alpha<
-2\lambda(x+2\lambda t-x_-(\lambda))$
for $u_0(x)=2(1+x^2)^{-1}$ with $t=0.7$.  Left: $x=2$ (to the left of
the
    oscillatory region for $\ueps(x,t)$).  
Center: $x=2.5$ (within the oscillatory
    region for $\ueps(x,t)$).  
Right: $x=3$ (to the right of the oscillatory region for $\ueps(x,t)$).
    The line $\alpha=0$ of discontinuity of the integrand is
    superimposed, and the intersections of the boundary with
    this line are indicated with arrows.}}
\label{fig:IntRegion}
\end{figure}
The points where the boundary curves of this region intersect the line 
$\alpha=0$ (where the integrand is discontinuous) obviously will play an
important role in the differentiation of $U(x,t)$ with respect to $x$.
Moreover, these intersection points correspond (simply by changing the sign)
to the branches of the multivalued solution of Burgers' equation with initial
data $u_0$.  This explains their appearance in the formula for the weak limit
of $\ueps(x,t)$.  All details of this calculation will be given 
in \S\ref{sec:differentiation}, which will complete the proof of 
Theorem~\ref{MainTheorem}.

\subsection{Asymptotics of traces of powers of $\tilde{\mat{A}}_\epsilon$.  Proof
of Proposition~\ref{asymptotic preperties for trace 3}}
\label{sec:moments}
The definition \eqref{value of eigenvalue given by a fomula} implies that
where $F(\lambda)$ is bounded and bounded away from zero, the numbers
$\{\lamt_n\}_{n=1}^{N(\epsilon)}$ are locally nearly equally spaced, but they
are more dilute near the ``soft edge'' of the spectrum $\lambda=-L$ and
more dense near the ``hard edge'' of the spectrum $\lambda=0$.  Taking
into account the soft edge behavior we may obtain a uniform estimate:
\begin{lemma}
\label{lem:uniformspacing}
There is a constant $C_\lambda>0$ independent of $\epsilon$ such that
\begin{equation}
|\lamt_n-\lamt_m|\le C_\lambda\epsilon^{2/3}|n-m|^{2/3}
\end{equation}
holds for all $n$ and $m$ between $1$ and $N(\epsilon)$.
\end{lemma}
\begin{proof}
Since $F$ is a monotone increasing function with $F(-L)=0$, it is bounded 
away from zero except in a right-neighborhood of $\lambda=-L$.  Using
the lower bound given in \eqref{Lemma for function F equation 2} from
Lemma~\ref{Lemma for function F} we obtain a lower bound $F(\lambda)\ge
C\sqrt{L+\lambda}$ valid uniformly for $-L<\lambda<0$ with $0<C\le C_{-L}/2$.
Then, using the definition \eqref{value of eigenvalue given by a fomula}
we have (assuming $n\ge m$ without loss of generality)
\begin{equation}
\epsilon|n-m|=\int_{\lamt_m}^{\lamt_n}F(\lambda)\,d\lambda\ge C
\int_{\lamt_m}^{\lamt_n}\sqrt{L+\lambda}\,d\lambda \ge
C\int_{0}^{|\lamt_n-\lamt_m|}\sqrt{\xi}\,d\xi  
=\frac{2C}{3}|\lamt_n-\lamt_m|^{3/2},
\end{equation}
so the desired inequality follows with $C_\lambda:=(2C/3)^{-2/3}$.
\end{proof}

We decompose the matrix $\tilde{\mat{A}}_\epsilon$ into a sum 
$\tilde{\mat{A}}_\epsilon=\mat{D}+\mat{H}$ of its diagonal part
\begin{equation}
\mat{D}:=\mathrm{diag}(D_1,D_2,\dots,D_{N(\epsilon)}),\quad
D_k:=D(\lamt_k;x,t),
\end{equation}
where $D(\lambda;x,t)$ is defined by \eqref{eq:Dfuncdef},
and its off-diagonal part $\mat{H}$ whose matrix elements are given by
\begin{equation}
\label{small dispersion matrix calculation 7}
H_{nm}=\frac{2i\epsilon
\sqrt{\lamt_{n}\lamt_{m}}}{\lamt_{n}-\lamt_{m}},\quad 
\text{for $n\neq m$, and $H_{nn}=0$.}
\end{equation}
We also will soon need the quantities $\{\varphi_n\}_{n=1}^{N(\epsilon)}$ 
defined by
\begin{equation}
\varphi_n:=\varphi(\lamt_n),\quad n=1,\dots,N(\epsilon),
\label{eq:varphikdef}
\end{equation}
where $\varphi(\lambda)$ is given by \eqref{eq:varphidef}.

\begin{lemma}
There is a constant $C_\varphi>0$ and for each $R>0$ there is a constant 
$C_{D,R}>0$ such that
\begin{equation}
|\varphi_n|\le C_\varphi
\label{eq:varphinbound}
\end{equation}
and 
\begin{equation}
\sup_{x^2+t^2\le R^2}|D_n|\le C_{D,R}
\label{eq:Dnbound}
\end{equation}
both hold for all $\epsilon>0$ and all $n$ between $1$ and $N(\epsilon)$.  Also,
\begin{equation}
|\varphi_n-\varphi_m|\le C_\varphi\epsilon^{\sigma/3}|n-m|^{\sigma/3}
\end{equation}
and
\begin{equation}
\sup_{x^2+t^2\le R^2}|D_n-D_m|\le C_{D,R}\epsilon^{\sigma/3}|n-m|^{\sigma/3}
\label{eq:Ddiffbound}
\end{equation}
both hold for all $\epsilon>0$ and for all $n$ and $m$ between $1$ and
$N(\epsilon)$.  Here $\sigma$ is the positive H\"older exponent of
Lemma~\ref{Lemma for function F}.
\label{lem:D}
\end{lemma}
\begin{proof}
This is an easy consequence of the H\"older continuity of $\varphi(\cdot)$ and
$D(\cdot;x,t)$ guaranteed by Lemma~\ref{Lemma for function F}, and of 
the spacing estimate
for $\{\lamt_k\}_{k=1}^{N(\epsilon)}$ given in Lemma~\ref{lem:uniformspacing}.
In fact, since $D$ is H\"older continuous with exponent $\sigma$ while
$\varphi$ has exponent $\sigma/2$ the most natural bound for $|D_n-D_m|$
is proportional to $\epsilon^{2\sigma/3}|n-m|^{2\sigma/3}$, and to obtain
\eqref{eq:Ddiffbound} we use the fact that 
$\epsilon|n-m|\le 2\epsilon N(\epsilon)$ 
is uniformly bounded to reduce the exponent to $\sigma/3$.
\end{proof}

\begin{lemma}
There is a constant $C_H>0$ such that
\begin{equation}
\left|(n-m)H_{nm}\right|\le C_H
\label{eq:Hnmbound}
\end{equation}
and
\begin{equation}
\left|(n-m)H_{nm}-2i\varphi_n\varphi_m\right|\le C_H\epsilon^{\sigma/3}
|n-m|^{\sigma/3}
\label{eq:Hdiffbound}
\end{equation}
both hold for all $\epsilon>0$ and all $n\neq m$ between $1$ and $N(\epsilon)$.
Again, $\sigma>0$ is the H\"older exponent of Lemma~\ref{Lemma for function F}.
\label{lem:Hnm}
\end{lemma}
\begin{proof}
Suppose without loss of generality that $n>m$, implying that 
$\lamt_m<\lamt_n<0$.  Then
\begin{equation}
\begin{split}
-i(n-m)H_{nm}&=2\sqrt{-\lamt_n}\sqrt{-\lamt_m}
\frac{\epsilon(n-m)}{\lamt_n-\lamt_m}\\
&\le ([-\lamt_n]+[-\lamt_m])\frac{\epsilon(n-m)}{\lamt_n-\lamt_m}\\
&= \epsilon (n-m) -2\lamt_n\frac{\epsilon(n-m)}{\lamt_n-\lamt_m}.
\end{split}
\end{equation}
Now, recalling the definition \eqref{value of eigenvalue given by a fomula}
of the numbers $\{\lamt_k\}_{k=1}^{N(\epsilon)}$ and applying the Mean Value
Theorem we may write the latter difference quotient as $F(\xi)$ for 
some $\xi$ with $\lamt_m\le\xi\le\lamt_n$, and since $F$ is increasing
we have $F(\xi)\le F(\lamt_n)$, so
\begin{equation}
-i(n-m)H_{nm}\le 
2\epsilon(n-m)-2\lamt_nF(\lamt_n) = 2\epsilon(n-m)+2\varphi_n^2,
\label{eq:nmHnmupper}
\end{equation}
where we have also replaced $\epsilon(n-m)$ with $2\epsilon(n-m)$.
On the other hand,
we may write
\begin{equation}
-i(n-m)H_{nm}= 2\epsilon(n-m)
\frac{\sqrt{-\lamt_n}\sqrt{-\lamt_m}+\lamt_m}{\lamt_n-\lamt_m}
- 2\lamt_m\frac{\epsilon(n-m)}{\lamt_n-\lamt_m}.
\end{equation}
Again the difference quotient may be replaced by $F(\xi)\ge F(\lamt_m)$, and
since
\begin{equation}
\frac{\sqrt{-\lamt_n}\sqrt{-\lamt_m}}{\lamt_n-\lamt_m} = -\frac{\sqrt{-\lamt_m}}{\sqrt{\lamt_n}+\sqrt{\lamt_m}}\ge -1,
\end{equation}
we obtain
\begin{equation}
-i(n-m)H_{nm}\ge -2\epsilon(n-m)-2\lamt_mF(\lamt_m) = -2\epsilon(n-m)+2\varphi_m^2.
\label{eq:nmHnmlower}
\end{equation}
Combining \eqref{eq:nmHnmupper} and \eqref{eq:nmHnmlower} gives
\begin{equation}
\left|(n-m)H_{nm}-2i\varphi_n\varphi_m\right|\le 2\epsilon|n-m|+2
\max\{\varphi_n,\varphi_m\}|\varphi_n-\varphi_m|,
\end{equation}
and then applying Lemma~\ref{lem:D} we obtain
\begin{equation}
\left|(n-m)H_{nm}-2i\varphi_n\varphi_m\right|\le 2\epsilon |n-m| + 
2C_\varphi^2\epsilon^{\sigma/3}|n-m|^{\sigma/3}.
\end{equation}
Now, $0\le \epsilon |n-m|\le 2\epsilon N(\epsilon)$, and this upper
bound has a limit as $\epsilon\downarrow 0$, so $\epsilon |n-m|$ is
nonnegative and bounded.  Since $\sigma\le 3$ we have therefore proved 
\eqref{eq:Hdiffbound}.  Since $\varphi_n\varphi_m$ and $\epsilon |n-m|$ are
bounded, \eqref{eq:Hnmbound} then follows from \eqref{eq:Hdiffbound}.
\end{proof}

For any nonnegative integer power $p$, the $p$th moment of the measure
$\mu_\epsilon$
can be written in terms of $\mat{D}$ and $\mat{H}$ with the use of
\eqref{eq:trace}:
\begin{equation}
\label{small dispersion matrix calculation 8}
\int_\mathbb{R}\alpha^pd\mu_\epsilon(\alpha) = 
\sum_{j=0}^{p}Z_{pj},
\end{equation}
where $Z_{pj}$ contains the contribution to the trace coming from products
of matrices involving exactly $j$ factors of $\mat{H}$:
\begin{equation}
Z_{pj}:=
\frac{M}{N(\epsilon)}
\mathop{\sum_{d_{1}+d_{2}+\cdots+d_{s}=p-j}}_{h_{1}+h_{2}+\cdots+h_{s}=j}
\mathrm{tr}\left(
\mat{D}^{d_1}\mat{H}^{h_1}\cdots \mat{D}^{d_s}\mat{H}^{h_s}\right),
\label{eq:Zpj}
\end{equation}
and where $d_{1}\ge 0$ and $h_{s} \geq 0$, while $d_k>0$ for $2\le k\le s$
and $h_k>0$ for $1\le k\le s-1$.  Since $p$ is a fixed number, it will suffice
to compute the limit of $Z_{pj}$ as $\epsilon\downarrow 0$ for $j=0,\dots,p$.
Actually, it will be enough to consider even values of $j$ as the following
result shows.

\begin{lemma}
\label{asymptotic preperties for trace 1} If $j$ is an odd number, then
$Z_{pj}=0$.
\end{lemma}
\begin{proof}
Since $\mathrm{tr}(\mat{M})=
\mathrm{tr}(\mat{M}^{\mathsf{T}})$ for all 
square matrices $\mat{A}$, 
\begin{equation}
\label{asymptotic preperties for trace 1 1}
\begin{split}
\frac{N(\epsilon)}{M}Z_{pj}
=&\mathop{\sum_{d_{1}+d_{2}+\cdots+d_{s}=p-j}}_{h_{1}+h_{2}+\cdots+h_{s}=j}\mathrm{tr}\left(\left(
\mat{D}^{d_1}\mat{H}^{h_1}\cdots \mat{D}^{d_s}\mat{H}^{h_s}\right)^{\mathsf{T}}\right)
\\=&(-1)^j\mathop{\sum_{d_{1}+d_{2}+\cdots+d_{s}=p-j}}_{h_{1}+h_{2}+\cdots+h_{s}=j}
\mathrm{tr}\left(
\mat{H}^{h_s}\mat{D}^{d_s}\cdots \mat{H}^{h_1}\mat{D}^{d_1}\right)
\end{split}
\end{equation}
where in the second line we have used the facts that
$\mat{D}^{\mathsf{T}}=\mat{D}$ and $\mat{H}^{\mathsf{T}}=-\mat{H}$.
By relabeling the terms in the sum we therefore obtain
\begin{equation}
\frac{N(\epsilon)}{M}Z_{pj}=(-1)^j\frac{N(\epsilon)}{M}Z_{pj}.
\end{equation}
Since $N(\epsilon)>0$ and $M<\infty$, the desired result follows.
\end{proof}

An important role will be played below by the Toeplitz (discrete
convolution) operator
$\mathcal{T}_f:\ell^2(\mathbb{Z})\to\ell^2(\mathbb{Z})$ defined by 
\begin{equation}
\left(\mathcal{T}_fc\right)_n := \sum_{m\in\mathbb{Z}}f_{n-m}c_m,
\quad\{c_m\}_{m\in\mathbb{Z}}\in\ell^2(\mathbb{Z}),
\label{eq:ToeplitzOperator}
\end{equation}
where $\{f_n\}_{n\in\mathbb{Z}}\in\ell^2(\mathbb{Z})$ is the sequence
\begin{equation}
f_n:=\begin{cases} n^{-1},\quad &n\neq 0\\
0,\quad &n=0.
\end{cases}
\label{eq:fsequence}
\end{equation}
\begin{lemma}
\label{infinite sum convergence}
For any even positive integer  $j$, we have
\begin{equation}
\label{infinite sum convergence equation}
\sum_{n_2,\dots,n_j\in\mathbb{Z}}f_{-n_2}
\left[\prod_{\ell=2}^{j-1}f_{n_\ell-n_{\ell+1}}\right]f_{n_j}
=\frac{(i\pi)^j}{j+1},
\end{equation}
where the $j-1$-fold infinite sum converges absolutely.
\end{lemma}
\begin{proof}
Note that since $\{f_n\}_{n\in\mathbb{Z}}\in \ell^2(\mathbb{Z})$, 
$\{g_n\}_{n\in\mathbb{Z}}\in\ell^2(\mathbb{Z})$ as well, where $g_n:=|f_n|$ for
all $n\in\mathbb{Z}$.  The corresponding Fourier series converge in the
mean-square sense to functions $f(\cdot)$ and $g(\cdot)$ in $L^2[0,2\pi]$:
\begin{equation}
f(\theta):=\sum_{n\in\mathbb{Z}}f_ne^{in\theta}=i(\pi-\theta),\quad 0<\theta<2\pi
\label{eq:fthetadefine}
\end{equation}
and
\begin{equation}
g(\theta):=\sum_{n\in\mathbb{Z}}g_ne^{in\theta}=
-\log(2(1-\cos(\theta))),\quad 0<\theta<2\pi.
\end{equation}

First we establish the absolute convergence of the series on the
left-hand side of \eqref{infinite sum convergence equation}.  Using
\eqref{eq:ToeplitzOperator}, observe that
\begin{equation}
\sum_{n_2,\dots,n_j\in\mathbb{Z}}|f_{-n_2}|
\left[\prod_{\ell=2}^{j-1}|f_{n_\ell-n_{\ell+1}}|\right]
|f_{n_j}|
= \left(\mathcal{T}_g^{j-1}g\right)_0
\end{equation}
where $\mathcal{T}_g$ is the Toeplitz operator associated with the
sequence $\{g_n\}_{n\in\mathbb{Z}}$.  Now, $g(\cdot)$ has a logarithmic
singularity at $\theta=0\pmod{2\pi}$, but this is sufficiently mild that
$g(\cdot)^m\in L^2[0,2\pi]\subset L^1[0,2\pi]$ for any positive integer
power $m$.  Now for any function $k(\cdot)\in L^2[0,2\pi]$, the corresponding
Fourier coefficients are
\begin{equation}
k_n:=\frac{1}{2\pi}\int_0^{2\pi}k(\theta)e^{-in\theta}\,d\theta,
\end{equation}
so in particular we see that $(\mathcal{T}_g^{j-1}g)_0$ is the average value
of the function whose Fourier coefficients are 
$\{(\mathcal{T}_g^{j-1}g)_n\}_{n\in\mathbb{Z}}$.  But by the convolution theorem:
\begin{equation}
  w_n:=\sum_{m\in\mathbb{Z}}u_{n-m}v_m \quad\quad\Longleftrightarrow\quad\quad
w(\theta)=u(\theta)v(\theta),
\end{equation}
so it follows that 
\begin{equation}
\left(\mathcal{T}_g^{j-1}g\right)_0 = \frac{1}{2\pi}\int_0^{2\pi}
g(\theta)^j\,d\theta
\end{equation}
which is finite because $g(\cdot)^j\in L^1[0,2\pi]$.

Now we find the exact value of the $j-1$-fold infinite sum by the same
reasoning:
\begin{equation}
\sum_{n_2,\dots,n_j\in\mathbb{Z}}f_{-n_2}
\left[\prod_{\ell=2}^{j-1}f_{n_\ell-n_{\ell+1}}\right]
f_{n_j}= \left(\mathcal{T}_f^{j-1}f\right)_0 = 
\frac{1}{2\pi}\int_0^{2\pi}f(\theta)^j\,d\theta,
\end{equation}
and by direct calculation using \eqref{eq:fthetadefine},
\begin{equation}
\frac{1}{2\pi}\int_0^{2\pi}f(\theta)^j\,d\theta= 
\frac{1}{2\pi}\int_0^{2\pi}[i(\pi-\theta)]^j\,d\theta = \frac{(i\pi)^j}{j+1}
\end{equation}
for $j$ even (the integral vanishes by symmetry for $j$ odd).
\end{proof}

Now we consider separately each of the terms in $Z_{pj}$ for $j$ even.
\begin{lemma}
\label{asymptotic preperties for trace 2}If $j$ is an even number
and $h_1+\cdots + h_s=j$ while $d_1+\cdots + d_s=p-j$, then
\begin{equation}
\label{asymptotic preperties for trace 2 1}
\lim_{\epsilon \downarrow 0}\frac{M}{N(\epsilon)}\mathrm{tr}\left(
\mat{D}^{d_1}\mat{H}^{h_1}\cdots \mat{D}^{d_s}\mat{H}^{h_s}\right)
{}=\frac{(2\pi)^{j}}{j+1}\int_{-L}^{0}D(\lambda;x,t)^{p-j}\varphi(\lambda)^{2j}
F(\lambda)
\,d\lambda,
\end{equation}
with the limit being uniform with respect to $(x, t)$ in any compact set.
\end{lemma}
\begin{proof}
Recalling the matrix elements $D_n$ and $H_{nm}$ of $\mat{D}$ and $\mat{H}$
respectively, we have
\begin{equation}
\mathrm{tr}\left(\mat{D}^{d_1}\mat{H}^{h_1}\cdots\mat{D}^{d_s}\mat{H}^{h_s}
\right) = 
\sum_{a_1,a_2,\dots,a_j=1}^{N(\epsilon)}
\left[\prod_{i=1}^j D_{a_i}^{m_i}\right]
\left[\prod_{\ell=1}^{j-1}H_{a_\ell a_{\ell+1}}\right]H_{a_ja_1},
\label{eq:trformula}
\end{equation}
where the exponents $m_1,\dots,m_j$ are given by
\begin{equation}
m_i:=\begin{cases}
d_1,\quad &i=1\\
d_{b+1},\quad &i=1+h_1+h_2+\cdots + h_b\quad\text{for some $0<b<s$}\\
0,\quad &\text{otherwise}.
\end{cases}
\end{equation}
Note that $m_1+m_2+\cdots + m_j=d_1+d_2+\cdots d_s=p-j$.

Now, the matrix element $H_{nm}$ is relatively small unless 
$n\approx m$, and this
suggests that the $j$-fold sum in \eqref{eq:trformula} 
should concentrate near the diagonal, where
$a_k=a_1$ for all $k$.  Making this precise, given any $r>0$ we will
first show that
\begin{equation}
\lim_{\epsilon\downarrow 0}Z_{\mathrm{OD}}(\epsilon)=0,
\label{eq:ZODlimit}
\end{equation}
where
\begin{equation}
Z_\mathrm{OD}(\epsilon):=\frac{M}{N(\epsilon)}
\mathop{\sum_{a_1,a_2,\dots,a_j=1}^{N(\epsilon)}}_{\exists k: |a_k-a_1|>\epsilon^{-r}}
\left[\prod_{i=1}^j D_{a_i}^{m_i}\right]
\left[\prod_{\ell=1}^{j-1}H_{a_\ell a_{\ell+1}}\right] H_{a_ja_1},
\end{equation}
with the limit being uniform for $(x,t)$ in compact sets.  Indeed, if
$x^2+t^2\le R^2$, then using \eqref{eq:Dnbound} from Lemma~\ref{lem:D}
and \eqref{eq:Hnmbound} from Lemma~\ref{lem:Hnm} we obtain
\begin{equation}
\begin{split}
|Z_\mathrm{OD}(\epsilon)|&\le\frac{MC_{D,R}^{p-j}C_H^j}{N(\epsilon)}
\mathop{\sum_{a_1,a_2,\dots,a_j=1}^{N(\epsilon)}}_{\exists k:|a_k-a_1|>
\epsilon^{-r}}
\left[\prod_{\ell=1}^{j-1}|f_{a_\ell - a_{\ell+1}}|
\right] |f_{a_j-a_1}|
\\
&=\frac{MC_{D,R}^{p-j}C_H^j}{N(\epsilon)}\sum_{a_1=1}^{N(\epsilon)}
\mathop{\sum_{a_2,a_3,\dots,a_j=1}^{N(\epsilon)}}_{\exists k:|a_k-a_1|>\epsilon^{-r}}
\left[\prod_{\ell=1}^{j-1}|f_{a_\ell-a_{\ell+1}}|\right]|f_{a_j-a_1}|
\\
&\le\frac{MC_{D,R}^{p-j}C_H^j}{N(\epsilon)}\sum_{a_1=1}^{N(\epsilon)}
\mathop{\sum_{a_2,a_3,\dots,a_j\in\mathbb{Z}}}_{\exists k:|a_k-a_1|>\epsilon^{-r}}
\left[\prod_{\ell=1}^{j-1}|f_{a_\ell-a_{\ell+1}}|\right]|f_{a_j-a_1}|.
\end{split}
\end{equation}
With the inner sum extended over $\mathbb{Z}^{j-1}$ in this way, it becomes
independent of the outer sum index $a_1$ as can be seen by the substitution
$n_k=a_k-a_1$ for $k=2,3,\dots j$.  Thus
\begin{equation}
|Z_\mathrm{OD}(\epsilon)|\le MC_{D,R}^{p-j}C_H^j
\mathop{\sum_{n_2,n_3,\dots,n_j\in\mathbb{Z}}}_{\exists k:|n_k|>\epsilon^{-r}}
|f_{-n_2}|\left[\prod_{\ell=2}^{j-1}|f_{n_\ell-n_{\ell+1}}|\right]
|f_{n_j}|,
\end{equation}
and the latter upper bound is of course independent of $(x,t)$ with 
$x^2+t^2\le R^2$ and tends to zero for $r>0$ by 
Lemma~\ref{infinite sum convergence}.  

It follows from \eqref{eq:ZODlimit} that
\begin{equation}
\lim_{\epsilon\downarrow 0}
\frac{M}{N(\epsilon)}\mathrm{tr}\left(\mat{D}^{d_1}
\mat{H}^{h_1}\cdots \mat{D}^{d_s}\mat{H}^{h_s}\right) = 
\lim_{\epsilon\downarrow 0}Z_\mathrm{D}(\epsilon)
\label{eq:ZDanswer}
\end{equation}
where the diagonally-concentrated terms are
\begin{equation}
Z_\mathrm{D}(\epsilon):=\frac{M}{N(\epsilon)}
\mathop{\sum_{a_1,a_2,\dots,a_j=1}^{N(\epsilon)}}_{\forall k: 
|a_k-a_1|\le\epsilon^{-r}}
\left[\prod_{i=1}^j D_{a_i}^{m_i}\right]
\left[\prod_{\ell=1}^{j-1}H_{a_\ell a_{\ell+1}}\right]H_{a_ja_1}.
\label{eq:ZDepsilon}
\end{equation}
We will analyze $Z_\mathrm{D}(\epsilon)$ under the additional
assumption that $r<1$.  

The first step is show that if $r<1$ each occurrence of $H_{nm}$ 
in \eqref{eq:ZDepsilon}
may be replaced by $2i\varphi_n\varphi_mf_{n-m}$ without
affecting the limiting value of $Z_\mathrm{D}(\epsilon)$ as 
$\epsilon\downarrow 0$.
Indeed, by making this substitution $j$ times in succession each time keeping track of the error using Lemma~\ref{lem:Hnm} along with the estimates
\eqref{eq:varphinbound} and \eqref{eq:Dnbound} from Lemma~\ref{lem:D},
one sees that 
with $K_R>0$ defined by
\begin{equation}
K_R:=C_{D,R}^{p-j}\sum_{k=1}^j(2C_\varphi^2)^{k-1}C_H^{j-k+1},
\end{equation}
for all $j$-tuples of integers $a_1,\dots,a_j$ between $1$ and $N(\epsilon)$
satisfying $|a_k-a_1|\le\epsilon^{-r}$ for all $k$,
\begin{multline}
\left|\left[\prod_{i=1}^jD_{a_i}^{m_i}\right]\left[\prod_{\ell=1}^{j-1}
H_{a_\ell a_{\ell+1}}\right]H_{a_ja_1}-
(2i)^j\left[\prod_{i=1}^j D_{a_i}^{m_i}\varphi_{a_i}^2\right]
\left[\prod_{\ell=1}^{j-1}f_{a_\ell-a_{\ell+1}}\right]f_{a_j-a_1}
\right|\\
{}\le
K_R\epsilon^{(1-r)\sigma/3}\left[\prod_{\ell=1}^{j-1}|f_{a_\ell-a_{\ell+1}}|
\right]|f_{a_j-a_1}|.
\end{multline}
Therefore, if we define a modification of $Z_\mathrm{D}(\epsilon)$
by
\begin{equation}
Z_\mathrm{D}^\mathrm{I}(\epsilon):=
\frac{(2i)^jM}{N(\epsilon)}
\mathop{\sum_{a_1,a_2,\dots,a_j=1}^{N(\epsilon)}}_{\forall k:
|a_k-a_1|\le\epsilon^{-r}}
\left[\prod_{i=1}^jD_{a_i}^{m_i}\varphi_{a_i}^2\right]
\left[\prod_{\ell=1}^{j-1}f_{a_\ell-a_{\ell+1}}\right]f_{a_j-a_1},
\label{eq:ZDI}
\end{equation}
we have
\begin{equation}
\begin{split}
\left|Z_\mathrm{D}(\epsilon)-Z_\mathrm{D}^\mathrm{I}(\epsilon)\right|
&\le \frac{MK_R\epsilon^{(1-r)\sigma/3}}{N(\epsilon)}
\mathop{\sum_{a_1,a_2,\dots,a_j=1}^{N(\epsilon)}}_{\forall k:|a_k-a_1|\le
\epsilon^{-r}}\left[\prod_{\ell=1}^{j-1}|f_{a_\ell-a_{\ell+1}}|\right]
|f_{a_j-a_1}|\\
&\le\frac{MK_R\epsilon^{(1-r)\sigma/3}}{N(\epsilon)}
\sum_{a_1=1}^{N(\epsilon)}\left(\sum_{a_2,a_3,\dots,a_j\in\mathbb{Z}}
\left[\prod_{\ell=1}^{j-1}|f_{a_\ell-a_{\ell+1}}|\right]
|f_{a_j-a_1}|\right).
\end{split}
\end{equation}
By the substitution
$n_\ell=a_\ell-a_1$ one sees that the inner sum is independent of $a_1$, 
and it is finite by Lemma~\ref{infinite sum convergence}.
Since $\sigma>0$ and $r<1$, we therefore have
\begin{equation}
\lim_{\epsilon\downarrow 0}Z_\mathrm{D}(\epsilon) = 
\lim_{\epsilon\downarrow 0}Z_\mathrm{D}^\mathrm{I}(\epsilon)
\end{equation}
uniformly for $x^2+t^2\le R^2$.

The second step is to show that if $r<1$ we may replace $D_{a_i}^{m_i}\varphi_{a_i}^2$ with $D_{a_1}^{m_i}\varphi_{a_1}^2$ for each $i$ in 
\eqref{eq:ZDI} without changing the
limiting value of $Z_\mathrm{D}^\mathrm{I}(\epsilon)$.  Indeed, applying
Lemma~\ref{lem:D} we see that with $K_R^\mathrm{I}>0$ defined by
\begin{equation}
K_R^\mathrm{I}:=(p+j)C_{D,R}^{p-j}C_\varphi^{2j},
\end{equation}
we see that for all $j$-tuples of integers $a_1,\dots,a_j$ between 
$1$ and $N(\epsilon)$ satisfying $|a_k-a_1|\le \epsilon^{-r}$
for all $k$,
\begin{equation}
\left|\prod_{i=1}^jD_{a_i}^{m_i}\varphi_{a_i}^2 - D_{a_1}^{p-j}\varphi_{a_1}^{2j}
\right|\le K_R^\mathrm{I}\epsilon^{(1-r)\sigma/3}.
\end{equation}
Hence, definining a subsequent modification of $Z_\mathrm{D}^\mathrm{I}(\epsilon)$
by
\begin{equation}
Z_\mathrm{D}^\mathrm{II}(\epsilon):=\frac{(2i)^jM}{N(\epsilon)}
\sum_{a_1=1}^{N(\epsilon)}D_{a_1}^{p-j}\varphi_{a_1}^{2j}
\mathop{\sum_{a_2,a_3,\dots,a_j=1}^{N(\epsilon)}}_{\forall k:|a_k-a_1|\le\epsilon^{-r}}
\left[\prod_{\ell=1}^{j-1}f_{a_\ell-a_{\ell+1}}\right]f_{a_j-a_1},
\end{equation}
we see that
\begin{equation}
\begin{split}
\left|Z_\mathrm{D}^\mathrm{I}(\epsilon)-Z_\mathrm{D}^\mathrm{II}(\epsilon)
\right|&\le \frac{2^jMK_R^\mathrm{I}\epsilon^{(1-r)\sigma/3}}{N(\epsilon)}
\sum_{a_1=1}^{N(\epsilon)}\mathop{\sum_{a_2,a_3,\dots,a_j=1}^{N(\epsilon)}}_{\forall k|a_k-a_1|\le\epsilon^{-r}}\left[\prod_{\ell=1}^{j-1}|f_{a_\ell-a_{\ell+1}}|
\right]|f_{a_j-a_1}|\\
&\le
\frac{2^jMK_R^\mathrm{I}\epsilon^{(1-r)\sigma/3}}{N(\epsilon)}
\sum_{a_1=1}^{N(\epsilon)}\left(
\sum_{a_2,a_3,\dots,a_j\in\mathbb{Z}}\left[\prod_{\ell=1}^{j-1}|f_{a_\ell-a_{\ell+1}}|
\right]|f_{a_j-a_1}|\right),
\end{split}
\end{equation}
and so exactly as before
\begin{equation}
\lim_{\epsilon\downarrow 0}Z_\mathrm{D}^\mathrm{I}(\epsilon)=
\lim_{\epsilon\downarrow 0}Z_\mathrm{D}^\mathrm{II}(\epsilon)
\end{equation}
uniformly for $x^2+t^2\le R^2$.

The third step is to show that if $r<1$ 
one may neglect a small fraction of the
terms in the outer sum corresponding to $a_1\le 1+\epsilon^{-r}$
and $a_1\ge N(\epsilon)-\epsilon^{-r}$ without changing the limiting
value of $Z_\mathrm{D}^\mathrm{II}(\epsilon)$.  Indeed, defining the
index set
\begin{equation}
S_\epsilon:=\{n\in\mathbb{Z},\;
1+\epsilon^{-r}<n<N(\epsilon)-\epsilon^{-r}\},
\end{equation}
and then setting
\begin{equation}
Z_\mathrm{D}^\mathrm{III}(\epsilon):=\frac{(2i)^jM}{N(\epsilon)}
\sum_{a_1\in S_\epsilon} D_{a_1}^{p-j}\varphi_{a_1}^{2j}
\mathop{\sum_{a_2,a_3,\dots a_j=1}^{N(\epsilon)}}_{\forall k:|a_k-a_1|\le\epsilon^{-r}}
\left[\prod_{\ell=1}^{j-1}f_{a_\ell-a_{\ell+1}}\right]f_{a_j-a_1},
\label{eq:ZDIII}
\end{equation}
we easily obtain from \eqref{eq:varphinbound} and \eqref{eq:Dnbound}
in Lemma~\ref{lem:D} that 
\begin{equation}
\begin{split}
\left|Z_\mathrm{D}^\mathrm{II}(\epsilon)-Z_\mathrm{D}^\mathrm{III}(\epsilon)
\right|&\le \frac{2^jMC_{D,R}^{p-j}C_\varphi^{2j}}{N(\epsilon)}
\mathop{\sum_{a_1=1}^{N(\epsilon)}}_{a_1\not\in S_\epsilon}\mathop{\sum_{a_2,a_3,\dots,a_j=1}^{N(\epsilon)}}_{\forall k:|a_k-a_1|\le\epsilon^{-r}}
\left[\prod_{\ell=1}^{j-1}|f_{a_\ell-a_{\ell+1}}|\right]|f_{a_j-a_1}|\\
&\le
\frac{2^jMC_{D,R}^{p-j}C_\varphi^{2j}}{N(\epsilon)}
\mathop{\sum_{a_1=1}^{N(\epsilon)}}_{a_1\not\in S_\epsilon}
\sum_{a_2,a_3,\dots,a_j\in\mathbb{Z}}
\left[\prod_{\ell=1}^{j-1}|f_{a_\ell-a_{\ell+1}}|\right]|f_{a_j-a_1}|.
\end{split}
\end{equation}
But the inner sum is independent of $a_1$ and is convergent by 
Lemma~\ref{infinite sum convergence} and the outer sum has
$O(\epsilon^{-r})$
terms while $N(\epsilon)$ is proportional to $\epsilon^{-1}$, so with
$r<1$ we have
\begin{equation}
\lim_{\epsilon\downarrow 0}Z_\mathrm{D}^\mathrm{II}(\epsilon)=
\lim_{\epsilon\downarrow 0}Z_\mathrm{D}^\mathrm{III}(\epsilon)
\label{eq:ZDII-IIIlimit}
\end{equation}
uniformly for $x^2+t^2\le R^2$.  

The next step in analyzing $Z_\mathrm{D}(\epsilon)$ is to 
deal with the inner sum in the definition \eqref{eq:ZDIII} of
$Z_\mathrm{D}^\mathrm{III}(\epsilon)$.  
Taking
into account the conditions on $a_1$ in the outer sum,
it is obvious that the conditions
$1\le a_k\le N(\epsilon)$ are superfluous in the inner sum:
\begin{equation}
Z_\mathrm{D}^\mathrm{III}(\epsilon)=
\frac{(2i)^jM}{N(\epsilon)}\sum_{a_1\in S_\epsilon}
D_{a_1}^{p-j}\varphi_{a_1}^{2j}\mathop{\sum_{a_2,a_3,\dots,a_j\in\mathbb{Z}}}_{\forall k:|a_k-a_1|\le\epsilon^{-r}}
\left[\prod_{\ell=1}^{j-1}f_{a_\ell-a_{\ell+1}}\right]f_{a_j-a_1}.
\end{equation}
By introducing the differences $n_k=a_k-a_1$ it now becomes clear that
the inner sum is independent of $a_1$:
\begin{equation}
Z_\mathrm{D}^\mathrm{III}(\epsilon)=
\frac{(2i)^jM}{N(\epsilon)}\left(\sum_{a_1\in S_\epsilon}
D_{a_1}^{p-j}\varphi_{a_1}^{2j}\right)
\left(\mathop{\sum_{n_2,n_3,\dots,n_j\in\mathbb{Z}}}_{\forall k:|n_k|\le\epsilon^{-r}}
f_{-n_2}\left[\prod_{\ell=2}^{j-1}f_{a_\ell-a_{\ell+1}}\right]f_{n_j}\right).
\end{equation}
Now, according to Lemma~\ref{infinite sum convergence}, 
the latter sum has the limit $(i\pi)^j/(j+1)$ as $\epsilon\downarrow 0$
with $r>0$,
so
\begin{equation}
\lim_{\epsilon\downarrow 0}Z_\mathrm{D}^\mathrm{III}(\epsilon) = 
\lim_{\epsilon\downarrow 0}Z_\mathrm{D}^\mathrm{IV}(\epsilon),
\label{eq:ZDIII-IVlimit}
\end{equation}
uniformly for $x^2+t^2\le R^2$, where
\begin{equation}
Z_\mathrm{D}^\mathrm{IV}(\epsilon):=\frac{(2\pi)^j}{j+1}\cdot
\frac{M}{N(\epsilon)}\sum_{a_1\in S_\epsilon}D_{a_1}^{p-j}\varphi_{a_1}^{2j}.
\end{equation}

The final step in the analysis of $Z_\mathrm{D}(\epsilon)$ is simply to 
evaluate the limit on the right-hand side of 
\label{eq:ZDIII-IVlimit} by recognizing the sum as a Riemann sum
for an integral:
\begin{equation}
\lim_{\epsilon\downarrow 0}Z_\mathrm{D}(\epsilon)=
\lim_{\epsilon\downarrow 0}Z_\mathrm{D}^{\mathrm IV}(\epsilon) = 
\frac{(2\pi)^j}{j+1}\int_{-L}^0 D(\lambda;x,t)^{p-j}\varphi(\lambda)^{2j}F(\lambda)\,d\lambda.
\end{equation}
Note that since the summand $D_{a_1}^{p-j}\varphi_{a_1}^{2j}$ 
is polynomial in $x$ and $t$, the convergence of the Riemann sum is uniform
for $(x,t)$ in compact sets.  Comparing with
\eqref{eq:ZDanswer} we see that the proof is complete.
\end{proof}

Now we may complete the proof of Proposition~\ref{asymptotic preperties 
for trace 3}.  Lemma~\ref{asymptotic preperties for trace 2} shows that
each of the terms in the formula \eqref{eq:Zpj} for $Z_{pj}$ has the same
limit as $\epsilon\downarrow 0$.  Therefore, for all even $j$,
\begin{equation}
\begin{split}
\lim_{\epsilon\downarrow 0}Z_{pj} &= \mathop{\sum_{d_1+d_2+\cdots d_s=p-j}}_{h_1+h_2+\cdots + h_s=j} 
\frac{(2\pi)^j}{j+1}\int_{-L}^0 D(\lambda;x,t)^{p-j}\varphi(\lambda)^{2j}F(\lambda)\,d\lambda \\
&= \binom{p}{j}\frac{(2\pi)^j}{j+1}\int_{-L}^0 D(\lambda;x,t)^{p-j}\varphi(\lambda)^{2j}F(\lambda)\,d\lambda.
\end{split}
\end{equation}
Combining this result with Lemma~\ref{asymptotic preperties for trace 1}
and the formula \eqref{small dispersion matrix calculation 8} for the
$p^\mathrm{th}$ moment, we obtain
\begin{equation}
Q_p=\lim_{\epsilon\downarrow 0}\int_\mathbb{R}\alpha^pd\mu_\epsilon(\alpha)=
\sum_{k=0}^{\lfloor p/2\rfloor}\binom{p}{2k}\frac{(2\pi)^{2k}}{2k+1}
\int_{-L}^0 D(\lambda;x,t)^{p-2k}\varphi(\lambda)^{4k}F(\lambda)\,d\lambda,
\end{equation}
uniformly for $(x,t)$ in compact sets.
Now we apply the identity
\begin{equation}
\label{identity for theorem 4-1}
\sum_{k=0}^{\lfloor p/2 \rfloor} \frac{1}{2k+1} \binom{p}{2k} a^{2k}b^{p-2k} = \frac{(b+a)^{p+1} - (b-a)^{p+1}}{2a(1+p)},
\end{equation}
holding for any integer $p\ge 0$ and real numbers $a$ and $b$.  (This identity can be most easily obtained by expanding the binomials on the right-hand side.)
Recalling the definitions \eqref{eq:Dfuncdef} and \eqref{eq:varphidef}
of $D(\lambda;x,t)$ and $\varphi(\lambda)$, and using the fact that
$x_\pm(\lambda)=\pm \pi F(\lambda)-\gamma(\lambda)$
then completes the proof of Proposition~\ref{asymptotic preperties for trace 3}.

\subsection{Convergence of measures and locally 
uniform convergence of $\Ut$.
Proof of Proposition~\ref{limit of sx theorem}}
\label{sec:measures}
Recall the 
measures $\mu_\epsilon$ defined by 
\eqref{eq:ncm}.
\begin{lemma}
For each nonnegative integer $p$,
\begin{equation}
\lim_{\epsilon\downarrow 0}
\int_\mathbb{R} \alpha^p\,d\mu_\epsilon(\alpha) = \int_\mathbb{R}
\alpha^p\,d\mu(\alpha)
\end{equation}
where $\mu$ is the absolutely continuous (with respect to Lebesgue
measure on $\mathbb{R}$) measure defined by
$d\mu(\alpha)=G(\alpha;x,t)\,d\alpha$,
and the compactly
supported integrable density function $G(\alpha;x,t)$ is given by
\eqref{limit of sx 2}.  The limit is uniform with respect to $(x,t)$
in compact sets.
Also, like each $\mu_\epsilon$, $\mu$ is a measure with mass $M$.
\label{lem:limitingmoments}
\end{lemma}
\begin{proof}
Recalling Proposition~\ref{asymptotic preperties for trace 3},
we first show that the given measure $\mu$ satisfies
\begin{equation}
\int_{\mathbb{R}}\alpha^p\,d\mu(\alpha) = Q_p,
\label{eq:limitingmoment}
\end{equation}
where $Q_p$ is given by \eqref{eq:hp},
for all nonnegative $p\in\mathbb{Z}$.
Equivalently, we
may construct a measure with the desired moments as follows:
the characteristic
function of the measure $\mu$ is the Fourier transform
\begin{equation}
\hat{G}(\xi;x,t):=\int_{\mathbb{R}}G(\alpha;x,t)e^{-i\alpha\xi}\,d\alpha,
\end{equation}
and this function necessarily has the desired moments $\{Q_p\}_{p=0}^\infty$
as its derivatives at $\xi=0$:
\begin{equation}
\frac{d^p\hat{G}}{d\xi^p}(0;x,t) = (-i)^pQ_p.
\end{equation}
So $\hat{G}(\xi;x,t)$ has the Taylor series
\begin{equation}
\hat{G}(\xi;x,t) = \sum_{p=0}^\infty\frac{(-i\xi)^p}{p!}Q_p.
\label{eq:GhatTaylor}
\end{equation}
Now from the obvious inequality $|x+2\lambda t-x_\pm(\lambda)|\le
|x-x_0|+2L|t|+2\pi F(\lambda)$,
we obtain 
\begin{equation}
\begin{split}
|Q_p|&\le \frac{1}{\pi(p+1)}\int_{-L}^0\left(|x-x_0|+2L|t|+2\pi F(\lambda)
\right)^{p+1}(-2\lambda)^p\,d\lambda\\
&\le \frac{1}{\pi(p+1)}\int_{-L}^0
(2L|x-x_0|+4L^2|t|-4\pi\lambda F(\lambda))^p(|x-x_0|+2L|t|+2\pi F(\lambda))\,d\lambda.
\end{split}
\end{equation}
Also, from Lemma~\ref{Lemma for function F}, there is a constant
$K>0$ such that $0\le -\lambda F(\lambda)\le K$, so 
for $(x-x_0)^2+t^2\le R^2$,
\begin{equation}
\begin{split}
|Q_p|&\le \frac{(2LR+4L^2R+4\pi K)^p}{\pi (p+1)}\int_{-L}^0 \left(
|x-x_0|+2L|t|
+2\pi F(\lambda)\right)\,d\lambda\\
&\le \frac{(LR +2L^2R+2\pi M)}{\pi (p+1)}(2LR+4L^2R+4\pi K)^p\\
&\le \frac{1}{\pi}(LR+2L^2R+2\pi M)(2LR+4L^2R+4\pi K)^p,
\end{split}
\end{equation}
where in the last step we used \eqref{eq:intF}.
This inequality implies 
that the Taylor series \eqref{eq:GhatTaylor} converges for
all $\xi\in\mathbb{C}$ to an entire function of exponential type.

Now we will sum the Taylor series \eqref{eq:GhatTaylor} in closed form
by substituting from the formula \eqref{eq:hp} and 
exchanging the order of summation and integration.  Indeed, since
\begin{equation}
\sum_{p=0}^\infty \frac{(-i\xi)^p}{p!}\cdot\frac{(-2\lambda)^p(x+2\lambda t-x_\pm(\lambda))^{p+1}}{p+1} = 
\frac{e^{2i\xi\lambda[x+2\lambda t-x_\pm(\lambda)]}-1}{2i\xi\lambda},
\end{equation}
we obtain the formula
\begin{equation}
\hat{G}(\xi;x,t)=\int_{-L}^0\frac{e^{2i\xi\lambda[x+2\lambda t-x_-(\lambda)]}-
e^{2i\xi\lambda[x+2\lambda t-x_+(\lambda)]}}{4\pi i\xi\lambda}\,d\lambda.
\end{equation}
Computing the inverse Fourier transform 
\begin{equation}
G(\alpha;x,t) = \frac{1}{2\pi}\int_\mathbb{R}\hat{G}(\xi;x,t)e^{i\alpha\xi}\,d\xi
\end{equation}
by exchanging the order of integration leads directly to the claimed
formula \eqref{limit of sx 2}.

It is obvious that $G(\alpha;x,t)$ is a nonnegative function, and since
by Lemma~\ref{Lemma for function F}
\begin{equation}
\inf_{-L<\lambda<0}-2\lambda(x+2\lambda t-x_+(\lambda))>-\infty
\quad\text{and}\quad
\sup_{-L<\lambda<0}-2\lambda (x+2\lambda t-x_-(\lambda))<+\infty
\end{equation}
for every $(x,t)$, it is clear that $G(\alpha;x,t)$ has compact support.
It is also straightforward to verify that $\mu$ has mass $M$:
\begin{equation}
\begin{split}
\int_{\mathbb{R}}d\mu(\alpha)&=\int_{\mathbb{R}}G(\alpha;x,t)\,d\alpha\\
&=-\frac{1}{4\pi}\int_{\mathbb{R}}\int_{-L}^0
\chi_{[-2\lambda(x+2\lambda t-x_+(\lambda)),-2\lambda(x+2\lambda t-x_-(\lambda))]}(\alpha)\frac{d\lambda}{\lambda}\,d\alpha\\
&=-\frac{1}{4\pi}\int_{-L}^0\frac{1}{\lambda}
\int_{\mathbb{R}}\chi_{[-2\lambda(x+2\lambda t-x_+(\lambda)),-2\lambda(x+2\lambda t-x_-(\lambda))]}(\alpha)\,d\alpha\,d\lambda\\
&=-\frac{1}{4\pi}\int_{-L}^0\frac{1}{\lambda}\int_{-2\lambda(x+2\lambda t-x_+(\lambda))}^{-2\lambda(x+2\lambda t-x_-(\lambda))}\,d\alpha\,d\lambda\\
&=\int_{-L}^0F(\lambda)\,d\lambda\\
&=M,
\end{split}
\end{equation}
according to \eqref{eq:intF}.
Therefore $\mu$ is indeed an absolutely
continuous compactly supported (nonnegative) measure of mass $M$.
\end{proof}
Note that the reconstruction of the the measure $\mu$ from its moments
is virtually the same calculation as took place on the direct scattering
side in our discussion of Matsuno's method in \S\ref{sec:MatsunosMethod}.

\begin{lemma}
There is a compact interval $\Omega\subset\mathbb{R}$ containing the
support of all of the measures $\{\mu_\epsilon\}_{\epsilon>0}$ as well
as that of the measure $\mu$, and $\Omega$ may be chosen independent
of $(x,t)$ in any given compact set.
\label{lem:tight}
\end{lemma}
\begin{proof}
  Since $\mu$ has compact support certainly contained within the interval
\begin{equation}
\inf_{-L<\lambda<0}[2\lambda x_+(\lambda)] -2L|x|-4L^2|t|
\le\alpha\le
\sup_{-L<\lambda<0}[2\lambda x_-(\lambda)] +
2L|x|+4L^2|t|
\end{equation}
that is clearly bounded uniformly for $(x,t)$ in any compact set,
it is enough to show that the
  support of $\mu_\epsilon$ is uniformly bounded as
  $\epsilon\downarrow 0$.  But by definition of $\mu_\epsilon$ this is
  equivalent to showing that the eigenvalue of $\mat{A}_\epsilon$
  with the largest magnitude remains uniformly bounded as
  $\epsilon\downarrow 0$.

Since the matrix $\tilde{\mat{A}}_{\epsilon}$ is Hermitian, we have
\begin{equation}
\label{proof of the bounded eigenvalues of W
0-1}\|\tilde{\mat{A}}_{\epsilon}\|_{2}=\max_{1\leq j \leq N(\epsilon)}|\alpha_{j}|,
\end{equation}
so to prove that the eigenvalue of $\tilde{\mat{A}}_{\epsilon}$ with the
largest magnitude remains uniformly bounded, it is completely
equivalent to prove that the $\ell^{2}$ (induced) matrix norm of
$\tilde{\mat{A}}_{\epsilon}$ is uniformly bounded as $\epsilon\downarrow 0$
independent of $(x,t)$ in any given compact set. 

Recalling the decomposition $\tilde{\mat{A}}_\epsilon = \mat{D}+\mat{H}$ from
the proof of Proposition~\ref{asymptotic preperties for trace 3} given
in \S\ref{sec:moments}, the triangle inequality gives $\|\tilde{\mat{A}}_\epsilon\|_2
\le \|\mat{D}\|_2 +\|\mat{H}\|_2$, and since $\mat{D}$ is diagonal,
\begin{equation}
\begin{split}
\|\mat{D}\|_2 &= 
\max_{1\le n\le N(\epsilon)}|2\lamt_n(x+2\lamt_nt+\gamma(\lamt_n))|\\
& \le\sup_{-L<\lambda<0}|2\lambda(x+2\lambda t +\gamma(\lambda))|\\
&\le \sup_{-L<\lambda<0}|2\lambda\gamma(\lambda)| + 2L|x|+4L^2|t|,
\end{split}
\end{equation}
so since $\lambda\gamma(\lambda)$ is bounded according to
Lemma~\ref{Lemma for function F}, and $\mat{H}$ is independent of $x$ and $t$,
it is sufficient to show that $\|\mat{H}\|_2$ remains bounded as 
$\epsilon\downarrow 0$.

To estimate $\|\mat{H}\|_2$, we write $\mat{H}$ in the following form:
$\mat{H}=\mat{B}\mat{T}\mat{B} +\mat{E}$ where 
\begin{equation}
\mat{B}=\mathrm{diag}\left(e^{i\pi/4}\sqrt{-2\lamt_1F(\lamt_1)
\vphantom{-2\lamt_{N(\epsilon)}F(\lamt_{N(\epsilon)})}},\dots,
e^{i\pi/4}\sqrt{-2\lamt_{N(\epsilon)}F(\lamt_{N(\epsilon)})}\right),
\end{equation}
and 
$\mat{T}$ is the $N(\epsilon)\times N(\epsilon)$ Toeplitz matrix with elements
$T_{nm}=f_{n-m}$,
where the sequence $\{f_n\}_{n\in\mathbb{Z}}$ is defined by \eqref{eq:fsequence}.
Of course $\mat{E}:=\mat{H}-\mat{B}\mat{T}\mat{B}$.  Therefore
$\|\mat{H}\|_2 \le \|\mat{B}\|_2^2\|\mat{T}\|_2 + \|\mat{E}\|_2$.  Because 
$\mat{B}$ is diagonal, 
\begin{equation}
\|\mat{B}\|_2^2\le\max_{1\le n\le N(\epsilon)}[-2\lamt_nF(\lamt_n)]
\le\sup_{-L<\lambda<0}[-2\lambda F(\lambda)]
\end{equation}
which is finite by Lemma~\ref{Lemma for function F}.
The Toeplitz matrix $\mat{T}$ can be written as
$\mat{T}=\mathcal{P} \mathcal{T}_f\mathcal{P}$, where $\mathcal{P}$ is
the orthogonal projection from $\ell^2(\mathbb{Z})$ onto
$\mathbb{C}^N$ viewed as a subset of $\ell^2(\mathbb{Z})$ associated
with components having indices
$\{1,2,\dots,N(\epsilon)\}\subset\mathbb{Z}$, and where
$\mathcal{T}_f:\ell^2(\mathbb{Z})\to\ell^2(\mathbb{Z})$ is the Toeplitz
operator defined by \eqref{eq:ToeplitzOperator} from \S\ref{sec:moments}.
The $\ell^2(\mathbb{Z})$ operator norm of $\mathcal{P}$ is clearly equal
to one, and since
\begin{equation}
\sum_{l\in\mathbb{Z}}f_le^{il\theta} = i(\pi-\theta),\quad 0<\theta<2\pi,
\end{equation}
the Pythagorean Theorem in $L^2(0,2\pi)$ gives
\begin{equation}
\begin{split}
\sum_{n\in\mathbb{Z}}|(\mathcal{T}c)_n|^2 &=\frac{1}{2\pi}\int_0^{2\pi}
\left|\sum_{n\in\mathbb{Z}}(\mathcal{T}c)_ne^{in\theta}\right|^2\,d\theta\\
&=\frac{1}{2\pi}\int_0^{2\pi}\left|\sum_{n\in\mathbb{Z}}\sum_{m\in\mathbb{Z}}
f_{n-m}c_me^{in\theta}\right|^2\,d\theta\\
&=\frac{1}{2\pi}\int_0^{2\pi}\left|\sum_{m\in\mathbb{Z}}c_me^{im\theta}
\sum_{n\in\mathbb{Z}}f_{n-m}e^{i[n-m]\theta}\right|^2\,d\theta\\
&=\frac{1}{2\pi}\int_0^{2\pi}(\pi-\theta)^2\left|\sum_{m\in\mathbb{Z}}c_me^{im\theta}\right|^2\,d\theta\\
&\le \pi^2\frac{1}{2\pi}\int_0^{2\pi}\left|\sum_{m\in\mathbb{Z}}c_me^{im\theta}
\right|^2\,d\theta\\
&= \pi^2\sum_{m\in\mathbb{Z}}|c_m|^2,
\end{split}
\end{equation}
the $\ell^2(\mathbb{Z})$ operator norm of $\mathcal{T}_f$ is bounded by $\pi$.
It follows that $\|\mat{H}\|_2 \le \pi +\|\mat{E}\|_2$, so it suffices to
show that $\|\mat{E}\|_2$ remains bounded as $\epsilon\downarrow 0$.

So far, we have exploited the special structure of the dominant parts of the
matrix $\tilde{\mat{A}}_\epsilon$ and applied correspondingly
specialized norm estimates to these terms.  The error term $\mat{E}$
has less structure, but is it smaller; to estimate its norm it will be
sufficient to use the rather crude inequality 
$\|\mat{E}\|_2\le\|\mat{E}\|_{\mathrm{HS}}$ and work with the Hilbert-Schmidt
norm
\begin{equation}
\|\mat{E}\|_{\mathrm{HS}}^2:=\sum_{n=1}^{N(\epsilon)}\sum_{m=1}^{N(\epsilon)}|E_{nm}|^2,
\end{equation}
where the elements of $\mat{E}$ are explicitly given by
\begin{equation}
E_{nm}:=2i\left[\frac{\epsilon\sqrt{\lamt_n\lamt_m}}{\lamt_n-\lamt_m}
-\frac{\sqrt{\lamt_nF(\lamt_n)\lamt_mF(\lamt_m)}}{n-m}\right],
\quad\text{for $n\neq m$, and $E_{nn}=0$.}
\end{equation}
If we introduce continuous variables $a:=(n-\tfrac{1}{2})\epsilon$
and $b:=(m-\tfrac{1}{2})\epsilon$, then it is easy to see that the
square of the Hilbert-Schmidt norm of $\mat{E}$ is a Riemann sum approximation
of a certain double integral:
\begin{equation}
\lim_{\epsilon\downarrow 0}\|\mat{E}\|_{\mathrm{HS}}^2 = 
\iint_{[0,M]^2}e_0(a,b)\,da\,db,
\end{equation}
provided the double integral exists, where
\begin{equation}
e_0(a,b):=4\left[\frac{\sqrt{m^{-1}(a)m^{-1}(b)}}{m^{-1}(a)-m^{-1}(b)}-
\frac{\sqrt{m^{-1}(a)F(m^{-1}(a))m^{-1}(b)F(m^{-1}(b))}}{a-b}\right]^2,
\end{equation}
and where $m^{-1}(\cdot)$ denotes the inverse function to the monotone
function $m(\cdot)$ given by
\begin{equation}
m(\lambda):=\int_{-L}^\lambda F(\lambda')\,d\lambda'.
\end{equation}
By changing variables to $\kappa = m^{-1}(a)$ and $\lambda = m^{-1}(b)$, 
\begin{equation}
\iint_{[0,M]^2}e_0(a,b)\,da\,db = \iint_{[-L,0]^2}e(\kappa,\lambda)\,d\kappa\,
d\lambda,
\label{eq:dblinteqn}
\end{equation}
where
\begin{equation}
e(\kappa,\lambda):=4\left[\frac{\sqrt{\kappa\lambda}}{\kappa-\lambda}-
\frac{\sqrt{\kappa F(\kappa)\lambda F(\lambda)}}{m(\kappa)-m(\lambda)}
\right]^2F(\kappa)F(\lambda).
\end{equation}
Note that since $F\ge 0$ by Lemma~\ref{Lemma for function F},
$e(\kappa,\lambda)\ge 0$ for $(\kappa,\lambda)\in [-L,0]^2$.  To
complete the proof of the Lemma it is enough to show that the double
integral on the right-hand side of \eqref{eq:dblinteqn} is finite.

In order to estimate the double integral, we divide the square
$[-L,0]^2$ into polygonal regions as follows (see
Figure~\ref{fig:Regions}):
\begin{itemize}
\item
  The square $[-L,-L+\delta]^2$ contains those ordered pairs
  $(\kappa,\lambda)$ for which both $\kappa$ and $\lambda$ are near
  the ``soft edge'' of the eigenvalue spectrum at $-L$.  We divide
  this square into diagonal and off-diagonal parts according to
  whether $(\kappa+L)/2\le\lambda+L\le 2(\kappa+L)$ 
(the diagonal part, $S_\mathrm{D}$)
or not (the off-diagonal parts, $S_\mathrm{OD}$).
\item The square $[-\delta,0]^2$ contains those ordered pairs
$(\kappa,\lambda)$ for which both $\kappa$ and $\lambda$ are near the
``hard edge'' of the eigenvalue spectrum at $0$.  We divide this square
into diagonal and off-diagonal parts according to whether
$2\kappa<\lambda<\kappa/2$ (the diagonal part, $H_\mathrm{D}$)
or not (the off-diagonal parts $H_\mathrm{OD}$).
\item The remaining part of $[-L,0]^2$ contains those ordered pairs
$(\kappa,\lambda)$ for which at least one of the coordinates lies in the
``bulk'' of the eigenvalue spectrum, bounded away from both edges.  This
is divided into a diagonal part $B_\mathrm{D}$ and two off-diagonal 
parts $B_\mathrm{OD}$ along two straight line segments parallel to the
diagonal as indicated in
Figure~\ref{fig:Regions}.
\end{itemize}
\begin{figure}[h]
\begin{center}
\includegraphics{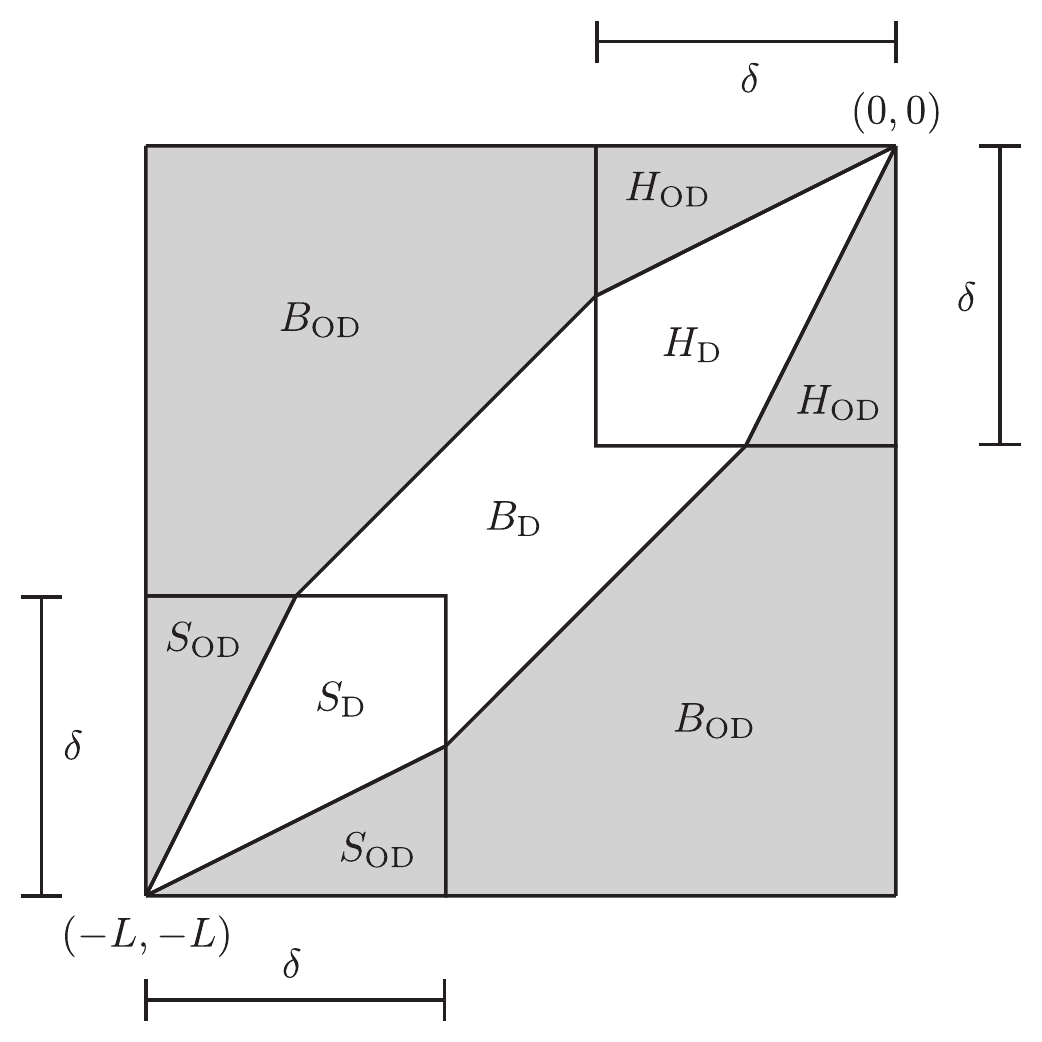}
\end{center}
\caption{\emph{The square $[-L,0]^2$ in the $(\kappa,\lambda)$-plane is covered
by the six regions $S_\mathrm{D}$, $S_\mathrm{OD}$, $B_\mathrm{D}$, $B_\mathrm{OD}$,
$H_{D}$, and $H_{OD}$.}}
\label{fig:Regions}
\end{figure}
Here, the constant $\delta>0$ is as
specified in Lemma~\ref{Lemma for function F}.  
As $e(\kappa,\lambda)=e(\lambda,\kappa)$ it will be enough to show integrability
of $e$ over the part of $[-L,0]^2$ with $\kappa<\lambda$, an inequality that
we will assume tacitly below.  

First we consider integrating $e(\kappa,\lambda)$ over 
the ``off-diagonal'' shaded regions $S_\mathrm{OD}$, $B_\mathrm{OD}$, and 
$H_\mathrm{OD}$ shown in Figure~\ref{fig:Regions}.
An upper bound for $e(\kappa,\lambda)$ useful 
in these regions is easily obtained from the inequality $(a-b)^2\le 2a^2+2b^2$:
\begin{equation}
e(\kappa,\lambda)\le 8\kappa F(\kappa)\lambda F(\lambda)
\left[\frac{1}{(\kappa-\lambda)^2} + \frac{F(\kappa)F(\lambda)}
  {(m(\kappa)-m(\lambda))^2}\right],\quad (\kappa,\lambda)\in (-L,0)^2.
\end{equation}
Applying the Mean Value Theorem to this estimate yields
\begin{equation}
e(\kappa,\lambda)\le \frac{8\kappa F(\kappa)\lambda F(\lambda)}{(\kappa-\lambda)^2}\left[1+\frac{F(\kappa)F(\lambda)}{F(\xi)^2}\right],
\end{equation}
where $\kappa\le\xi\le\lambda$.  Finally, since $F$ is monotone increasing 
according to Lemma~\ref{Lemma for
  function F} we obtain
\begin{equation}
e(\kappa,\lambda)\le \frac{8\kappa F(\kappa)\lambda F(\lambda)}
{(\kappa-\lambda)^2}\left[1+\frac{F(\lambda)}{F(\kappa)}\right] = 
\frac{8\kappa F(\kappa)\lambda F(\lambda)}{(\kappa-\lambda)^2} +
\frac{8\kappa \lambda F(\lambda)}{(\kappa-\lambda)^2}F(\lambda) .
\label{eq:eboundOD}
\end{equation}

Now, for $(\kappa,\lambda)\in B_\mathrm{OD}$, we have that
$\kappa-\lambda$ is bounded away from zero while by Lemma~\ref{Lemma
  for function F} $\kappa F(\kappa)$ and $\lambda F(\lambda)$ are
bounded (and of course $|\kappa|<L$) while $F(\lambda)$ is integrable.
Hence we easily conclude that $e(\kappa,\lambda)$ is integrable on
$B_\mathrm{OD}$.

If $(\kappa,\lambda)\in H_\mathrm{OD}$ with $\kappa<\lambda$, then we have
the inequality
\begin{equation}
(\kappa-\lambda)^2 = \left(\left[\lambda-\frac{\kappa}{2}\right] +
\left[-\frac{\kappa}{2}\right]\right)^2\ge \frac{\kappa^2}{4},
\end{equation}
and also since both $-\delta<\kappa<0$ and $-\delta<\lambda<0$ we may
use the upper bound for $F$ given in  \eqref{Lemma for function F equation 4}
from Lemma~\ref{Lemma for function F} to replace \eqref{eq:eboundOD}
with
\begin{equation}
e(\kappa,\lambda)\le 32C_0^2(-\kappa)^{-1-1/q}(-\lambda)^{1-1/q} +32C_0^2(-\kappa)^{-1}(-\lambda)^{1-2/q},
\end{equation}
where $C_0>0$ and $q>1$ are the constants in \eqref{Lemma for function
  F equation 4}.  This estimate is easily seen to be integrable on the
component of $H_\mathrm{OD}$ with $\kappa<\lambda$ by direct calculation of
the iterated integrals.

If $(\kappa,\lambda)\in S_\mathrm{OD}$ with $\kappa<\lambda$, then we have the
inequality
\begin{equation}
(\kappa-\lambda)^2 = \left(\left[\frac{\lambda+L}{2}\right]+
\left[\frac{\lambda+L}{2}-(\kappa+L)\right]\right)^2\ge
\frac{(\lambda+L)^2}{4},
\end{equation}
and also since both $-L<\kappa<-L+\delta$ and $-L<\lambda<-L+\delta$
we may use the upper bound for $F$ given in \eqref{Lemma for function
  F equation 2} from Lemma~\ref{Lemma for function F} along with the
inequalities $|\kappa|<L$ and $|\lambda|<L$ to replace
\eqref{eq:eboundOD} with
\begin{equation}
e(\kappa,\lambda)\le 32 L^2C_{-L}^2(\kappa+L)^{1/2}(\lambda+L)^{-3/2} +
32 L^2C_{-L}^2(\lambda+L)^{-1}.
\end{equation}
This upper bound is obviously integrable on the component of $S_\mathrm{OD}$
with $\kappa<\lambda$.

Now we consider integrating $e(\kappa,\lambda)$ over the ``diagonal''
unshaded regions $S_\mathrm{D}$, $B_\mathrm{D}$, and $H_\mathrm{D}$
shown in Figure~\ref{fig:Regions}.
By the Mean Value Theorem and the monotonicity of
$F$ guaranteed by Lemma~\ref{Lemma for function F}, we obtain an upper bound
more useful when $\kappa\approx\lambda$:
\begin{equation}
e(\kappa,\lambda)\le 4\kappa F(\kappa)\lambda F(\lambda)
\left[\frac{F(\kappa)-F(\lambda)}{m(\kappa)-m(\lambda)}\right]^2,
\quad (\kappa,\lambda)\in (-L,0)^2.
\end{equation}
Again using the Mean Value Theorem and monotonicity of $F$ we may make
the upper bound larger for $\kappa<\lambda$:
\begin{equation}
e(\kappa,\lambda)\le \frac{4\kappa\lambda F(\lambda)F'(\xi)^2}{F(\kappa)},
\label{eq:eboundD}
\end{equation}
where $\kappa\le \xi\le\lambda$.  

For $(\kappa,\lambda)\in B_\mathrm{D}$ with
$\kappa<\lambda$, both $\kappa$ and $\lambda$ are bounded away from
the soft and hard edges of the eigenvalue spectrum, so 
Lemma~\ref{Lemma for function F} guarantees that $F$ and $F'$ are bounded,
and $F$ is also bounded away from zero by strict monotonicity and 
the boundary condition $F(-L)=0$.  It follows from \eqref{eq:eboundD}
that $e(\kappa,\lambda)$
is bounded and hence integrable on $B_\mathrm{D}$.  

If $(\kappa,\lambda)\in H_\mathrm{D}$ then we may use the estimates
\eqref{Lemma for function F equation 4} and \eqref{Lemma for function
  F equation 3} from Lemma~\ref{Lemma for function F} to replace
\eqref{eq:eboundD} with
\begin{equation}
e(\kappa,\lambda)\le \frac{8C_0^2}{q^2}(-\kappa)^{1+1/q}(-\lambda)^{1-1/q}(-\xi)^{-2/q-2}\le \frac{8C_0^2}{q^2}(-\kappa)^{1+1/q}(-\lambda)^{-1-3/q}.
\end{equation}
The double integral of this upper bound over the region $H_\mathrm{D}$
with $\kappa<\lambda$ is easily computed by iterated integration and
is clearly finite as a consequence of the fact that $q>1$.

Finally, if $(\kappa,\lambda)\in S_\mathrm{D}$ with $\kappa<\lambda$, then
we may use the estimates \eqref{Lemma for function F equation 2} and 
\eqref{Lemma for function F equation 1} from Lemma~\ref{Lemma for function F}
together with the inequalities $|\kappa|<L$ and $|\lambda|<L$ 
to replace \eqref{eq:eboundD} with
\begin{equation}
e(\kappa,\lambda)\le 2L^2C_{-L}^2(\kappa+L)^{-1/2}(\lambda+L)^{1/2}(\xi+L)^{-1}
\le 2L^2C_{-L}^2(\kappa+L)^{-3/2}(\lambda+L)^{1/2},
\end{equation}
an upper bound that is clearly integrable over the part of $S_\mathrm{D}$
with $\kappa<\lambda$.  
\end{proof}

\begin{lemma}
The measure $\mu_\epsilon$ converges in the weak-$*$ sense to $\mu$,
uniformly for $(x,t)$ in compact sets. 
That is, for each continuous function $f:\mathbb{R}\to\mathbb{C}$,
\begin{equation}
\lim_{\epsilon\downarrow 0}\int_\mathbb{R}f(\alpha)\,d\mu_\epsilon(\alpha)
= \int_\mathbb{R}f(\alpha)\,d\mu(\alpha),
\end{equation}
with the limit being uniform with respect to $(x,t)$ in compact sets.
\label{lem:weakstar}
\end{lemma}
\begin{proof}
According to Lemma~\ref{lem:limitingmoments}, for each polynomial
$p(\alpha)$ we have the following limit, uniform for $(x,t)$ in compact sets:
\begin{equation}
\lim_{\epsilon\downarrow 0}\int_\mathbb{R}p(\alpha)\,d\mu_\epsilon(\alpha) =
\int_\mathbb{R}p(\alpha)\,d\mu(\alpha).
\label{eq:polyconvergence}
\end{equation}
But by Lemma~\ref{lem:tight} we can equivalently integrate over
the compact interval $\Omega$ (independent of $(x,t)$ in any given compact set)
with the same result.  Now by the Weierstra\ss\
Approximation Theorem, given any continuous function $f:\mathbb{R}\to\mathbb{C}$
and any $\rho>0$ there is a polynomial $p^f_\rho(\alpha)$ for which
\begin{equation}
\sup_{\lambda\in\Omega}|f(\alpha)-p^f_\rho(\alpha)|<\frac{\rho}{M},
\end{equation}
so for any measure $\nu$ of mass $M$ with support in $\Omega$
(like $\mu_\epsilon$ and $\mu$),
\begin{equation}
\left|\int_\mathbb{R}[f(\alpha)-p_\rho^f(\alpha)]\,d\nu(\alpha)
\right|\le\int_\Omega|f(\alpha)-p_\rho^f(\alpha)|\,d\nu(\alpha) <\rho.
\label{eq:probmeasure}
\end{equation}
Let $\omega>0$ be an arbitrarily small positive number.  Then if we write
\begin{equation}
\nu[g]:=\int_\Omega g(\alpha)\,d\nu(\alpha),
\end{equation}
we have
\begin{equation}
\begin{split}
\left|\int_\mathbb{R}f(\alpha)\,d\mu_\epsilon(\alpha) -
\int_\mathbb{R}f(\alpha)\,d\mu(\alpha)\right| &=
\left|\mu_\epsilon[f] -
\mu[f]\right|\\
&=
\left|\left[\mu_\epsilon[p^f_{\omega/3}] -
\mu[p^f_{\omega/3}]\right] +
\mu_\epsilon[f-p^f_{\omega/3}]
-\mu[f-p^f_{\omega/3}]\right|\\
&\le
\left|\mu_\epsilon[p^f_{\omega/3}] -
\mu[p^f_{\omega/3}]\right| +
\left|\mu_\epsilon[f-p^f_{\omega/3}]
\right| +
\left|\mu[f-p^f_{\omega/3}]
\right|\\
&<\left|\mu_\epsilon[p^f_{\omega/3}] -
\mu[p^f_{\omega/3}]\right| + \frac{2}{3}\omega,
\end{split}
\end{equation}
with the last inequality following from \eqref{eq:probmeasure}.  But
with $\omega>0$ fixed, \eqref{eq:polyconvergence} implies that $\epsilon>0$
may be chosen sufficiently small, independently of $(x,t)$ in any given 
compact set, that
\begin{equation}
\left|\mu_\epsilon[p_{\omega/3}^f]-\mu[p_{\omega/3}^f]\right|<\frac{1}{3}\omega,
\end{equation}
which implies
\begin{equation}
\left|\int_\mathbb{R} f(\alpha)\,d\mu_\epsilon(\alpha)-
\int_\mathbb{R} f(\alpha)\,d\mu(\alpha)\right|<\omega
\end{equation}
thereby completing the proof.
\end{proof}

Now we are in a position to complete the proof of
Proposition~\ref{limit of sx theorem}.  
We begin by writing $\Ut(x,t)$ as defined by \eqref{small
  dispersion matrix calculation 3_1} in terms of the normalized (to
mass $M$) counting measure $\mu_\epsilon$:
\begin{equation}
\Ut(x,t) = 
\left[\frac{\epsilon N(\epsilon)}{M}\right]\int_\mathbb{R}2\arctan(\epsilon^{-1}\alpha)\,d\mu_\epsilon(\alpha).
\end{equation}
Define the continuous functions
\begin{equation}
a_+(\alpha):=\pi+4H(-\alpha)\arctan(\alpha),\quad
a_-(\alpha):=-a_+(-\alpha),\quad\alpha\in\mathbb{R},
\end{equation}
where $H(\cdot)$ denotes the Heaviside step function.
It is then easy to check (see Figure~\ref{fig:ArcTanCurves}) that for
any $E>0$,
\begin{figure}[h]
\begin{center}
\includegraphics{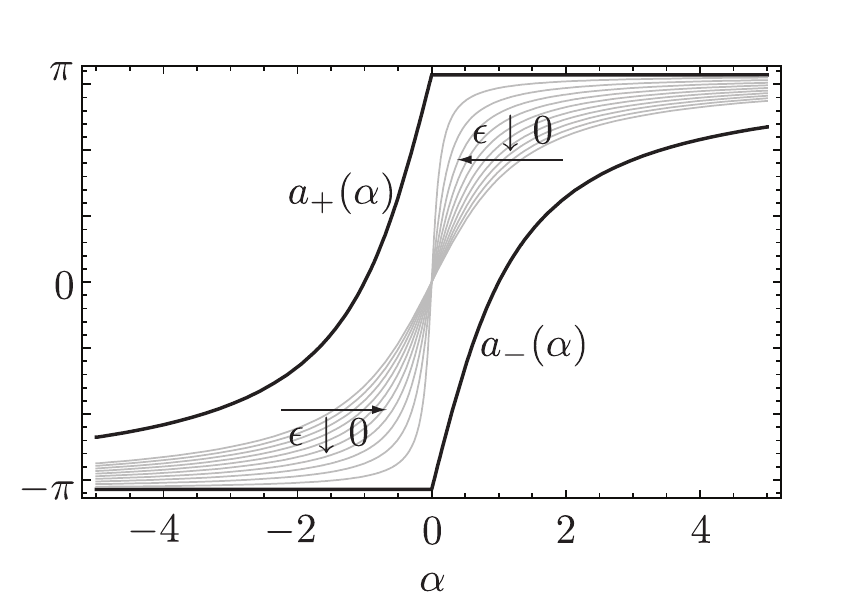}
\end{center}
\caption{\emph{The graphs of $a_-(\alpha)<a_+(\alpha)$ (black) and several
graphs of $2\arctan(\epsilon^{-1}\alpha)$ for $\epsilon\le 1$ (gray).}}
\label{fig:ArcTanCurves}
\end{figure}
\begin{equation}
0<\epsilon\le E\quad \implies\quad a_-(E^{-1}\alpha)\le
2\arctan(\epsilon^{-1}\alpha)\le a_+(E^{-1}\alpha),\quad\alpha\in\mathbb{R}.
\end{equation}
Therefore, for any $E>0$ and all $0<\epsilon<E$,
\begin{equation}
\int_\mathbb{R}a_-(E^{-1}\alpha)\,d\mu_\epsilon(\alpha)\le \int_\mathbb{R}
2\arctan(\epsilon^{-1}\alpha)\,d\mu_\epsilon(\alpha)\le \int_\mathbb{R}
a_+(E^{-1}\alpha)\,d\mu_\epsilon(\alpha).
\end{equation}
Using Lemma~\ref{lem:weakstar} we may pass to the limit
$\epsilon\downarrow 0$ in the lower and upper bounds to obtain
\begin{equation}
\liminf_{\epsilon\downarrow 0}\int_\mathbb{R}2\arctan(\epsilon^{-1}\alpha)\,
d\mu_\epsilon(\alpha)\ge \int_\mathbb{R}a_-(E^{-1}\alpha)\,d\mu(\alpha)
\label{eq:liminf}
\end{equation}
and also
\begin{equation}
\limsup_{\epsilon\downarrow 0}\int_\mathbb{R}2\arctan(\epsilon^{-1}\alpha)\,
d\mu_\epsilon(\alpha)\le \int_\mathbb{R}a_+(E^{-1}\alpha)\,d\mu(\alpha).
\label{eq:limsup}
\end{equation}
In these statements, $E>0$ is an arbitrary parameter, and the limits
are uniform for $(x,t)$ in compact sets.  But
$a_\pm(E^{-1}\alpha)$ are uniformly bounded functions that both tend
pointwise for $\alpha\neq 0$ to the same limit function
$\pi\,\mathrm{sgn}(\lambda)$ as
$E\downarrow 0$, while $\mu$ is a fixed measure that is
absolutely continuous with respect to Lebesgue measure on $\mathbb{R}$, so
by the Lebesgue Dominated Convergence Theorem,
\begin{equation}
\lim_{E\downarrow 0}\int_\mathbb{R}a_\pm(E^{-1}\alpha)\,d\mu(\alpha)=
\int_\mathbb{R}\pi\,\mathrm{sgn}(\alpha)\,d\mu(\alpha).
\end{equation}
By letting $E\downarrow 0$, it then follows from \eqref{eq:liminf}
and \eqref{eq:limsup} that
\begin{equation}
\lim_{\epsilon\downarrow 0}\int_\mathbb{R}2\arctan(\epsilon^{-1}\alpha)\,
d\mu_\epsilon(\alpha) = \int_\mathbb{R}\pi\,\mathrm{sgn}(\alpha)\,d\mu(\alpha)
\label{eq:limintegral}
\end{equation}
with the limit being uniform for $(x,t)$ in any given compact set.
Finally, 
according to \eqref{eq:epsilonNepsilon}, 
we have (independent of $x$ and $t$)
\begin{equation}
\lim_{\epsilon\downarrow 0}\frac{\epsilon N(\epsilon)}{M}=1,
\end{equation}
so combining this result with \eqref{eq:limintegral} and noting that
$d\mu(\alpha)=G(\alpha;x,t)\,d\alpha$ completes the proof of
Proposition~\ref{limit of sx theorem}.

\subsection{Differentiation of $\Ut$.  Burgers' equation and 
weak convergence of $\ut$}
\label{sec:differentiation}
Let $\phi\in\mathscr{D}(\mathbb{R})$ be a test function.  Then by
integration by parts and the uniform convergence of $\Ut(x,t)$
to $U(x,t)$ on compact sets in the $(x,t)$-plane guaranteed by
Proposition~\ref{limit of sx theorem},
\begin{equation}
\begin{split}
\lim_{\epsilon\downarrow 0}\int_\mathbb{R}\ut(x,t)\phi(x)\,dx &= 
\lim_{\epsilon\downarrow 0}\int_\mathbb{R}\frac{\partial \Ut}{\partial x}(x,t)\phi(x)\,dx\\
&=-\lim_{\epsilon\downarrow 0}\int_\mathbb{R}\Ut(x,t)\phi'(x)\,dx\\
&=-\int_\mathbb{R}U(x,t)\phi'(x)\,dx.
\end{split}
\label{eq:uUdist}
\end{equation}

\begin{lemma}
The limit function $U(x,t)$ is continuously differentiable with 
respect to $x$,
and if $(x,t)$ is a point for which there are $2P(x,t)+1$ solutions
$u_0^\mathrm{B}(x,t)<\cdots<u_{2P(x,t)}^\mathrm{B}(x,t)$ 
of the
implicit equation \eqref{eq:implicitBurgers},
\begin{equation}
\frac{\partial U}{\partial x}(x,t) = \sum_{n=0}^{2P(x,t)}
(-1)^nu_n^\mathrm{B}(x,t),
\label{eq:Ux}
\end{equation}
and the above formula is extended to nongeneric $(x,t)$ by continuity.
\end{lemma}

\begin{proof}
Exchanging the order of integration in
the double-integral formula for $U(x,t)$ obtained by substituting
$d\mu(\alpha)=G(\alpha;x,t)\,d\alpha$ with $G$ given by \eqref{limit
  of sx 2} into \eqref{eq:Uformula}, we obtain
\begin{equation}
U(x,t)=\int_{-L}^0J(\lambda;x,t)\,d\lambda,
\label{eq:Urepn}
\end{equation}
where 
\begin{equation}
J(\lambda;x,t):=-\frac{1}{4\lambda}
\int_{-2\lambda(x+2\lambda t-x_+(\lambda))}^{-2\lambda(x+2\lambda t -x_-(\lambda))}\mathrm{sgn}(\alpha)\,d\alpha.
\end{equation}
Note that for $\lambda\in [-L,0]$ the upper limit of integration 
is greater than or equal
to the lower limit.  Moreover
the integral in $J(\lambda;x,t)$ is easily evaluated; for $-L<\lambda<0$,
\begin{equation}
J(\lambda;x,t)=\begin{cases}
-\pi F(\lambda),\quad &
x+2\lambda t-x_-(\lambda)<0,\\
x+2\lambda t+\gamma(\lambda),\quad &
x+2\lambda t-x_+(\lambda)\le 0\le
x+2\lambda t-x_-(\lambda)\\
\pi F(\lambda),\quad & x+2\lambda t-x_+(\lambda)>0.
\end{cases}
\label{eq:Jcases}
\end{equation}
It follows from the relations $x_\pm(\lambda)=\pm\pi
F(\lambda)-\gamma(\lambda)$ that for an admissible initial condition
$u_0$, $J$ is a continuous function of $x$ for each fixed $t$,
uniformly with respect to $ \lambda\in [-L,0]$, and hence also from
\eqref{eq:Urepn} that $U(\cdot,t)$ is continuous on $\mathbb{R}$ for
each $t$.  To prove that $U(\cdot,t)$ is continuously differentiable
it will therefore suffice to establish continuous differentiability on
the complement of a finite set of points and that the resulting
piecewise formula for $\partial U/\partial x$ extends continuously
to the whole real line.

To use the formula \eqref{eq:Jcases} 
in the representation \eqref{eq:Urepn} we therefore
need to know those points $\lambda\in (-L,0)$ at which one of the two
quantities $x+2\lambda t-x_+(\lambda)<x+2\lambda t-x_-(\lambda)$ changes
sign.  Under the variable substitution $\lambda =-u^\mathrm{B}$, the definition
of the turning points $x_\pm(\lambda)$ as branches of the inverse function
of $u_0$ implies that the union of solutions
of the two equations $x+2\lambda t-x_\pm(\lambda)=0$ is exactly the totality
of solutions of the implicit equation
\begin{equation}
u^\mathrm{B} = u_0(x-2u^{\mathrm{B}}t).
\label{eq:implicitBurgers}
\end{equation}
In other words, the transitional points $\lambda$ for the formula
\eqref{eq:Jcases} correspond under the sign change $u^\mathrm{B}=-\lambda$
to the branches of the multivalued solution of Burgers' equation
\begin{equation}
\frac{\partial u^\mathrm{B}}{\partial t}+2u^\mathrm{B}
\frac{\partial u^\mathrm{B}}{\partial x}=0
\end{equation}
subject to the admissible initial condition $u^\mathrm{B}(x,0)=u_0(x)$.  

Note that admissibility of $u_0$ implies (see Definition~\ref{def:admissible})
that given any $t\in\mathbb{R}$ there exist only a finite
number of breaking points $(x_\xi,t_\xi)$ with $t_\xi$ in the closed interval
between $0$ and $t$. 
Indeed, the breaking points correspond to values of $\xi\in\mathbb{R}$
for which $u_0''(\xi)=0$ but $u_0'''(\xi)\neq 0$, and the breaking times
are $t_\xi=(-2u_0'(\xi))^{-1}$; since $u_0'(\xi)$ decays to zero for large $\xi$,
bounded breaking times $t_\xi$ correspond to bounded $\xi$, and there are
only finitely many of these by hypothesis.  Moreover, each breaking point
$(x_\xi,t_\xi)$ generates a new fold in the solution surface lying between 
two caustic
curves emerging in the direction of increasing $|t|$ from $(x_\xi,t_\xi)$, and
because $u'''(\xi)\neq 0$ 
there are exactly two more sheets of the multivalued solution of 
Burgers' equation born within the fold as a result of a simple pitchfork
bifurcation.
Therefore, the union of caustic curves and breaking points meets any line of 
constant
$t$ in the $(x,t)$-plane
in a finite set of points $\{x_j^\mathrm{crit}(t)\}$, and 
on every connected component
of the set $S_t:=\{(x,t)|x\in\mathbb{R}\setminus\{x_j^\mathrm{crit}(t)\}\}$,
there is a finite, odd,
and constant (with respect to $x$) number $2P(x,t)+1$ of roots of the equation
\eqref{eq:implicitBurgers}, and all roots are simple (and hence differentiable
with respect to $x$).  

If $t\ge 0$, then by admissibility of $u_0$ the quantity
$b_-(\lambda;x,t):=x+2\lambda t-x_-(\lambda)$ is strictly increasing
as a function of $\lambda$ on the interval $(-L,0)$, and therefore in
this interval there can exist at most one root of $b_-(\lambda;x,t)$,
regardless of the value of $x\in\mathbb{R}$.  Moreover,
$b_-(\lambda;x,t)\to +\infty$ as $\lambda\uparrow 0$, so there will be
exactly one root in $(-L,0)$ if $b_-(-L;x,t)=x-x_0-2Lt<0$ and no root
in $(-L,0)$ if $x-x_0-2Lt>0$.  Since $b_+(\lambda;x,t):=x+2\lambda
t-x_+(\lambda)<b_-(\lambda;x,t)$ for $-L<\lambda<0$, if $x-x_0-2Lt<0$,
all roots of $b_+(\lambda;x,t)$ in $(-L,0)$ must lie to the right
of the root of $b_-(\lambda;x,t)$.  Thus, for $x\in S_t\setminus\{x_0+2Lt\}$,
we either have
\begin{equation}
\begin{split}
U(x,t)&=\int_{-u_0^\mathrm{B}}^0(x+2\lambda t+\gamma(\lambda))\,d\lambda\\
&\quad\quad\quad{}
+\sum_{p=1}^{P(x,t)}\left[\pi\int_{-u_{2p-1}^\mathrm{B}}^{-u_{2p-2}^\mathrm{B}}
F(\lambda)\,d\lambda +\int_{-u_{2p}^\mathrm{B}}^{-u_{2p-1}^\mathrm{B}}
(x+2\lambda t+\gamma(\lambda))\,d\lambda\right] \\
&\quad\quad\quad\quad\quad{}
+\pi\int_{-L}^{-u_{2P(x,t)}^\mathrm{B}}F(\lambda)\,d\lambda,\quad 
x\in S_t,\quad x>x_0+2Lt,
\end{split}
\end{equation}
in which case $u_0^\mathrm{B}(x,t)<\cdots <u_{2P(x,t)}^\mathrm{B}(x,t)$
are all roots of $b_+(-u^\mathrm{B};x,t)$, 
or
\begin{equation}
\begin{split}
U(x,t)&=\int_{-u_0^\mathrm{B}}^0(x+2\lambda t+\gamma(\lambda))\,d\lambda\\
&\quad\quad\quad{}
+\sum_{p=1}^{P(x,t)}\left[\pi\int_{-u_{2p-1}^\mathrm{B}}^{-u_{2p-2}^\mathrm{B}}
F(\lambda)\,d\lambda +\int_{-u_{2p}^\mathrm{B}}^{-u_{2p-1}^\mathrm{B}}
(x+2\lambda t+\gamma(\lambda))\,d\lambda\right] \\
&\quad\quad\quad\quad\quad{}
-\pi\int_{-L}^{-u_{2P(x,t)}^\mathrm{B}}F(\lambda)\,d\lambda,\quad x\in S_t,
\quad x<x_0+2Lt,
\end{split}
\end{equation}
in which case $u_0^\mathrm{B}(x,t)<\cdots<u_{2P(x,t)-1}^\mathrm{B}(x,t)$
are roots of $b_+(-u^\mathrm{B};x,t)$ while $u_{2P(x,t)}^\mathrm{B}(x,t)$
with $u_{2P(x,t)}^\mathrm{B}(x,t)>u_{2P(x,t)-1}^\mathrm{B}(x,t)$ is a 
root of $b_-(-u^\mathrm{B};x,t)$.  In both cases, the condition $x\in S_t$
guarantees that all roots are differentiable with respect to $x$, so
we may calculate $\partial U/\partial x$ by Leibniz' rule:
\begin{equation}
\begin{split}
\frac{\partial U}{\partial x}(x,t) &= 
b_+(-u_{2P}^\mathrm{B}(x,t);x,t)
\frac{\partial u_{2P}^\mathrm{B}}{\partial x}(x,t) +
\sum_{n=0}^{2P-1} (-1)^nb_+(-u_n^\mathrm{B}(x,t);x,t)
\frac{\partial u_n^\mathrm{B}}{\partial x}(x,t) \\
&\quad\quad{}+
\sum_{n=0}^{2P} (-1)^nu_n^\mathrm{B}(x,t),\quad x\in S_t,\quad x>x_0+2Lt,
\end{split}
\end{equation}
or
\begin{equation}
\begin{split}
\frac{\partial U}{\partial x}(x,t) &= 
b_-(-u_{2P}^\mathrm{B}(x,t);x,t)
\frac{\partial u_{2P}^\mathrm{B}}{\partial x}(x,t) +
\sum_{n=0}^{2P-1} (-1)^nb_+(-u_n^\mathrm{B}(x,t);x,t)
\frac{\partial u_n^\mathrm{B}}{\partial x}(x,t) \\
&\quad\quad{}+
\sum_{n=0}^{2P} (-1)^nu_n^\mathrm{B}(x,t),\quad x\in S_t,\quad x<x_0+2Lt,
\end{split}
\end{equation}
where in both cases $P=P(x,t)$ is a constant nonnegative integer on each
connected component of $S_t$.  The terms on the first line in each of
these formulae arise from differentiating the limits of integration and
using $x_\pm(\lambda)=\pm \pi F(\lambda)-\gamma(\lambda)$, while the terms
on the second line arise from the explicit partial differentiation of the
integrand $x+2\lambda t+\gamma(\lambda)$ with respect to $x$.  It follows
from our division of the solutions of \eqref{eq:implicitBurgers}
among the roots of $b_+$ and $b_-$ that in both cases the terms on the
first line vanish identically, with the result that
\begin{equation}
\frac{\partial U}{\partial x}(x,t) = \sum_{n=0}^{2P(x,t)}(-1)^nu_n^\mathrm{B}(x,t),
\quad x\in S_t\setminus\{x_0+2Lt\}.
\label{eq:dUdxpen}
\end{equation}
This expression is clearly continuous in $x$ on each connected
component of $S_t\setminus\{x_0+2Lt\}$.  Moreover, it extends
continuously to the finite complement in $\mathbb{R}_x$ (at fixed
$t\ge 0$) because at caustics pairs of solution branches entering into
\eqref{eq:dUdxpen} with opposite signs simply coalesce.  Therefore
$U(\cdot,t)$ is indeed continuously differentiable for $t\ge 0$ and
its derivative is given by the desired simple formula \eqref{eq:Ux}.  
Virtually the
same argument applies to $t\le 0$ with the roles of $b_\pm(\lambda;x,t)$
reversed, and the resulting formula for $\partial U/\partial x$ is the
same.
\end{proof}

It follows from this result that we may integrate by parts in \eqref{eq:uUdist}
and obtain
\begin{equation}
\lim_{\epsilon\downarrow 0}\int_\mathbb{R}\ut(x,t)\phi(x)\,dx = 
\int_\mathbb{R}\frac{\partial U}{\partial x}(x,t)\phi(x)\,dx
\label{eq:uUxdist}
\end{equation}
for every test function $\phi\in\mathscr{D}(\mathbb{R})$.
Now let $v\in L^2(\mathbb{R})$.  Since $\mathscr{D}(\mathbb{R})$ is dense
in $L^2(\mathbb{R})$, for each $\sigma>0$ there exists a test function
$\phi_\sigma\in\mathscr{D}(\mathbb{R})$ such that 
\begin{equation}
\|\phi_\sigma-v\|_2^2:= \int_\mathbb{R}|\phi_\sigma(x)-v(x)|^2\,dx<\sigma^2.
\end{equation}
Then,
\begin{equation}
\begin{split}
\int_\mathbb{R}\left[\ut(x,t)-\frac{\partial U}{\partial x}(x,t)
\right]v(x)\,dx & =
\int_\mathbb{R}\left[\ut(x,t)-\frac{\partial U}{\partial x}(x,t)
\right]\phi_\sigma(x)\,dx \\
&\quad\quad {} +\int_\mathbb{R}\frac{\partial U}{\partial x}(x,t)\left[\phi_\sigma(x)-v(x)\right]\,dx\\
&\quad\quad {}-\int_\mathbb{R}\ut(x,t)\left[\phi_\sigma(x)-v(x)\right]\,dx.
\end{split}
\label{eq:udiffweak}
\end{equation}

Observe that, according to the definition (see Definition~\ref{def:modified}) of
$\ut(x,t)$ in terms of the modified scattering data, it follows
from \eqref{eq:I2} and 
\eqref{eq:equivalentintegrals} that
\begin{equation}
\int_\mathbb{R}\ut(x,t)^2\,dx = -4\pi\epsilon\sum_{n=1}^{N(\epsilon)}
\lamt_n.
\end{equation}
This Riemann sum converges as $\epsilon\downarrow 0$:
\begin{equation}
\lim_{\epsilon\downarrow 0}\int_\mathbb{R}\ut(x,t)^2\,dx = -4\pi
\int_{-L}^0 \lambda F(\lambda)\,d\lambda = \int_\mathbb{R}u_0(x)^2\,dx,
\label{eq:limintuepssquared}
\end{equation}
where the second equality follows from the identities 
\eqref{eq:momentrelations}, which essentially define $F(\lambda)$ in 
terms of the admissible initial condition $u_0$.  Therefore, 
$\|\ut(\cdot,t)\|_2$ is bounded for sufficiently small $\epsilon$,
independently of $t$.  

Also, $\partial U/\partial x$ is independent of
$\epsilon$ and from the formula \eqref{eq:Ux} it is easy to check that
it is positive and bounded above by the constant $L$ for all $(x,t)$.  
Therefore
\begin{equation}
\left\|\frac{\partial U}{\partial x}(\cdot,t)\right\|_2^2
\le L\int_\mathbb{R}\frac{\partial U}{\partial x}(x,t)\,dx.
\end{equation}
By the formula \eqref{eq:Ux}, the latter integral is equal to the 
area between the graph of the
multivalued solution curve for Burgers' equation and the $x$-axis.
Since points on the graph at the same height move with the same speed, 
this area is independent of time $t$, and hence we have
\begin{equation}
\left\|\frac{\partial U}{\partial x}(\cdot,t)\right\|_2^2\le
2\pi LM,
\end{equation}
where the mass $M$ is defined in terms of the initial condition $u_0$
by \eqref{eq:massdef}.
In fact, for $0\le t<T$, where $T$ is the breaking
time, it follows from the
fact that $\partial U/\partial x$ as given by \eqref{eq:Ux} reduces to the 
classical solution $u_0^\mathrm{B}(x,t)$ 
of Burgers' equation with initial data $u_0$, 
which conserves exactly the
$L^2(\mathbb{R}_x)$ norm, that
\begin{equation}
\left\|\frac{\partial U}{\partial x}(\cdot,t)
\right\|^2_2 = \|u_0^B(\cdot,t)\|^2_2 = \int_\mathbb{R}u_0(x)^2\,dx,
\quad 0\le t<T.
\label{eq:conservationBurgers}
\end{equation}
We will use this fact below in \S\ref{sec:corollary} when we prove
Corollary~\ref{corr:strong}.  In any case, these considerations show
that for all $\epsilon>0$ sufficiently small 
there exists a constant $K>0$ independent of $t$ such that
\begin{equation}
\left\|\frac{\partial U}{\partial x}(\cdot,t)\right\|_2 +
\|\ut(\cdot,t)\|_2\le K
\end{equation}
holds for all $t\ge 0$.

Now, by Cauchy-Schwarz it follows that
\begin{equation}
\left|\int_\mathbb{R}\frac{\partial U}{\partial x}(x,t)
\left[\phi_\sigma(x)-v(x)\right]\,dx
-\int_\mathbb{R}\ut(x,t)\left[\phi_\sigma(x)-v(x)\right]\,dx\right|\le
K\|\phi_\sigma-v\|_2.
\end{equation}
Given $\omega>0$ arbitrarily small, we then choose $\sigma=\omega/(2M)$ and
then \eqref{eq:udiffweak} implies that
\begin{equation}
\left|\int_\mathbb{R}\left[\ut(x,t)-\frac{\partial U}{\partial x}(x,t)
\right]v(x)\,dx\right|
\le\left|\int_\mathbb{R}\left[\ut(x,t)-\frac{\partial U}{\partial x}(x,t)
\right]\phi_{\omega/(2M)}(x)\,dx\right| +\frac{\omega}{2}.
\end{equation}
Finally, since $\phi_{\omega/(2M)}$ is a test function 
independent of $\epsilon$,
we may use \eqref{eq:uUxdist} to choose $\epsilon>0$ so small that the
first term on the right-hand side is less than $\omega/2$.  

This proves that 
\begin{equation}
\mathop{\mathrm{w}_x\mathrm{-lim}}_{\epsilon\downarrow 0}
\ut(x,t) = 
\frac{\partial U}{\partial x}(x,t)
\label{eq:uweaklimit}
\end{equation}
(weak $L^2$ convergence) uniformly for $t$ in bounded intervals.  
Combining \eqref{eq:Ux} with \eqref{eq:uweaklimit} completes the proof 
of Theorem~\ref{MainTheorem}.

\section{Strong Convergence Before Breaking}
\label{sec:corollary}
In this brief section we give a proof of Corollary~\ref{corr:strong}, following
closely Lax and Levermore (see Theorem~4.5 in part II of \cite{Lax 1983-1}).
Starting from the identity
\begin{equation}
\|\ut(\cdot,t)-u_0^\mathrm{B}(\cdot,t)\|_2^2 = 
\int_\mathbb{R}\ut(x,t)^2\,dx + \int_\mathbb{R}u_0^\mathrm{B}(x,t)^2\,dx
-2\int_\mathbb{R}\ut(x,t)u_0^\mathrm{B}(x,t)\,dx,
\end{equation}
we note that for $0\le t<T$, where $T$ is the breaking time,
\eqref{eq:limintuepssquared} and \eqref{eq:conservationBurgers} imply
that
\begin{equation}
\lim_{\epsilon\downarrow 0}\|\ut(\cdot,t)-u_0^\mathrm{B}(\cdot,t)\|_2^2
=2\int_\mathbb{R}u_0(x)^2\,dx-2\lim_{\epsilon\downarrow 0}\int_\mathbb{R}
\ut(x,t)u_0^\mathrm{B}(x,t)\,dx.
\end{equation}
But $u_0^\mathrm{B}(\cdot,t)\in L^2(\mathbb{R})$ is independent of $\epsilon$,
so by Theorem~\ref{MainTheorem},
\begin{equation}
\lim_{\epsilon\downarrow 0}\int_\mathbb{R}\ut(x,t)u_0^\mathrm{B}(x,t)\,dx
= \int_\mathbb{R}u_0^\mathrm{B}(x,t)^2\,dx = 
\int_\mathbb{R}u_0(x)^2\,dx,
\end{equation}
with the second equality following from \eqref{eq:conservationBurgers} for
$0\le t <T$.
Therefore
\begin{equation}
\lim_{\epsilon\downarrow 0}\|\ut(\cdot,t)-u_0^\mathrm{B}(\cdot,t)\|_2
= 0
\end{equation}
as desired, and the proof is complete.

\section{Numerical Verification}
\label{sec:numerics}
To illustrate the weak convergence of $\ut(x,t)$ as guaranteed
by Theorem~\ref{MainTheorem}, and to attempt to empirically
quantify the rate of convergence, we have directly used the
exact formula \eqref{small dispersion matrix calculation 3_1} 
for $\Ut(x,t)$ having first
chosen the modified scattering data corresponding to the admissible
initial condition $u_0(x)=2(1+x^2)^{-1}$ as specified in 
Definition~\ref{def:modified}, and compared the result for several different
values of $\epsilon$ with the limiting formula \eqref{eq:Uformula}
for $U(x,t)$.
Our results are shown in Figure~\ref{fig:UUepsilon}.
\begin{figure}[h]
\begin{center}
\includegraphics{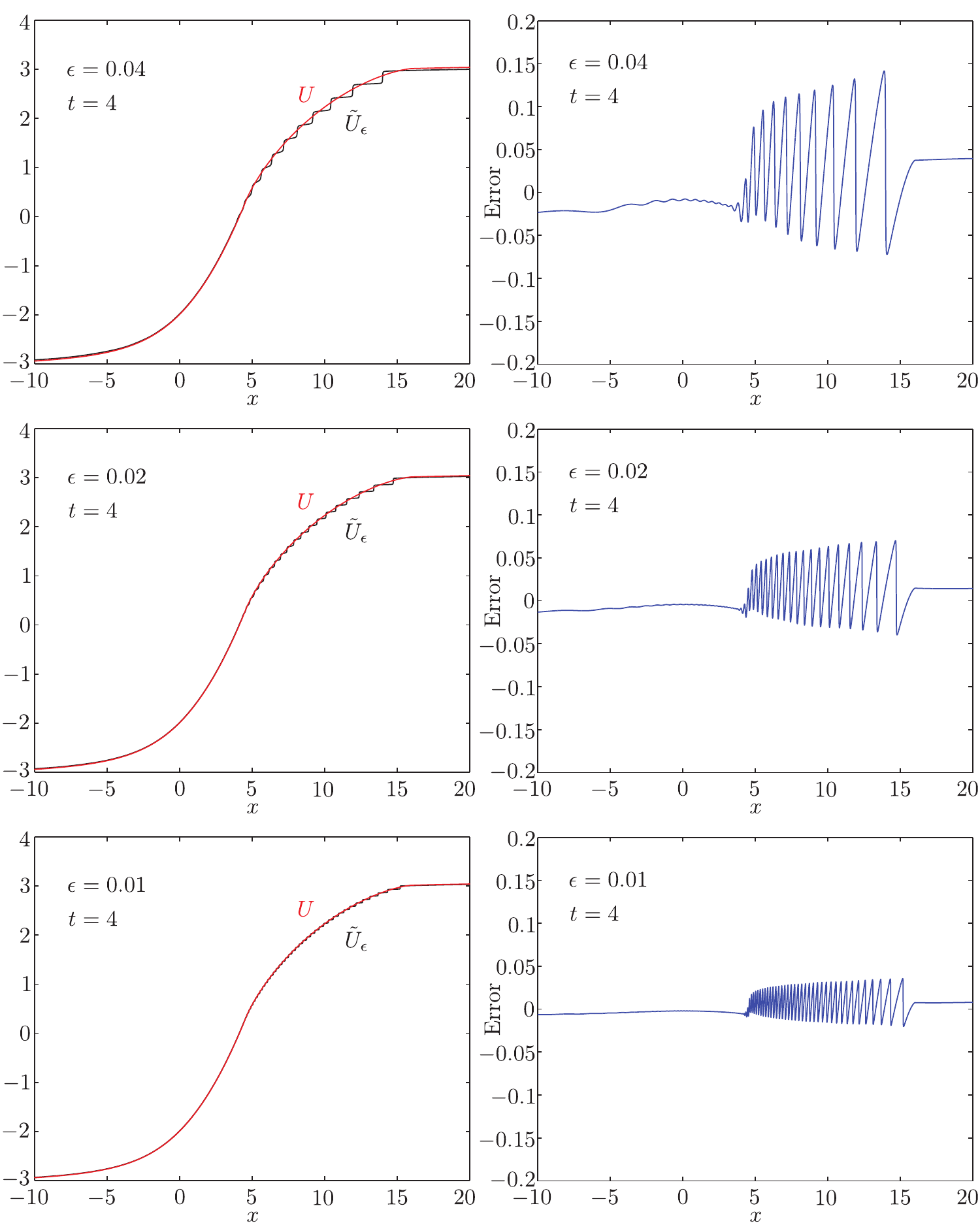}
\end{center}
\caption{\emph{Left: plots of $\Ut(x,t)$ (black) and its
    locally uniform limit $U(x,t)$ (red) at $t=4$ for various values
of $\epsilon$.  For these plots, $u_0(x):=2(1+x^2)^{-1}$.
Right:  corresponding plots of the error $U(x,t)-\Ut(x,t)$.
}}
\label{fig:UUepsilon}
\end{figure}
These plots clearly display the locally uniform convergence specified
in Proposition~\ref{limit of sx theorem}.  An interesting feature is the
apparent regular ``staircase'' form of the graph of $\Ut(x,t)$ 
as a function of $x$;
that the steps have nearly equal height is a consequence 
of the fact that near the leading edge of the oscillation
zone for $\ueps$ (which lies approximately in the range $4<x<16$ in these plots)
the undular bore wavetrain that is generated from the smooth initial data
resolves into a train of solitons of the BO equation,
each of which has a fixed mass proportional to $\epsilon$ (independent of 
amplitude and velocity).  

To the eye, the size of the error between $\Ut(x,t)$ and $U(x,t)$
appears to scale with $\epsilon$.  To confirm this more quantitatively,
we collected numerical data from several experiments, each performed with
a different value of $\epsilon$ at the fixed time $t=4$.  The supremum
norm, calculated over the interval $-10<x<20$, 
of the error resulting from each of these experiments is plotted in
Figure~\ref{fig:SupNormErrors}.  
\begin{figure}[h]
\begin{center}
\includegraphics{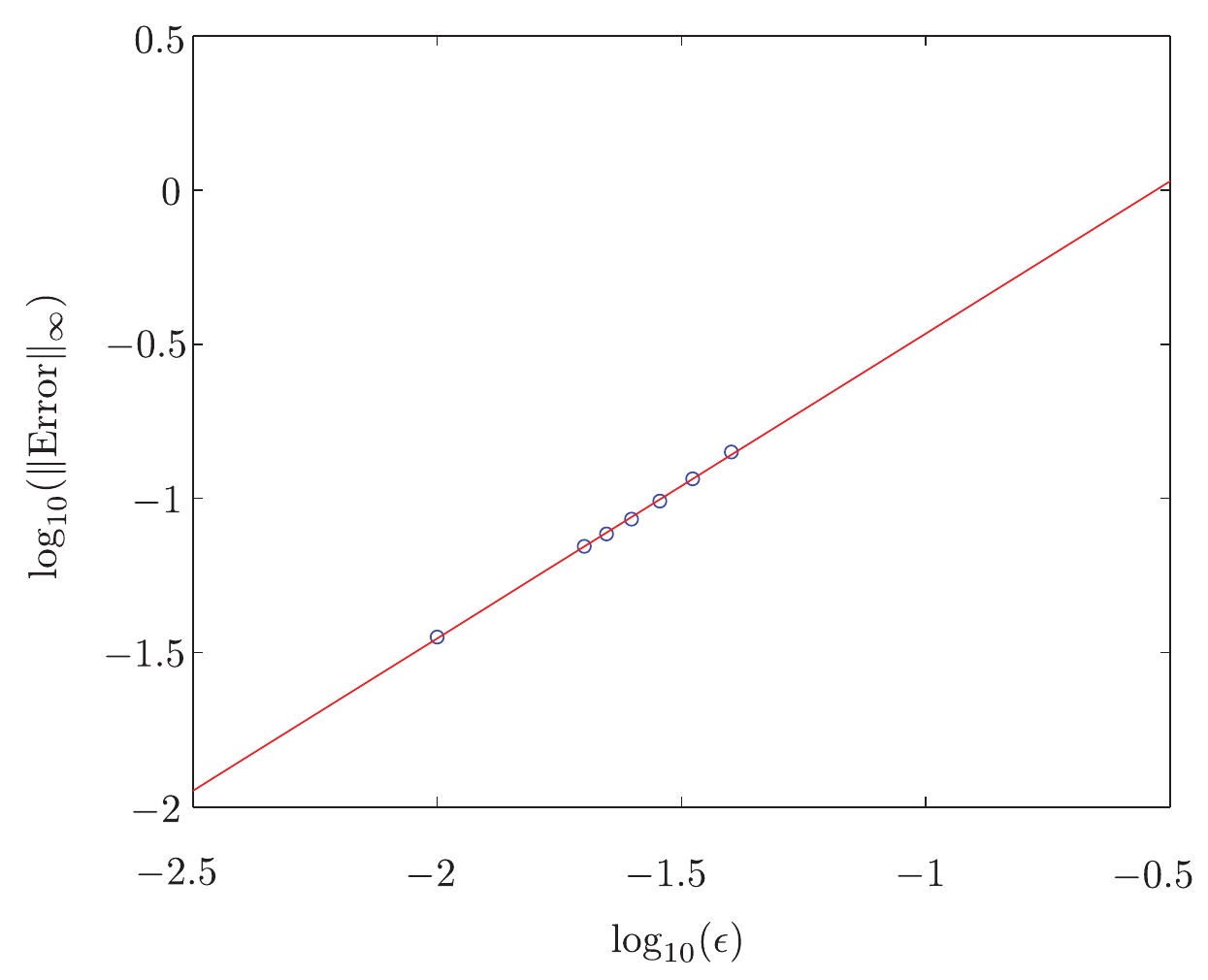}
\end{center}
\caption{\emph{Circles:  
$\log_{10}(\|\Ut(\cdot,4)-U(\cdot,4)\|_\infty)$ for
$\epsilon=1/25$, $1/30$, $1/35$, $1/40$, $1/45$, $1/50$, and $1/100$,
as a function of $\log_{10}(\epsilon)$.  In red:  The least-squares
linear fit.  }}
\label{fig:SupNormErrors}
\end{figure}
On this plot with logarithmic axes, the data points
appear to lie along a straight line, and we calculated the 
least squares linear fit to the data to be given by
\begin{equation}
\log_{10}(\|\Ut(\cdot,4)-U(\cdot,4)\|_\infty) = 
0.988 \log_{10}(\epsilon) + 0.523
\end{equation}
where the slope and intercept are given to three significant digits.
This strongly suggests a linear rate of convergence, in which the
error is asymptotically proportional to $\epsilon$ as
$\epsilon\downarrow 0$.

The initial data $u_0(x)=2(1+x^2)^{-1}$ was chosen for these
experiments because it is the only initial condition (up to a constant
multiple) for which the \emph{exact} scattering data is known for a
sequence of values of $\epsilon$ tending to zero.  This is the result
of a calculation of Kodama, Ablowitz, and Satsuma \cite{KodamaAS82},
who showed that if $u_0(x) = 2(1+x^2)^{-1}$, then
the reflection coefficient $\beta(\lambda)$ vanishes identically if
$\epsilon=1/N$ for any positive integer $N$.  Moreover, there are 
in this case exactly
$N$ eigenvalues $\lambda_1<\lambda_2<\cdots <\lambda_N$
of the operator $\mathcal{L}$ defined by 
\eqref{eq:Loperator}, and they are given implicitly by the equation
\begin{equation}
L_N\left(-\frac{2\lambda_n}{\epsilon}\right)=L_N\left(-2N\lambda_n\right) = 0,
\end{equation}
where $L_N$ is the Laguerre polynomial\footnote{The asymptotic 
density of zeros (here scaled by the factor $-2N$) 
of the Laguerre polynomials is well-known:
\[
F(\lambda) = \frac{1}{\pi}\sqrt{\frac{2+\lambda}{-\lambda}},\quad
-2<\lambda<0,
\]
a distribution also known in random matrix theory as the
Marchenko-Pastur law.  This asymptotic formula agrees exactly with
Matsuno's formula for $F(\lambda)$ in the case when $u_0(x)=2(1+x^2)^{-1}$, 
which gives some independent justification for its
validity.} of degree $N$.  The corresponding phase constants $\gamma_n$ 
all vanish
exactly.  The approximate eigenvalues determined from the initial
condition $u_0$ via the formula \eqref{value of eigenvalue given by a
  fomula} do not agree exactly with the scaled roots of the Laguerre
polynomial of degree $N$ (although the approximate phase constants
agree exactly with the true phase constants), so it is a worthwhile
exercise to compare the function $\ut(x,t)$ as specified by
Definition~\ref{def:modified} with the true solution $\ueps(x,t)$
of the Cauchy
problem for the BO equation with initial data $u_0(x)=2(1+x^2)^{-1}$.
Of course Corollary~\ref{corr:strong} guarantees strong convergence in
$L^2$ at $t=0$ (that is, $\ut(\cdot,0)$ is $L^2$-close to
$u_0(\cdot)$) but this alone does not guarantee that $\ut(x,t)$
approximates $\ueps(x,t)$ in any sense for
$t>0$.  We made the comparison for several values of $\epsilon>0$
corresponding to a reflectionless exact solution of the Cauchy problem
constructed\footnote{In fact this is
  the numerical method we used to create the plots in
  Figure~\ref{fig:numerics1}.  This has a tremendous advantage over
  taking a more traditional numerical approach to the Cauchy problem
  for the BO equation (that is, one involving time stepping) since the
  calculations necessary to find the solution for any two given values of 
$t$ are
completely independent, so errors do not propagate (and to find the
solution for any given time $t$ it is not necessary to perform any
calculations at all for intervening times from the initial instant).
The only source of error in the use of the determinantal formula
\eqref{eq:Hirotatransform}, at least if the differentiation is carried out
explicitly resulting in a sum of $N$ determinants, is due to
round-off.}
 from the determinantal formula
\eqref{eq:Hirotatransform}
at the time $t=4$, which is well beyond the breaking
time.  Our results are shown in Figure~\ref{fig:CauchyCompare}.
\begin{figure}[h]
\begin{center}
\includegraphics{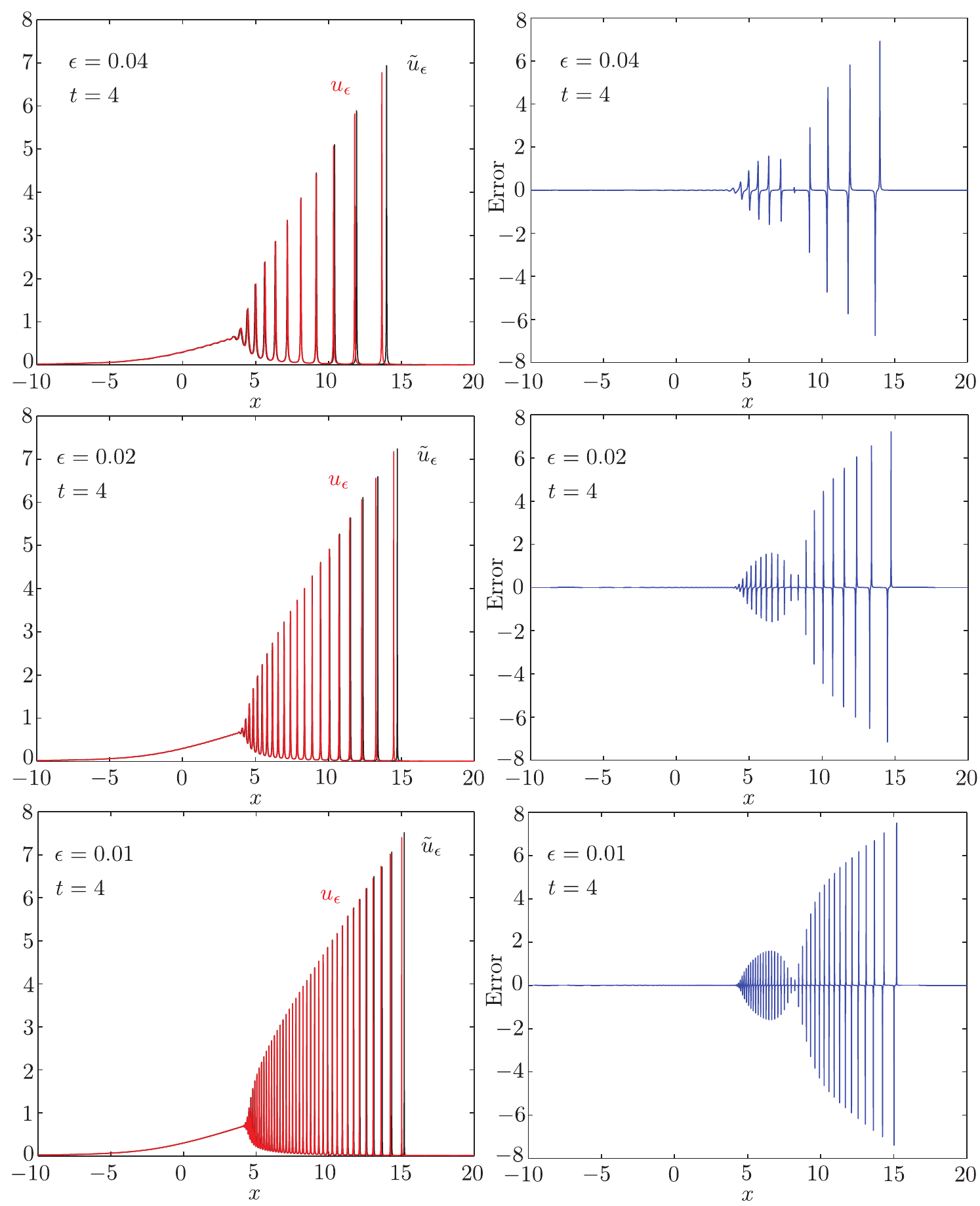}
\end{center}
\caption{\emph{Left:  plots of $\ut(x,t)$ (black) shown together with
$\ueps(x,t)$ (red) for the initial data
$u_0(x)=2(1+x^2)^{-1}$ shown for several values of $\epsilon$ at $t=4$.
Right:  The error $\ueps(x,t)-\ut(x,t)$.}}
\label{fig:CauchyCompare}
\end{figure}
These plots show that the modification of the scattering data used to
construct $\ut(x,t)$ results in a phase shift relative to $\ueps(x,t)$
that is proportional to $\epsilon$, the approximate wavelength of the 
oscillations in the undular
bore structure.  In particular, $\ut(x,t)$ does not remain close to
$\ueps(x,t)$ after the breaking time in any strong sense,
although is appears highly likely that convergence is restored in the weak 
topology.  

\section{Conclusion}
\label{sec:conclusion}
In this paper, we have obtained the first rigorous results regarding the 
zero-dispersion limit of the Cauchy problem for the BO
equation.  As suggested by the formal mulitphase 
averaging of modulated $N$-phase wavetrain solutions carried out by
Dobrokhotov and Krichever \cite{Dobrokhotov 1991}, the scalar inviscid 
Burgers' equation \emph{and its multivalued solution after wave breaking}
characterizes the limit.

To analyze the BO Cauchy problem, we used a remarkable formula of
Matsuno \cite{Matsuno 1981, Matsuno 1982} for the density of eigenvalues
of the nonlocal operator $\mathcal{L}$ appearing in the scattering theory,
and we have proposed a new asymptotic formula \eqref{formula for gamma 
in terms of x} for the corresponding phase constants necessary to set up 
the inverse-scattering problem.  We then developed an analogue of 
the Lax-Levermore method \cite{Lax 1983-1} to study the inverse-scattering
problem, and we obtained an explicit formula \eqref{limit of sx 2}
for a measure $\mu$ with density $G(\alpha;x,t)$ that is the BO equivalent
of the extremal measure in Lax-Levermore theory.  By contrast with the
KdV case, the formula we obtain for the weak limit from this measure
is remarkably simple and explicit, writing the weak limit as a signed sum
of branches of the multivalued solution of Burgers' equation.

A useful generalization of the weak limit given in this paper is to consider
the weak limits of all of the variational derivatives
\begin{equation}
K_n:=\frac{\delta I_n}{\delta u},\quad n=3,4,5,\dots,
\end{equation}
where the $I_n$ is the $n^\mathrm{th}$ conserved quantity given by 
\eqref{eq:integralsofmotion}.  Each of the quantities $K_n$ can be written
in the form
\begin{equation}
K_{n+2}=\frac{\partial}{\partial t_n}
\Ut(x,t_1,t_2,\dots,t_n)
\end{equation} 
where after the differentiation $t_1$ is set equal to $t$ and all $t_k$ for 
$k>1$ are set to zero.  Here $\Ut(x,t_1,t_2,\dots,t_n)$ is given by 
\eqref{small dispersion matrix calculation 3_1} 
with each occurrence of $2t$ in the expression $-2\lamt_k(x+2\lamt_kt)$
occurring in the matrix elements of $\tilde{\mat{A}}_\epsilon$ replaced by
$2t_1-3\lamt_k t_2+4\lamt_k^2 t_3 - 5\lamt_k^3 t_4 +\cdots +
(-1)^{n+1}(n+1)\lamt_k^{n-1}t_n$.  
Formulas for these higher weak limits can also be obtained
within the framework of our method and will be published in a subsequent
paper.

We are also currently investigating prospects for strengthening the limit
after wave breaking occurs.  The goal here is to rigorously establish
an asymptotic formula for $\ut(x,t)$ that is valid
\emph{pointwise} for $(x,t)$ in the oscillation zone. Such a formula
should accurately resolve the microscopic (wavelength proportional to
$\epsilon$) oscillations, including finding the phase up to error
terms that are bounded by a vanishingly small fraction of the
wavelength.  One expects the asymptotic form of the wavetrain to be
given by the rational-exponential formulae found by Dobrokhotov and
Krichever \cite{Dobrokhotov 1991}.  For the KdV equation such
pointwise asymptotics have been obtained \cite{Deift 1997} using the
Deift-Zhou steepest descent technique for matrix-valued
Riemann-Hilbert problems.  We are working to extend this kind of
methodology to the context of scalar Riemann-Hilbert problems with
\emph{nonlocal} jump conditions, as occurs in the inverse-scattering
transform (generally with nonvanishing reflection coefficient) for the
BO equation.

It has been recently conjectured by Dubrovin \cite{Dubrovin08} that
near the earliest breaking point $(x_\xi,t_\xi)$ the solution of the
Cauchy problem for quite general weakly dispersive Hamiltonian
perturbations of Burgers' equation should exhibit a universal form
expressed in terms of Painlev\'e transcendents.  This conjecture has
been confirmed for general initial data for the KdV equation (as a
particular case of a perturbation considered by Dubrovin) by Claeys
and Grava \cite{ClaeysG09}.  It would be interesting to determine by
direct calculation of the solution near the breaking point
$(x_\xi,t_\xi)$ whether the BO equation should be considered to fall
within the universality class of equations conjectured by Dubrovin.

\section{Acknowledgements}
This work was supported by the National Science Foundation under grant number
DMS-0807653.


\begin{thebibliography}{99}

\bibitem{Ablowitz 1983}
M.~J.~Ablowitz, A.~S.~Fokas and R.~L.~Anderson, ``The direct
linearizing transform and the Benjamin-Ono equation,'' \textit{Phys. Lett.
A} \textbf{93}, 375--378 (1983).

\bibitem{Benjamin 1967}
T.~B. Benjamin, ``Internal waves of permanent form in fluids of
great depth,'' \textit{J. Fluid Mech.} \textbf{29}, 559--592 (1967).

\bibitem{Bock 1979}
T.~L.~Bock and M.~D.~Kruskal, ``A two-parameter Miura transformation
of the Benjamin-Ono equation,'' \textit{Phys. Lett.} \textbf{74A}, 
173--176 (1979).

\bibitem{Choi 1996}
W.~Choi and R.~Camassa, ``Weakly nonlinear internal waves in a
two-fluid system,'' \textit{J. Fluid. Mech.} \textbf{313}, 83--103 (1996).

\bibitem{ClaeysG09}
T.~Claeys and T.~Grava, 
``Universality of the break-up profile for the KdV equation in the 
small dispersion limit using the Riemann-Hilbert approach,''
\textit{Comm. Math. Phys.} \textbf{286}, 979--1009 (2009). 


\bibitem{Coifman 1990}
R.~R.~Coifman and M.~V.~Wickerhauser, ``The scattering transform
for the Benjamin-Ono equation,'' \textit{Inverse Prob.} \textbf{6}, 825--861
(1990).

\bibitem{Davis 1967}
R.~E. Davis and A.~Acrivos, ``Solitary internal waves in deep
water,'' \textit{J. Fluid Mech.} \textbf{29}, 593--607 (1967).

\bibitem{Deift 1997}
P.~Deift, S.~Venakides and X.~Zhou, ``New results in small
dispersion KdV by an extension of the steepest descent method for
Riemann-Hilbert problems,'' \textit{Internat. Math. Research Not.} \textbf{6},
286-299 (1997).

\bibitem{Dobrokhotov 1991}
S.~Yu.~Dobrokhotov and I.~M.~Krichever, ``Multi-phase solutions of
the Benjamin-Ono equation and their averaging,'' \textit{Mat. Zametki} 
\textbf{49}, 42--58 (1991).

\bibitem{Dubrovin08}
B.~A.~Dubrovin, 
``On universality of critical behaviour in Hamiltonian PDEs,'' 
\textit{Geometry, topology, and mathematical physics,  
Amer. Math. Soc. Transl. Ser. 2}, \textbf{224}, 59--109, 
Amer. Math. Soc., Providence, RI (2008). 

\bibitem{DubrovinMN76}
B.~A.~Dubrovin, V.~B.~Matveev and S.~P.~Novikov, ``Non-linear equations
of Korteweg-de Vries type, finite-zone operators, and Abelian varieties,''
\textit{Russian Math. Surveys} \textbf{31}, 59--146 (1976).

\bibitem{FlaschkaFM80}
H.~Flaschka, M.~G.~Forest and D.~W.~McLaughlin, 
``Multiphase averaging and the inverse spectral solution of the 
Korteweg-de Vries equation,'' 
\textit{Comm. Pure Appl. Math.} \textbf{33}, 739--784 (1980).

\bibitem{Fokas 1983}
A.~S.~Fokas and M.~J.~Ablowitz, ``The inverse scattering transform
for the Benjamin-Ono equation: a pivot to multidimensional
problems,'' \textit{Stud. Appl. Math.} \textbf{68}, 1--10 (1983).

\bibitem{Gurevich 1974}
A.~V.~Gurevich and L.~P.~Pitaevskii, ``Nonstationary structure of a
collisionless shock wave,'' \textit{Soviet Phys. JETP} \textbf{38}, 
291--297 (1974).

\bibitem{ItsM75}
A.~R.~Its and V.~B.~Matveev, ``Hill operators with a finite number 
of lacunae,''
\textit{Functional Anal. Appl.} \textbf{9}, 65--66 (1975).

\bibitem{Jorge 1999}
M.~C.~Jorge, A.~A.~Minzoni and N.~F.~Smyth, ``Modulation solutions for
the Benjamin-Ono equation,'' \textit{Physica D} \textbf{132}, 1--18 (1999).

\bibitem{Kaup 1998}
D.~J.~Kaup and Y.~Matsuno, ``The inverse scattering transform for
the Benjamin-Ono equation,'' \textit{Stud. Appl. Math.} \textbf{101},
73--98 (1998).

\bibitem{KodamaAS82}
Y.~Kodama, M.~J.~Ablowitz and J.~Satsuma, ``Direct and inverse scattering
problems of the nonlinear intermediate long wave equation,''
\textit{J. Math. Phys.} \textbf{23}, 564--576 (1982).

\bibitem{Lax86}
P.~D.~Lax, ``On dispersive difference schemes,'' \textit{Physica}
\textbf{18D}, 250--254 (1986).

\bibitem{Lax 1983-1}
P.~D.~Lax and C.~D.~Levermore, ``The small dispersion limit of the
Korteweg-de Vries equation,'' \textit{Comm. Pure Appl. Math.} \textbf{36},
253--290, 571--593, and 809--929 (three parts) (1983).

\bibitem{Matsuno 1979}
Y.~Matsuno, ``Exact multi-soliton solution of the Benjamin-Ono
equation,'' \textit{J. Phys. A Math. Gen.} \textbf{12},
619--621 (1979).

\bibitem{Matsuno 1981}
Y.~Matsuno, ``Number density function of Benjamin-Ono solitons,''
\textit{Phys. Lett. A} \textbf{87}, 15--17 (1981).

\bibitem{Matsuno 1982}
Y.~Matsuno, ``Asymptotic properties of the Benjamin-Ono equation,''
\textit{J. Phys. Soc. Jpn.} \textbf{51}, 667--674 (1982).

\bibitem{Matsuno 1998-1}
Y.~Matsuno, ``Nonlinear modulation of periodic waves in the small
dispersion limit of the Benjamin-Ono equation,'' 
\textit{Phys. Rev. E} \textbf{58}, 7934--7940 (1998).

\bibitem{Matsuno 1998-2}
Y.~Matsuno, ``The small dispersion limit of the Benjamin-Ono
equation and the evolution of a step initial condition,'' 
\textit{J. Phys. Soc. Jpn.} \textbf{67}, 1814--1817 (1998).

\bibitem{Miller}
P.~D.~Miller, \textit{Applied Asymptotic Analysis}, Graduate Studies in
Mathematics \textbf{75}, American Mathematical Society, Providence, 2006.

\bibitem{Nakamura 1979}
A.~Nakamura, 
``B\"acklund transform and conservation laws of the Benjamin-Ono equation,'' 
\textit{J. Phys. Soc. Jpn.} \textbf{47}, 1335--1340 (1979).

\bibitem{Porter 2002}
A.~Porter and N.~F.~Smyth, ``Modelling the morning glory of the Gulf of
Carpentaria,'' \textit{J. Fluid Mech.} \textbf{454}, 1--20 (2002).

\bibitem{Ono 1975}
H.~Ono, ``Algebraic solitary waves in stratified fluids,'' \textit{J. Phys.
Soc. Jpn.} \textbf{39}, 1082--1091 (1975).

\bibitem{Venakides90}
S.~Venakides, 
``The Korteweg-de Vries equation with small dispersion: 
higher order Lax-Levermore theory,'' 
\textit{Comm. Pure Appl. Math.} \textbf{43}, 335--361 (1990).

\bibitem{Whitham65}
G.~B.~Whitham,
``Non-linear dispersive waves,'' 
\textit{Proc. Roy. Soc. Ser. A} \textbf{283}, 238--261 (1965). 

\bibitem{Wigner55}
E.~Wigner, ``Characteristic vectors of bordered matrices with infinite 
dimensions,'' \textit{Ann. of Math.} \textbf{62}, 548--564 (1955).

\bibitem{Wigner58}
E.~Wigner, ``On the distribution of the roots of certain symmetric matrices,'' 
\textit{Ann. of Math.} \textbf{67}, 325--328 (1958).

\end{thebibliography}
\end{document}